\documentclass[twocolumn]{autart}    

\usepackage{amsmath,amsfonts,bm,mathrsfs,amssymb,gensymb,multirow,color}
\usepackage{array}
\usepackage{textcomp}
\usepackage{stfloats}
\usepackage{url}
\usepackage{verbatim}
\usepackage{graphicx}
\usepackage{ntheorem}
\usepackage{algorithm}
\usepackage{etoolbox}
\usepackage{cite}
\usepackage{booktabs}
\usepackage[T1]{fontenc}
\usepackage[switch,pagewise]{lineno} 
\graphicspath{{Fig2/}}

\let\classAND\AND
\let\AND\relax
\usepackage{algorithmic}

\let\AND\classAND
\AtBeginEnvironment{algorithmic}{\let\AND\algoAND}

\newtheorem{theorem}{Theorem}
\newtheorem{lemma}{Lemma}
\newtheorem{definition}{Definition}

\makeatletter
\let\ps@copyright\ps@plain
\makeatother

\begin{document}
	
	\begin{frontmatter}
		\title{Extended Neighborhood Message Passing for Inference in Loopy Factor Graphs} 
		\thanks[footnoteinfo]{This work was supported in part by the National Natural Science Foundation of China under Grants 62473317, U21B2008, and 62371398. Corresponding author: Zengfu Wang.}

		\author[NPU]{Kuilong Yang}\ead{kuilongyang@mail.nwpu.edu.cn},    
		\author[NPU]{Zengfu Wang*}\ead{wangzengfu@nwpu.edu.cn},    
		\author[NPU]{Hua Lan}\ead{lanhua@nwpu.edu.cn}    

		\address[NPU]{School of Automation, Northwestern Polytechnical University,Xi'an 710072, China.}
		%
		

    \begin{abstract}
    Belief propagation (BP) is widely used for inference on probabilistic graphical models but often suffers from non-convergence and poor accuracy on loopy factor graphs, yielding inexact variance estimates in the Gaussian case. Generalized BP (GBP) improves accuracy but relies on manual design or application-specific heuristics for region-graph construction. We propose extended neighborhood message passing (ENMP), which automatically constructs extended neighborhoods of variable nodes from the bipartite factor graph topology, capturing loop-induced correlations without manual region design. The framework applies to general distributions on factor graphs with a tunable accuracy--complexity trade-off. In the exact regime, message updates converge to the unique fixed point of the ENMP equations in finitely many iterations, yielding exact marginals; in the approximate regime, ENMP provides approximate marginals. For Gaussian factor graphs, we derive closed-form update rules enabling exact variance estimation in the exact regime, addressing the long-standing inaccurate variance estimates of loopy Gaussian BP. We evaluate ENMP on densely loopy Gaussian factor graphs from power-grid benchmarks and on a discrete Ising factor graph. Results show accuracy improves with neighborhood size, with exact marginals achieved on several systems in both settings. On larger Gaussian systems, ENMP outperforms regularized Gaussian BP (sGaBP-m) in both mean and variance estimation accuracy.
    \end{abstract}

	\end{frontmatter}
	
	\section{Introduction}
	Probabilistic graphical models~(PGMs) provide a rigorous 
	mathematical framework for representing statistical 
	dependencies and performing inference under uncertainty 
	in high-dimensional 
	systems~\cite{koller2009probabilistic,Larranaga2011PGM}. 
	They have found broad applications in state 
	estimation~\cite{Pasqualetti2012,Marelli2015}, Bayesian 
	inference and learning~\cite{Zorzi2020,Berntorp2021}, 
	and system identification~\cite{Zhou2022Automatica,BONI2025107793}.
	Among PGMs, factor graphs provide a natural representation 
	of the factorization structure of a joint distribution by 
	separating variable nodes from factor nodes, and support 
	message-passing inference through the sum-product 
	algorithm~\cite{Kschischang2001Factor,Loeliger2004Introduction}. 
	Factor graphs have been successfully applied to a wide 
	range of inference problems, including robot 
	perception~\cite{Dellaert2017FactorGraphs}, cooperative 
	localization~\cite{Wymeersch2009}, and multitarget 
	tracking~\cite{LanMTT,Bai2024RT,QI2025}.
	In this paper, we consider marginal inference on loopy 
	factor graphs.

	Inference on factor graphs with loops presents significant challenges. Message-passing algorithms, also known as belief propagation~(BP) algorithms, effectively leverage local structural information~\cite{Kschischang2001Factor}. 
	While BP is theoretically guaranteed to be exact on tree-structured graphs, its performance on factor graphs with loops -- often referred to as loopy BP~(LBP) -- is less reliable. LBP is not guaranteed to converge and may converge to inaccurate fixed points~\cite{murphy1999loopy,weiss2000correctness,ihler2005loopy}. 
	Meanwhile, for Gaussian variables, BP presents a long-standing challenge: it may fail to converge, and even when it does converge, it often yields inexact variances~\cite{Weiss1999CorrectnessGBP,Su2015EstimationVariance}. 
	These limitations stem from BP treating incoming messages as independent, an assumption that breaks down in the presence of loops, particularly short loops that introduce strong correlations among neighboring nodes~\cite{Alec2021BPNL}.

	Substantial efforts have been devoted to mitigating loop effects in probabilistic
	graphical model inference.
	Classical clustering methods~\cite[Sec.~4.4.1]{Pearl1988} address loops by grouping
	variables into clusters, so that loop-induced correlations are handled locally within
	each cluster and message passing is performed between clusters.
	This broad principle is further developed in region-based inference frameworks, most
	notably generalized belief propagation~(GBP)~\cite{Yedidia2000GBP,Yedidia2005Constructing},
	where messages are passed between possibly overlapping regions in a region graph.
	For GBP, however, the marginal accuracy and computational cost depend on the selected
	region graph, while a generally optimal and fully automatic region-selection rule is still lacking, and effective region-graph construction typically relies on
	manual design or application-specific heuristics~\cite{ChoiceGBP}.
	This limitation motivates inference schemes in which enlarged neighborhoods are constructed automatically from the topology of the factor graph, rather than from a manually designed region graph.

	A different line of work, without manual region-graph construction, captures 
	short-range loop-induced correlations by enlarging the local neighborhood of each node.
	Kirkley \emph{et al.}~\cite{Alec2021BPNL} proposed a self-consistent message-passing 
	framework for Ising models that defines larger neighborhoods to capture short-range 
	correlations, improving accuracy over BP while eliminating the need for 
	manual region decomposition.
	Building on \cite{Alec2021BPNL}, subsequent work~\cite{TNBP} combines tensor networks 
	with alternative neighborhood definitions to further refine approximation accuracy.
	However, the methods in~\cite{Alec2021BPNL,TNBP} are developed for
	discrete pairwise MRFs, and extending them to general bipartite factor
	graphs is non-trivial due to topological and functional differences.
	Specifically, these approaches rely on aggregating messages over a larger
	Markov blanket, which is naturally defined in MRFs where variables are
	directly connected.
	In contrast, factor graphs exhibit a strict bipartite structure in which
	variable nodes never interact directly but are conditionally linked solely
	through intermediate factor nodes.
	Mapping these MRF-based neighborhoods to a factor graph
	does not reflect the alternating variable--factor structure and
	therefore does not directly yield a message-passing formulation for general
	bipartite factor graphs, necessitating a distinct formulation of neighborhood
	expansion and the associated message-passing rules.
	Moreover, in contrast to the discrete pairwise MRFs considered
	in~\cite{Alec2021BPNL,TNBP}, the factor graphs considered here may
	involve continuous variables.
	Although~\cite{Alec2021BPNL} establishes exactness under its bounded
	primitive-cycle condition, it does not establish that the message updates
	converge in a finite number of iterations, nor is this exactness result
	established for general bipartite factor graphs.

	A separate line of work focuses on Gaussian graphical models, where accurate variance estimation poses a long-standing challenge.
	To improve the accuracy and convergence of Gaussian BP~(GaBP) in loopy graphs, several variants of message-passing schemes have been developed, such as feedback message passing~\cite{Liu2012FMP}, and variance-improvement algorithms~\cite{Su2015EstimationVariance}.
	Notably, the regularized GaBP algorithm (sGaBP-m)~\cite{Kamper2019,Kamper2020}, designed for pairwise Markov graphs, employs a regularization parameter to guarantee convergence to the exact marginal means.
	While sGaBP-m yields exact means upon convergence with sufficient regularization, selecting the regularization parameter remains a heuristic process, and the estimated variances are still approximate without guarantees.
	Furthermore, the framework of sGaBP-m is tailored for pairwise Markov graphs and may not generalize directly to generic factor graphs~\cite{Kamper2020}.

	In summary, prior approaches face different limitations.
	Region-based methods such as GBP depend on the selected region graph, the
	effective construction of which remains problem-dependent and typically
	relies on manual design or application-specific heuristics;
	the neighborhood-expansion method in~\cite{Alec2021BPNL} is formulated for
	discrete pairwise MRFs and does not directly cover general bipartite factor
	graphs with continuous variables;
	and Gaussian-specific variants, such as GaBP and sGaBP-m, still cannot
	guarantee accurate variance estimates.
	These limitations motivate the present work, particularly for factor graphs
	with short loops.

	In response, we propose the extended neighborhood message
	passing~(ENMP) method for inference on loopy factor graphs. 
	Unlike region-based methods such as GBP, for which effective
	region-graph construction remains problem-dependent and typically
	relies on manual design or application-specific heuristics, ENMP
	automatically constructs an extended neighborhood directly from the
	bipartite factor-graph topology to capture loop-induced correlations.
	A single tunable parameter
	controls the size of this neighborhood, trading off inference
	accuracy against computational complexity. We formulate general
	message-passing update rules for factorized probability models and
	further specialize them to Gaussian factor graphs, deriving
	closed-form updates in canonical parameters. We analyze ENMP under
	two inference regimes. 
	In the \emph{exact regime}, where the extended neighborhoods are
	large enough to fully capture the loop-induced correlations in the
	factor graph, the ENMP message updates converge to a unique fixed
	point in a finite number of iterations and yield the exact marginals;
	such exactness and convergence within a finite number of iterations
	are not established for general bipartite factor graphs
	in~\cite{Alec2021BPNL}.
	For Gaussian factor graphs in particular, the exact regime yields
	not only exact marginal means but also exact marginal variances,
	addressing the well-known variance inaccuracy of loopy
	GaBP in this regime~\cite{Weiss1999CorrectnessGBP}.
	In the \emph{approximate regime}, ENMP provides approximate marginals.
	The main contributions of this work are summarized as follows:
	\begin{itemize}
		\item We propose ENMP, a message-passing framework 
		that automatically constructs extended neighborhoods from the bipartite 
		factor-graph topology to capture loop-induced correlations, without 
		requiring the manual design or application-specific heuristics typically 
		used for region-graph construction in GBP~\cite{Yedidia2000GBP, ChoiceGBP}. 
		A single tunable parameter controls the neighborhood size,
		providing a tradeoff between inference accuracy and computational complexity.
		The framework applies 
		to both discrete and continuous variables and admits closed-form 
		updates for Gaussian factor graphs. 
		
		\item
		We show that, in the exact regime, the ENMP message updates
		converge to a unique fixed point in finitely many iterations,
		yielding exact marginal beliefs; for Gaussian factor graphs,
		both the marginal means and variances are exact.
		In the approximate regime, ENMP provides approximate marginals.
		
		\item
		We evaluate ENMP on Gaussian factor graphs arising from power-system
		state estimation over standard IEEE bus test cases, and on a discrete
		Ising factor graph. 
		The results show that the inference accuracy of ENMP generally improves as
		the size of the extended neighborhoods increases, and that the method achieves
		exact marginal inference on several systems in both settings.
		On the larger Gaussian systems, ENMP also outperforms the tested
		sGaBP-m configurations in both mean and variance estimation accuracy,
		while the Ising experiments confirm its applicability to discrete
		factor graphs.
	\end{itemize}

	The remainder of this article is organized as follows.
	Section~\ref{sec:problem_formulation} formulates the inference problem
	on factor graphs. Section~\ref{Sec:ENMP} introduces the ENMP framework,
	including the automated construction of extended neighborhoods and the
	resulting message-passing update rules for general and Gaussian
	distributions. Section~\ref{sec:application} evaluates ENMP against
	the tested baselines on Gaussian power-system state estimation
	benchmarks and on discrete Ising factor graphs. Finally,
	Section~\ref{Sec:discussion} concludes the article.

	\section{Modeling and Problem Formulation}\label{sec:problem_formulation}
	Consider a joint probability distribution function~(PDF) or a probability mass function~(PMF) $P(\bm{x})$ of a random vector $\bm{x} = [x_1, x_2, \dots, x_{n_x}]^T$, where the variables can be either continuous or discrete. 
	Let $\mathbb{G} = (\mathcal{X}, \mathcal{F}, \mathcal{E})$ be the corresponding factor graph of $P(\bm{x})$. 
	Here, $\mathcal{X}=\{x_1,x_2,\dots,x_{n_x}\}$ and
	$\mathcal{F}=\{f_1,f_2,\dots,f_{n_f}\}$ denote the sets of variable
	nodes and factor nodes, respectively, where
	$n_x=|\mathcal X|$ and $n_f=|\mathcal F|$; $\mathcal E$ is the set
	of edges connecting variable nodes and factor nodes.

	Each factor node $f_a \in \mathcal{F}$ represents a local function that depends only on a subset of variables $\bm{x}_a \subseteq \mathcal{X}$, i.e., the variable nodes connected to $f_a$.
	Given these definitions, $P(\bm{x})$ can be decomposed as,
	\begin{equation}\label{Eq:P(x)}
		P(\bm{x}) =\frac{1}{Z} \psi(\bm{x})= \frac{1}{Z} {\textstyle \prod_{f_a \in \mathcal{F}} f_a\left ( \bm{x}_a    \right ) },      
	\end{equation}
	where $Z$ is the partition function.
	
	We aim to infer the marginal PDF of a specific variable $x_i$ ($i = 1, 2, \dots, n_x$), which can be written as,
	\begin{equation}\label{Eq:Marginal_primal}
		P(x_i) = \int P(\bm{x})d\bm{x}^{i-} = \int \frac{1}{Z} {\textstyle \prod_{f_a \in \mathcal{F}} f_a\left ( \bm{x}_a    \right )d\bm{x}^{i-} },
	\end{equation}
	where $\bm{x}^{i-}$ denotes the vector $\bm{x}$ excluding the component $x_i$.

	%
	
	Direct computation of the marginalization in
	\eqref{Eq:Marginal_primal} involves high-dimensional integration
	(or summation), which is generally computationally intractable for
	large-scale graphs. BP~\cite{Kschischang2001Factor} avoids this
	global marginalization by iteratively passing local messages
	$m_{i\rightarrow a}(x_i)$ and $m_{a\rightarrow i}(x_i)$ between
	variable node $x_i$ and factor node $f_a$ along each edge. The BP
	updates and the resulting belief $b_i(x_i)$ of $x_i$ are given by
	\begin{equation*}
		\begin{aligned}
			m_{i\rightarrow a}(x_i)
			&\propto \textstyle
			\prod_{f_b\in\mathcal N_{x_i}\setminus\{f_a\}}
			m_{b\rightarrow i}(x_i),\\
			m_{a\rightarrow i}(x_i)
			&\propto
			\int
			f_a(\bm{x}_a) \textstyle 
			\prod_{x_j\in\mathcal N_{f_a}\setminus\{x_i\}}
			m_{j\rightarrow a}(x_j)
			\,d\bm{x}_{a\setminus i},\\
			b_i(x_i)
			&\propto \textstyle
			\prod_{f_a\in\mathcal N_{x_i}}
			m_{a\rightarrow i}(x_i),
		\end{aligned}
	\end{equation*}
	where $\mathcal N_v$ denotes the set of immediate neighbors of node $v$,
	and $\bm{x}_{a\setminus i}$ denotes $\bm{x}_a$ excluding $x_i$.

	BP performs exact inference on trees but struggles in loopy graphs because it treats incoming messages as independent. This assumption breaks down in the presence of short loops, which introduce strong correlations and cause message feedback. Consequently, BP often fails to distinguish new information from recirculated beliefs, leading to oscillation and inaccurate variance estimates~\cite{Su2015EstimationVariance,Alec2021BPNL}.
	To overcome these challenges and achieve accurate inference in such densely loopy factor graphs, we propose the ENMP method in the following section.

	\section{The ENMP Method}\label{Sec:ENMP}
	
	In this section, we propose the ENMP method for inference in loopy factor graphs and establish its update rules for both general and Gaussian distributions.

	\subsection{Extended neighborhood modeling}\label{sec:extendneigh}
	
	Inspired by~\cite{Alec2021BPNL}, we propose an extended neighborhood for each variable node to account for loop-induced correlations among larger groups of neighboring nodes.
	Different from~\cite{Alec2021BPNL}, which defines larger neighborhoods 
	over MRFs where variables interact directly, our formulation is tailored 
	to the bipartite topology of factor graphs and applies to both discrete 
	and continuous variable domains.
	Specifically, if a variable node $x_i$ and its neighborhood nodes are contained in a loop of length less than a given value, then all other nodes within this loop are considered as extended neighbors.
	
	To formalize these concepts, we introduce the following definitions
	of walks, paths, and loops.
	For two nodes $u$ and $v$ in $\mathbb G$, a walk from $u$ to $v$
	is a node sequence
	$w=(v_0,v_1,\ldots,v_{\ell})$
	such that $v_0=u$, $v_{\ell}=v$, and $v_{k-1}$ and $v_k$
	are adjacent in $\mathbb G$ for every $k=1,2,\ldots,\ell$.
	The length of $w$ is defined as $\ell(w)=\ell$.
	A path from $u$ to $v$ is a walk
	$p=(v_0,v_1,\ldots,v_{\ell})$
	in which the nodes $v_0,v_1,\ldots,v_{\ell}$ are pairwise distinct.
	Let $\operatorname{Path}(u,v)$ denote the set of all such paths
	from $u$ to $v$ in $\mathbb G$.
	More generally, for a subgraph $\mathbb H$ of $\mathbb G$ and
	two nodes $u$ and $v$ in $\mathbb H$,
	$\operatorname{Path}_{\mathbb H}(u,v)$ denotes the set of all paths
	from $u$ to $v$ in $\mathbb H$.
	For a path $p=(v_0,v_1,\ldots,v_{\ell})$, its length is
	$\ell(p)=\ell$.
	A loop in $\mathbb G$ is a closed walk
	$C=(v_0,v_1,\ldots,v_\ell)$ with $\ell\ge3$ such that
	$v_\ell=v_0$ and the nodes
	$v_0,v_1,\ldots,v_{\ell-1}$ are pairwise distinct.
	Its length is $|C|=\ell$.
	
	For any node set
	$\mathcal S\subseteq\mathcal X\cup\mathcal F$, let
	$\mathbb G[\mathcal S]$ denote the subgraph of $\mathbb G$
	that retains the node set $\mathcal S$ together with all edges of
	$\mathbb G$ connecting any two nodes in $\mathcal S$.
	For a subgraph $\mathbb H$, a walk $w$, a path $p$, and a loop $C$,
	let $V(\mathbb H)$, $V(w)$, $V(p)$, and $V(C)$ denote their
	corresponding node sets.
	The length of the longest loop in $\mathbb G$ is denoted by
	$\mathcal L_{\max}$.
	
	To formulate ENMP on a bipartite factor graph, we introduce three
	types of neighborhoods: the preliminary $r$-neighborhood, the extended
	$r$-neighborhood, and the residual $r$-neighborhood. Their formal
	definitions are given below.
	\vspace{0.8\baselineskip}
	\begin{definition}\label{definition1}
		\textbf{[Preliminary $r$-neighborhood]}
		The preliminary $r$-neighborhood, $M_i^{(r)}$, of a variable node
		$x_i\in\mathcal X$ is defined as the subgraph of $\mathbb G$ induced
		by $x_i$, its immediate neighbors, and all paths of length at most
		$r$ between any two distinct neighbors of $x_i$ that do not pass
		through $x_i$. Specifically,
		\begin{equation*}
			\begin{aligned}
				\mathcal P_i^{(r)}
				&=\bigl\{p\mid
				p\in\operatorname{Path}_{
					\mathbb G[(\mathcal X\cup\mathcal F)\setminus\{x_i\}]}(u,v),\\
				&\hphantom{={}\bigl\{p\mid{}}
				u,v\in\mathcal N_{x_i},\;
				u\neq v,\;
				\ell(p)\le r
				\bigr\},
				\\
				M_i^{(r)}
				&=\mathbb G\!\left[
				\{x_i\}\cup\mathcal N_{x_i}
				\cup{\textstyle\bigcup_{p\in\mathcal P_i^{(r)}}}V(p)
				\right].
			\end{aligned}
		\end{equation*}
		Similarly, the preliminary $r$-neighborhood of a factor node
		$f_a\in\mathcal F$, denoted by $M_{f_a}^{(r)}$, is defined as
		\begin{equation*}
			\begin{aligned}
				\mathcal P_{f_a}^{(r)}
				&=\bigl\{p\mid
				p\in\operatorname{Path}_{
					\mathbb G[(\mathcal X\cup\mathcal F)\setminus\{f_a\}]}(u,v),\\
				&\hphantom{={}\bigl\{p\mid{}}
				u,v\in\mathcal N_{f_a},\;
				u\neq v,\;
				\ell(p)\le r
				\bigr\},\\
				M_{f_a}^{(r)}
				&=\mathbb G\!\left[
				\{f_a\}\cup\mathcal N_{f_a}
				\cup{\textstyle\bigcup_{p\in\mathcal P_{f_a}^{(r)}}}V(p)
				\right].
			\end{aligned}
		\end{equation*}
	\end{definition}
	\begin{definition}\label{definition3}
		\textbf{[Extended $r$-neighborhood]}
		The extended $r$-neighborhood, $N_i^{(r)}$, of a variable node
		$x_i$ is constructed by recursively incorporating the preliminary
		$r$-neighborhoods of the included factor nodes.
		For a node set
		$\mathcal S\subseteq\mathcal X\cup\mathcal F$, define
		\[
		T(\mathcal S)
		= \textstyle
		\mathcal S
		\cup
		\bigcup_{f_a\in\mathcal S\cap\mathcal F}
		V(M_{f_a}^{(r)}).
		\]
		Consider the sequence initialized with
		$\mathcal S_i^{(0)}=V(M_i^{(r)})$ and updated via
		$\mathcal S_i^{(k+1)}=T(\mathcal S_i^{(k)})$.
		Since $\mathbb G$ is finite, there exists
		$K\le|\mathcal F|$ such that
		$\mathcal S_i^{(K+1)}=\mathcal S_i^{(K)}$.
		The extended $r$-neighborhood is defined as
		$N_i^{(r)}=\mathbb G[\mathcal S_i^{(K)}]$.
	\end{definition}
	\vspace{0.8\baselineskip}
	\begin{definition}\label{definition4}
		\textbf{[Residual $r$-neighborhood]}
		Given the extended $r$-neighborhoods $N_i^{(r)}$ and
		$N_j^{(r)}$ of two distinct variable nodes $x_i$ and $x_j$
		with $x_j\in V(N_i^{(r)})$,
		let
		$\Delta_{j\setminus i}
		=\{x_j\}\cup
		\bigl(V(N_j^{(r)})\setminus V(N_i^{(r)})\bigr)$.
		To include the variable nodes connected to each factor node in
		$\Delta_{j\setminus i}$, the residual $r$-neighborhood is defined as
		$N_{j\setminus i}^{(r)}
		=\mathbb G\!\left[\Delta_{j\setminus i}\cup
		\bigcup_{f_a\in\Delta_{j\setminus i}\cap\mathcal F}
		\mathcal N_{f_a}\right]$.
	\end{definition}

	Having defined the preliminary, extended, and residual
	$r$-neighborhoods, Algorithm~\ref{alg1} provides an explicit
	implementation of the recursive construction of $N_i^{(r)}$
	in Definition~\ref{definition3}.
	\begin{algorithm}
		\caption{Generation of the \emph{extended $r$-neighborhood}
			$N_i^{(r)}$.}\label{alg1}
		\begin{algorithmic}[1]
			\REQUIRE The factor graph
			$\mathbb G=(\mathcal X,\mathcal F,\mathcal E)$, the variable node
			$x_i\in\mathcal X$, and the value of $r$;
			\ENSURE The \emph{extended $r$-neighborhood} $N_i^{(r)}$;
			\STATE Generate the \emph{preliminary $r$-neighborhood}
			$M_i^{(r)}$ by Definition~\ref{definition1};
			
			\STATE Let $\mathcal S_i=V(M_i^{(r)})$,
			$\mathcal S_i^{\mathrm{old}}=\emptyset$, and
			$\mathcal F_{\mathrm{visited}}=\emptyset$;
			\WHILE{$\mathcal S_i\neq\mathcal S_i^{\mathrm{old}}$}
			\STATE $\mathcal S_i^{\mathrm{old}}=\mathcal S_i$,
			$E=\emptyset$, and
			$\mathcal F_{\mathrm{current}}=\mathcal S_i\cap\mathcal F$;
			\FOR{each factor
				$f_a\in\mathcal F_{\mathrm{current}}
				\setminus\mathcal F_{\mathrm{visited}}$}
			\STATE Generate the \emph{preliminary $r$-neighborhood}
			$M_{f_a}^{(r)}$;
			\IF{$V(M_{f_a}^{(r)})\not\subseteq\mathcal S_i$}
			\STATE $E=E\cup
			\bigl(V(M_{f_a}^{(r)})\setminus\mathcal S_i\bigr)$;
			\ENDIF
			\ENDFOR
			\STATE $\mathcal F_{\mathrm{visited}}
			=\mathcal F_{\mathrm{visited}}
			\cup\mathcal F_{\mathrm{current}}$;
			\STATE $\mathcal S_i=\mathcal S_i\cup E$;
			\ENDWHILE
			\STATE Let $N_i^{(r)}=\mathbb G[\mathcal S_i]$;
			
			\STATE \textbf{Return} $N_i^{(r)}$.
		\end{algorithmic}
	\end{algorithm}
	
	Algorithm~\ref{alg1} realizes the recursive sequence
	$\mathcal S_i^{(k+1)}=T(\mathcal S_i^{(k)})$ in
	Definition~\ref{definition3}. At each iteration, the preliminary
	$r$-neighborhoods of the unvisited factor nodes in the current
	node set are incorporated through
	$E=\bigcup_{f_a}
	\bigl(V(M_{f_a}^{(r)})\setminus\mathcal S_i\bigr)$,
	and the node set is updated accordingly.
	Since each factor node is expanded at most once, the number of
	factor-node expansion operations is bounded by $|\mathcal F|$.
	Therefore, Algorithm~\ref{alg1} terminates regardless of the longest
	loop length $\mathcal L_{\max}$ and returns exactly the extended
	$r$-neighborhood $N_i^{(r)}$ defined in
	Definition~\ref{definition3}.

	Because of the bipartite structure of a factor graph, every path in
	$\mathbb G[(\mathcal X\cup\mathcal F)\setminus\{v\}]$ between two
	distinct nodes in $\mathcal N_v$ has an even length, i.e.,
	$2,4,6,\ldots$. Consequently, only even values of $r$ can enlarge the
	preliminary $r$-neighborhood, while $r=0$ includes no such path.
	Fig.~\ref{fig:neighborhoods}(a) illustrates the preliminary
	$6$-neighborhood, while Fig.~\ref{fig:neighborhoods}(b) illustrates
	the corresponding extended and residual $6$-neighborhoods on the same
	factor graph, with $x_1$ as the focal variable node.

	\begin{figure}[t]
		\centering
		\begin{minipage}{\columnwidth}
			\centering
			\includegraphics[width=0.92\columnwidth]{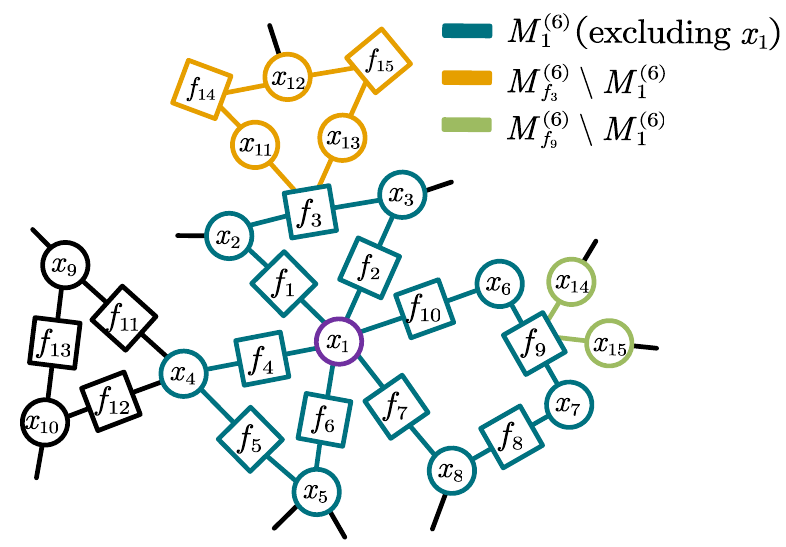}
			{\small (a) Preliminary $6$-neighborhood $M_1^{(6)}$ of variable node
				$x_1$ and the additional node sets
				$V(M_{f_3}^{(6)})\setminus V(M_1^{(6)})$ and
				$V(M_{f_9}^{(6)})\setminus V(M_1^{(6)})$.}
		\end{minipage}
		\begin{minipage}{\columnwidth}
			\centering
			\includegraphics[width=0.92\columnwidth]{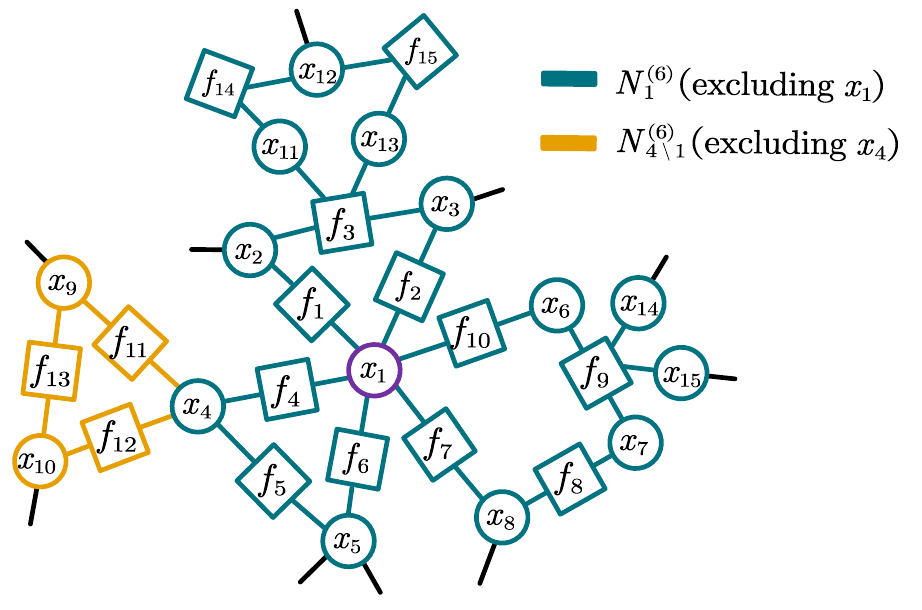}
			{\small (b) Extended $6$-neighborhood $N_1^{(6)}$ of variable node
				$x_1$ and residual $6$-neighborhood $N_{4\setminus1}^{(6)}$.}
		\end{minipage}
		\caption{Illustration of the preliminary, extended, and residual
			$r$-neighborhoods on the same factor graph for $r=6$.
			The focal variable node $x_1$ is shown in purple, and the highlighted
			node sets are identified by the legends in the respective panels.}
		\label{fig:neighborhoods}
	\end{figure}

	For the preliminary $r$-neighborhood shown in
	Fig.~\ref{fig:neighborhoods}(a),
	$V(M_{1}^{(0)})=\{x_1,f_1,f_2,f_4,f_6,f_7,f_{10}\}$, and
	$V(M_{1}^{(2)})=V(M_{1}^{(0)})$ since the minimum loop length is six.
	For $r=4$,
	$V(M_{1}^{(4)}) \!=\! \{x_1,x_2,x_3,x_4,x_5,
	f_1,f_2,f_3,f_4,f_5,f_6,f_7,f_{10}\}$.
	For $r=6$,
	$V(M_1^{(6)})=\{x_1,x_2,x_3,x_4,x_5,x_6,x_7,x_8,
	f_1,\allowbreak f_2,\allowbreak f_3,\allowbreak f_4,\allowbreak f_5,f_6,f_7,f_8,f_9,f_{10}\}$.
	
	For the extended $r$-neighborhood shown in
	Fig.~\ref{fig:neighborhoods}(b),
	$V(N_1^{(0)}) \!=\!\! \{x_1,x_2,x_3,x_4,x_5,x_6,x_8,
	f_1,f_2,f_4,f_6,f_7,f_{10}\}$, and
	$V(N_1^{(2)})=V(N_1^{(0)})$.
	For $r=4$,
	$V(N_1^{(4)})=\{x_1,\allowbreak x_2,\allowbreak x_3,\allowbreak x_4,x_5,x_6,x_8,
	x_{11},x_{12},x_{13},
	f_1,f_2,f_3,f_4,f_5,\allowbreak f_6, f_7,\allowbreak f_{10},f_{14},f_{15}\}$.
	For $r=6$,
	$V(N_1^{(6)})=
	\{x_1,x_2, x_3, x_4, x_5,\allowbreak x_6,x_7,x_8,
	x_{11},x_{12},x_{13},x_{14},x_{15},
	f_1,f_2,f_3, f_4, f_5, \! f_6,\! f_7, \! f_8,\allowbreak f_9,f_{10},f_{14},f_{15}\}$.
	In particular, the orange and green node sets in
	Fig.~\ref{fig:neighborhoods}(a),
	$V(M_{f_3}^{(6)})\setminus V(M_1^{(6)})$ and
	$V(M_{f_9}^{(6)})\setminus V(M_1^{(6)})$, respectively, are
	incorporated through the recursive construction of the extended
	$r$-neighborhood.
	
	Fig.~\ref{fig:neighborhoods}(b) also illustrates the residual
	$r$-neighborhood. For $r=0$,
	$V(N_4^{(0)})=
	\{x_1,x_4,x_5,x_9,x_{10},f_4,f_5,f_{11},f_{12}\}$.
	Together with
	$V(N_1^{(0)})=
	\{x_1,x_2,x_3,x_4,x_5,x_6,x_8,
	f_1,\allowbreak f_2,\allowbreak f_4,\allowbreak f_6,f_7,f_{10}\}$,
	this gives
	$V(N_{4\setminus1}^{(0)})=
	\{x_4,x_5,x_9,\allowbreak x_{10},\allowbreak f_5,f_{11},f_{12}\}$.
	Note that $x_5$ is included in $V(N_{4\setminus1}^{(0)})$ since it
	is connected to the factor node $f_5$.
	For $r=6$,
	$V(N_{4\setminus1}^{(6)})=
	\{x_4,x_9,x_{10},f_{11},f_{12},f_{13}\}$.

	\noindent\textbf{Remark~1:}
	The dominant computational cost of Algorithm~\ref{alg1} arises from
	enumerating paths of length at most $r$ between distinct neighbors of
	each node (Definition~\ref{definition1}), of which computational complexity is
	$\mathcal{O}\!\left((\max_{u\in\mathcal X\cup\mathcal F}
	|\mathcal N_u|)^r\right)$.
	In practice, 
	to reduce the computational cost, we can employ two exactness-preserving pruning rules.
	First, for a focal node $v$, any path joining two neighbors of $v$
	without passing through $v$ forms a loop through $v$ with the two
	incident edges; the search is thus confined to the corresponding
	biconnected component (BCC), computed by the standard depth-first
	lowpoint algorithm~\cite{Tarjan1972}, discarding tree-like branches
	that do not belong to such loops.
	Second, a multi-source breadth-first search from $\mathcal N_v$ on
	$\mathbb G\setminus\{v\}$ gives the shortest-hop distance
	$\operatorname{dist}_v(w)$ from $w$ to its nearest node in
	$\mathcal N_v$; when a partial path of length $\ell$
	is extended by one edge to a candidate node $w$, the extension is
	pruned once $\ell+1+\operatorname{dist}_v(w)>r$, since any continuation
	through $w$ exceeds the length bound $r$ before reaching another
	neighbor. Since $\operatorname{dist}_v(\cdot)$ is
	a lower bound on the remaining path length~\cite{Hart1968},
	this rule removes only infeasible branches.
	The two pruning rules together leave the returned neighborhood unchanged
	and substantially reduce the computational cost.
	
	\subsection{General message passing rules}\label{sec:generalMess}
	Given the definitions of the \emph{preliminary $r$-neighborhood},
	\emph{extended $r$-neighborhood}, and \emph{residual $r$-neighborhood},
	we now derive the general ENMP update rules.
	For notational brevity, we omit the superscript $r$ in the remainder of
	this article (e.g., using $N_{j \setminus i}$ to refer to
	$N_{j \setminus i}^{(r)}$).
	For a neighborhood $N$, either an extended neighborhood $N_i$
	or a residual neighborhood $N_{j\setminus i}$, let
	$\bm x_N=V(N)\cap\mathcal X$ and
	$\bm f_N=V(N)\cap\mathcal F$ denote its variable-node set and
	factor-node set, respectively. For $x_k\in\bm x_N$, let
	$\bm x_{N\setminus k}=\bm x_N\setminus\{x_k\}$.
	When the variables in $\bm x_N$ are ordered,
	$\bm x_N$ denotes the resulting stacked vector.

	Consider the message-passing process rooted at a variable
	node~$x_i$. The key idea of ENMP is to separate the computation
	of the global potential function $\psi(\bm x)$
	in~\eqref{Eq:P(x)} into local computations over the extended and
	residual neighborhoods. Specifically, let
	$\mathcal T_i^{(\ell)}$ denote the set of neighborhoods in the
	$\ell$th layer of the recursive expansion starting from the
	extended neighborhood $N_i$ of $x_i$. These layers are defined by
	\begin{equation}\label{Eq:recursive_layers}
		\begin{aligned}
			\!\! \mathcal T_i^{(1)}
			&=\{N_i\},
			\\
			\!\! \mathcal T_i^{(2)}
			&=
			\left\{
			N_{j\setminus i}
			\,\middle|\,
			x_j\in\bm x_{N_i\setminus i}
			\right\},
			\\
			\!\! \mathcal T_i^{(\ell+1)}
			&=\textstyle
			\bigcup_{N_{j\setminus h}\in\mathcal T_i^{(\ell)}}
			\left\{
			N_{k\setminus j}
			\,\middle|\,
			x_k\in
			\bm x_{N_{j\setminus h}\setminus j}
			\right\},
			\, \ell\ge2.
		\end{aligned}
	\end{equation}
	The sets $\{\mathcal T_i^{(\ell)}\}_{\ell\ge1}$ define a
	layered recursive expansion starting from $N_i$. If
	$\mathcal T_i^{(K+1)}=\emptyset$ for some finite $K$, then
	the expansion terminates after $K$ layers.

	Based on this layered expansion, we first derive an exact
	hierarchical decomposition of $\psi(\bm x)$ and the corresponding
	message-passing equations for
	$r\ge\mathcal L_{\max}-2$. We then discuss how the same equations
	can be used as an approximation when
	$r<\mathcal L_{\max}-2$.
	
	\subsubsection*{1) Exact regime}\label{sec:exact}
	When $r\ge\mathcal L_{\max}-2$, every loop in the factor graph
	$\mathbb G$ has length at most $r+2$. Under this condition, the
	recursive expansion starting from $N_i$ terminates after a finite
	number $K$ of layers, with
	$\mathcal T_i^{(K+1)}=\emptyset$. The global potential function
	$\psi(\bm x)$ then admits the hierarchical decomposition
	\begin{equation}\label{Eq:P_S_1}
		\psi(\bm x)
		=
		\prod_{\ell=1}^{K}
		\prod_{N\in\mathcal T_i^{(\ell)}}
		\Phi_N(\bm x_N),
	\end{equation}
	where $N=N_i$ in the first layer and $N$ is a residual
	neighborhood in every subsequent layer. The corresponding local
	potential function is defined by
	\begin{equation}\label{Eq:fai_1}
		\Phi_N(\bm{x}_N)
		=
		\prod_{f_a\in\bm{f}_N} f_a(\bm{x}_a).
	\end{equation}
	If a neighborhood $N$ contains no factor node, then
	$\bm f_N=\emptyset$, and~\eqref{Eq:fai_1} gives
	$\Phi_N(\bm x_N)=1$. Such a case arises when
	$V(N_j)\subseteq V(N_i)$: then
	$V(N_{j\setminus i})=\{x_j\}$ and
	$\Phi_{N_{j\setminus i}}
	(\bm x_{N_{j\setminus i}})=1$.
	For example, in Fig.~\ref{fig:neighborhoods}(b), when $r\geq4$,
	$V(N_{11})\subseteq V(N_1)$ and
	$V(N_{11\setminus1})=\{x_{11}\}$, so
	$\Phi_{N_{11\setminus1}}
	(\bm x_{N_{11\setminus1}})=1$.

	The following lemma establishes the separation and termination
	properties of the recursive expansion and guarantees the
	exactness of~\eqref{Eq:P_S_1}.

	\vspace{0.8\baselineskip}
	
	\begin{lemma}\label{lemma:exact_separation}
		Consider a connected factor graph $\mathbb G$
		containing at least one loop, and let
		$r\ge\mathcal L_{\max}-2$. For every variable node $x_i$, the
		following properties hold.
		
		\textit{(P1)}\;
		For every $x_j\in\bm x_{N_i\setminus i}$,
		$V(N_{j\setminus i})\cap V(N_i)=\{x_j\}$.
		
		\textit{(P2)}\;
		For every $x_j\in\bm x_{N_i\setminus i}$, let
		$\bm D_{j\setminus i}$ denote the union of the node sets of
		$N_{j\setminus i}$ and all residual
		$r$-neighborhoods generated recursively from it by~\eqref{Eq:recursive_layers}.
		Equivalently,
		\begin{equation}\label{Eq:Dji_def}
			\bm D_{j\setminus i}
			=\textstyle
			V(N_{j\setminus i})
			\cup
			\bigcup_{x_k\in\bm x_{N_{j\setminus i}\setminus j}}
			\bm D_{k\setminus j}.
		\end{equation}
		Then $\bm D_{j\setminus i}$ induces a connected subgraph of
		$\mathbb G$,
		$\bm D_{j\setminus i}\cap V(N_i)=\{x_j\}$, and, for any
		two distinct variables
		$x_{j_1},x_{j_2}\in\bm x_{N_i\setminus i}$,
		$\bm D_{j_1\setminus i}\cap\bm D_{j_2\setminus i}
		=\emptyset$.
		
		\textit{(P3)}\;
		For any variables $x_{j_1},x_{j_2},\ldots,x_{j_n}$, with
		$j_0=i$, satisfying
		\[
		x_{j_1}\in\bm x_{N_i\setminus i},
		\qquad
		x_{j_{\ell+1}}
		\in
		\bm x_{N_{j_\ell\setminus j_{\ell-1}}\setminus j_\ell},
		\quad \ell=1,\ldots,n-1,
		\]
		the variables $x_i,x_{j_1},x_{j_2},\ldots,x_{j_n}$ are pairwise
		distinct. Consequently, the recursive expansion
		in~\eqref{Eq:recursive_layers} is acyclic and satisfies
		$\mathcal T_i^{(K+1)}=\emptyset$ for some
		$K\le|\mathcal X|$.
		
		\textit{(P4)}\;
		The factor-node sets $\bm f_N$,
		$N\in\bigcup_{\ell=1}^{K}\mathcal T_i^{(\ell)}$,
		form a partition of $\mathcal F$. Consequently, the
		hierarchical decomposition~\eqref{Eq:P_S_1} holds with
		equality.
	\end{lemma}
	
	The proof of Lemma~\ref{lemma:exact_separation}, together with an auxiliary result, is provided in Appendix~\ref{appendix:exact_separation}.

	Since the decomposition~\eqref{Eq:P_S_1} holds with equality
	under $r\geq\mathcal L_{\max}-2$, the marginal PDF of $x_i$
	can be expressed as
	\begin{equation}\label{Eq:q_s_marginal_int}
		q_i(x_i)
		\propto
		\int
		\prod_{\ell=1}^{K}
		\prod_{N\in\mathcal T_i^{(\ell)}}
		\Phi_N(\bm x_N)d\bm x^{i-}.
	\end{equation}
	To evaluate~\eqref{Eq:q_s_marginal_int} recursively, for each
	$x_j\in\bm x_{N_i\setminus i}$, define
	\begin{equation}\label{Eq:branch_potential}
		\Psi_{j\setminus i}
		=
		\prod_{f_a\in
			\bm D_{j\setminus i}\cap\mathcal F}
		f_a(\bm x_a).
	\end{equation}
	Since $\mathcal T_i^{(1)}={N_i}$ and
	$\bm D_{j\setminus i}$ contains $N_{j\setminus i}$ and all
	residual neighborhoods generated recursively from it, (P2)
	and (P4) of Lemma~\ref{lemma:exact_separation} give
	\begin{equation}\label{Eq:branch_decomposition}
		\psi(\bm x)
		=
		\Phi_{N_i}(\bm x_{N_i})
		\prod_{x_j\in\bm x_{N_i\setminus i}}
		\Psi_{j\setminus i}.
	\end{equation}
	Moreover, by~\eqref{Eq:Dji_def} and (P2) of Lemma~\ref{lemma:exact_separation},
	\begin{equation}\label{Eq:branch_recursive_decomposition}
		\Psi_{j\setminus i}
		=
		\Phi_{N_{j\setminus i}}
		(\bm x_{N_{j\setminus i}})
		\prod_{x_k\in
			\bm x_{N_{j\setminus i}\setminus j}}
		\Psi_{k\setminus j}.
	\end{equation}
	Accordingly, let $q_{i\leftarrow j}(x_j)$ denote the result of
	marginalizing $\Psi_{j\setminus i}$ over all variables contained
	in $\bm D_{j\setminus i}$ except $x_j$. Using
	\eqref{Eq:branch_decomposition} and
	\eqref{Eq:branch_recursive_decomposition},
	\eqref{Eq:q_s_marginal_int} can then be written recursively as
	\begin{equation}\label{Eq:q_i_marginal_mess1_int}
		\begin{split}
			q_i(x_i) \! = \! \frac{1}{Z_i}  \!\!\int\!\! \Phi_{N_{i}} (\bm{x}_{N_{i}})  {\textstyle \prod_{x_j\in \bm{x}_{N_{i}\setminus{i}}}  q_{i\leftarrow j}(x_j)d\bm{x}_{N_{i}\setminus{i}}   }
		\end{split}
	\end{equation}
	with
	\begin{equation}\label{Eq:q_i_j_marginal_mess1_int}
		\begin{split}
			\!\!\!q_{i\leftarrow j}(x_j) \!\! =\! \!  \int \!\! \Phi_{N_{j\setminus {i}}} ( \bm{x}_{N_{j\setminus {i}}})
			\!\!\!\!\!\!\!\!\prod_{x_k\in \bm{x}_{N_{j\setminus {i}}\setminus{j}}} \!\!\!\!\!\!\! \! q_{j\leftarrow k}(x_{k}) d\bm{x}_{N_{j\setminus {i}}\setminus{j}},
		\end{split}
	\end{equation}
	where $1/Z_i$ is the normalization constant.
	

	Equations~\eqref{Eq:q_i_marginal_mess1_int} 
	and~\eqref{Eq:q_i_j_marginal_mess1_int} establish 
	the general update rules of ENMP, and the ENMP 
	algorithm is summarized in Algorithm~\ref{alg2}, 
	where $q_{i \leftarrow j}^t(x_j)$ denotes the 
	message from $x_j$ to $x_i$ at the $t$th iteration.

	\begin{algorithm}
		\caption{The ENMP algorithm for inferring $x_i$.}\label{alg2}
		\begin{algorithmic}[1]
			\REQUIRE The factor graph $\mathbb{G}$,
			maximum iterations $\beta$,
			convergence criterion $\delta$,
			and the value of $r$;
			\ENSURE Beliefs $q(x_i)$, $i=1,2,\ldots,n_x$;
			\STATE Generate \emph{extended $r$-neighborhood} $N_i$ via
			Algorithm~\ref{alg1} for all $x_i$;
			\STATE Initialize $t = 1$, $e = \infty$;
			\FORALL{$x_i$, $i=1,2,\ldots,n_x$}
			\STATE Initialize $q_{i \leftarrow j}^0 = \text{const}$,
			$\forall x_j \in \bm{x}_{N_i \setminus i}$;
			\ENDFOR
			\WHILE{$e > \delta$ \textbf{and} $t \le \beta$}
			\FORALL{$x_i$, $i=1,2,\ldots,n_x$}
			\FORALL{$x_j \in \bm{x}_{N_i \setminus i}$}
			\STATE Compute $q_{i \leftarrow j}^t(x_j)$ via
			\eqref{Eq:q_i_j_marginal_mess1_int} using
			$\{q_{j \leftarrow k}^{t-1}\}$,
			$\forall x_k \in
			\bm{x}_{N_{j \setminus i}\setminus j}$;
			\ENDFOR
			\ENDFOR
			\STATE $e \leftarrow \max_{i,j} \left\| q_{i \leftarrow j}^t - q_{i \leftarrow j}^{t-1} \right\|_{\infty}$;
			\STATE $t \leftarrow t + 1$;
			\ENDWHILE
			\STATE Compute beliefs $q(x_i)$ via
			\eqref{Eq:q_i_marginal_mess1_int} for all $i$;
			\STATE \textbf{Return} $q(x_i)$.
		\end{algorithmic}
	\end{algorithm}

	Theorem~\ref{thm:exact_inference} below establishes that, under the exact regime
	$r\geq\mathcal{L}_{\max}-2$, the message updates~\eqref{Eq:q_i_j_marginal_mess1_int}
	converge to a unique fixed point in at most $|\mathcal X|$ iterations, and the resulting
	beliefs are the exact marginals.
	
	\vspace{0.8\baselineskip}

	\begin{theorem}\label{thm:exact_inference}
		Consider a connected factor graph $\mathbb{G}$ containing at least one loop. If $r \ge \mathcal{L}_{\max}-2$, then the message updates~\eqref{Eq:q_i_j_marginal_mess1_int} converge to a unique fixed point in at most $|\mathcal{X}|$ iterations. Moreover, the resulting belief satisfies $q_i(x_i) = P(x_i)$ for every variable node $x_i$.
	\end{theorem}

	\vspace{-0.8\baselineskip}
	
	\begin{pf}
		
		We prove the three conclusions of the theorem in order.
		
		\medskip
		\noindent\emph{(i) The message
			updates~\eqref{Eq:q_i_j_marginal_mess1_int} converge in at most
			$|\mathcal X|$ iterations.}
		
		For each message $q_{i\leftarrow j}(x_j)$, define its update depth
		by
		\begin{equation}\label{Eq:residual_depth}
			d_{i\leftarrow j}
			=
			1+
			\max_{x_k\in\bm x_{N_{j\setminus i}\setminus j}}
			d_{j\leftarrow k},
			\qquad
			\max\emptyset=0.
		\end{equation}
		For any residual neighborhood
		$N_{b\setminus a}\in\mathcal T_i^{(\ell)}$,
		\eqref{Eq:recursive_layers} implies that
		$N_{c\setminus b}\in\mathcal T_i^{(\ell+1)}$ for every
		$x_c\in\bm x_{N_{b\setminus a}\setminus b}$.
		Hence, $d_{a\leftarrow b}$ depends only on the update depths
		associated with residual neighborhoods in
		$\mathcal T_i^{(\ell+1)}$.
		By Lemma~\ref{lemma:exact_separation}\,(P3), the recursive
		expansion is acyclic and
		$\mathcal T_i^{(K+1)}=\emptyset$ for some
		$K\le|\mathcal X|$. The update depths can therefore be defined
		backward from layer $K$, where they equal one. Backward
		induction gives
		\[
		d_{a\leftarrow b}\le K-\ell+1,
		\qquad
		N_{b\setminus a}\in\mathcal T_i^{(\ell)}.
		\]
		Since $N_{j\setminus i}\in\mathcal T_i^{(2)}$, it follows that
		$d_{i\leftarrow j}\le K-1\le|\mathcal X|$.

		For each message, define $\bar q_{i\leftarrow j}(x_j)$ recursively
		by
		\begin{equation}\label{Eq:exact_residual_msg}
			\begin{split}
				\!\!\bar{q}_{i\leftarrow j}(x_j) \!\! = \!\! 
				\int \!\!\! 
				\Phi_{N_{j\setminus {i}}} ( \bm{x}_{N_{j\setminus {i}}})
				\!\!\!\!\!\!\!\!
				\prod_{x_k\in \bm{x}_{N_{j\setminus {i}}\setminus{j}}} 
				\!\!\!\!\!\!\!\! \bar{q}_{j\leftarrow k}(x_k) 
				d\bm{x}_{N_{j\setminus {i}}\setminus{j}},\!\!\!\!
			\end{split}
		\end{equation}
		where the base case $d_{i\leftarrow j}=1$ requires no incoming
		messages, while for $d_{i\leftarrow j}>1$, every message
		$\bar q_{j\leftarrow k}$ on the right-hand side has update depth
		strictly less than $d_{i\leftarrow j}$.

		We now prove by induction on $d_{i\leftarrow j}$ that
		\begin{equation}\label{Eq:finite_convergence_claim}
			q_{i\leftarrow j}^{t}(x_j)
			=
			\bar q_{i\leftarrow j}(x_j),
			\qquad
			t\ge d_{i\leftarrow j}.
		\end{equation}
		If $d_{i\leftarrow j}=1$, then
		$\bm x_{N_{j\setminus i}\setminus j}=\emptyset$. Thus, the update
		of $q_{i\leftarrow j}$ is independent of all incoming messages,
		and its first update gives
		$q_{i\leftarrow j}^{1}=\bar q_{i\leftarrow j}$. 
		Since the same update is performed at every subsequent iteration,
		$q^t_{i\leftarrow j}(x_j)=\bar q_{i\leftarrow j}(x_j)$ for all
		$t\ge 1$.

		If $d_{i\leftarrow j}=s+1$ for some $s\ge 1$, suppose 
		that~\eqref{Eq:finite_convergence_claim} holds for all messages with 
		update depth at most $s$.
		By~\eqref{Eq:residual_depth}, every message $q_{j\leftarrow k}$ 
		with $x_k\in\bm{x}_{N_{j\setminus i}\setminus j}$ has update 
		depth at most $s$, so 
		$q_{j\leftarrow k}^{s}(x_k)=\bar q_{j\leftarrow k}(x_k)$ 
		for every $x_k\in\bm{x}_{N_{j\setminus i}\setminus j}$.
		Substituting $q_{j\leftarrow k}^{s}(x_k)=\bar{q}_{j\leftarrow k}(x_k)$ 
		into~\eqref{Eq:q_i_j_marginal_mess1_int} gives 
		$q_{i\leftarrow j}^{s+1}(x_j)=\bar{q}_{i\leftarrow j}(x_j)$, 
		and hence $q_{i\leftarrow j}^{t}(x_j)=\bar{q}_{i\leftarrow j}(x_j)$ 
		for all $t\ge s+1$.
		This completes the induction.
		
		Since $d_{i\leftarrow j}\le|\mathcal X|$ for every message,
		\eqref{Eq:finite_convergence_claim} implies that all messages satisfy
		$q_{i\leftarrow j}^{t}=\bar q_{i\leftarrow j}$ for
		$t\ge|\mathcal X|$, regardless of initialization. Moreover,
		\eqref{Eq:exact_residual_msg} shows that
		$\{\bar q_{i\leftarrow j}\}$ satisfies the fixed-point equations
		of~\eqref{Eq:q_i_j_marginal_mess1_int}. Hence, the message updates
		converge to this fixed point in at most $|\mathcal X|$ iterations.
		
		\medskip
		\noindent\emph{(ii) The fixed point of the message
			updates~\eqref{Eq:q_i_j_marginal_mess1_int} is unique.}
		
		We have shown that
		$\{\bar q_{i\leftarrow j}\}$ is a fixed point
		of~\eqref{Eq:q_i_j_marginal_mess1_int}. To establish uniqueness,
		let $\{q^*_{i\leftarrow j}\}$ be any fixed point. We prove by
		induction on $d_{i\leftarrow j}$ that
		$q^*_{i\leftarrow j}=\bar q_{i\leftarrow j}$ for every message.
		
		By~\eqref{Eq:q_i_j_marginal_mess1_int}, the update of
		$q_{i\leftarrow j}$ uses exactly the messages
		$
		\left\{
		q_{j\leftarrow k}:
		x_k\in\bm x_{N_{j\setminus i}\setminus j}
		\right\}
		$.
		
		If $d_{i\leftarrow j}=1$, then
		$\bm x_{N_{j\setminus i}\setminus j}=\emptyset$.
		Hence, the right-hand sides
		of~\eqref{Eq:q_i_j_marginal_mess1_int}
		and~\eqref{Eq:exact_residual_msg} are identical. Since
		$\{q^*_{i\leftarrow j}\}$ is a fixed point, the two equations give
		$q^*_{i\leftarrow j}=\bar q_{i\leftarrow j}$.
		
		If $d_{i\leftarrow j}=s+1$ for some $s\ge1$, assume that
		$q^*_{a\leftarrow b}=\bar q_{a\leftarrow b}$ for every message
		with $d_{a\leftarrow b}\le s$. By~\eqref{Eq:residual_depth},
		every message $q_{j\leftarrow k}$ used in the update of
		$q_{i\leftarrow j}$ satisfies $d_{j\leftarrow k}\le s$.
		Therefore,
		\[
		q^*_{j\leftarrow k}(x_k)
		=
		\bar q_{j\leftarrow k}(x_k),
		\qquad
		x_k\in\bm x_{N_{j\setminus i}\setminus j}.
		\]
		Since $\{q^*_{i\leftarrow j}\}$ is a fixed point,
		\eqref{Eq:q_i_j_marginal_mess1_int},
		\eqref{Eq:exact_residual_msg}, and
		$q^*_{j\leftarrow k}=\bar q_{j\leftarrow k}$ for every
		$x_k\in\bm x_{N_{j\setminus i}\setminus j}$ imply
		$q^*_{i\leftarrow j}=\bar q_{i\leftarrow j}$.
		This completes the induction.
		
		Since $\{q^*_{i\leftarrow j}\}$ was arbitrary, every fixed point
		coincides with $\{\bar q_{i\leftarrow j}\}$. Therefore, the fixed
		point is unique.
		
		\medskip
		\noindent\emph{(iii) The resulting beliefs satisfy
			$q_i(x_i)=P(x_i)$ for every $x_i\in\mathcal X$.}
		
		The preceding results show that the message
		updates~\eqref{Eq:q_i_j_marginal_mess1_int} converge to the unique
		fixed point $\{\bar q_{i\leftarrow j}\}$. Substituting these
		fixed-point messages into~\eqref{Eq:q_i_marginal_mess1_int} and
		recursively applying~\eqref{Eq:exact_residual_msg} gives the
		marginalization in~\eqref{Eq:q_s_marginal_int}.
		
		By Lemma~\ref{lemma:exact_separation}\,(P4), the hierarchical
		product in~\eqref{Eq:q_s_marginal_int} equals the global potential
		function:
		\begin{equation}\label{Eq:exact_decomposition_in_proof}
			\prod_{\ell=1}^{K}
			\prod_{N\in\mathcal T_i^{(\ell)}}
			\Phi_N(\bm x_N)
			=
			\prod_{f_a\in\mathcal F}
			f_a(\bm x_a)
			=
			\psi(\bm x).
		\end{equation}  
		Therefore,
		\begin{equation}\label{Eq:exact_marginal_in_proof}
			q_i(x_i)
			=
			\frac{1}{Z_i}
			\int\psi(\bm x)\,d\bm x^{i-}
			=
			P(x_i).
		\end{equation}
		This holds for every $x_i\in\mathcal X$, so the resulting beliefs
		are the exact marginals. This completes the proof.
		\hfill$\square$
		
	\end{pf}

	\subsubsection*{2) Approximate regime}\label{sec:approx}
	When $r < \mathcal{L}_{\max} - 2$, loops longer than $r+2$ exist
	in~$\mathbb{G}$, and the residual neighborhoods may reconnect
	outside~$N_i$.
	In this case, the
	decomposition in~\eqref{Eq:P_S_1} is not guaranteed to hold 
	with equality, as correlations induced by loops longer than $r+2$
	may not be fully captured by the extended neighborhoods.
	Nonetheless, as observed in~\cite{Alec2021BPNL},
	long loops typically introduce weaker correlations than short ones, 
	and their density is often low in practice. 
	Accordingly, the message-passing
	equations~\eqref{Eq:q_i_marginal_mess1_int}
	and~\eqref{Eq:q_i_j_marginal_mess1_int}
	can still be applied as an approximation: they incorporate
	information associated with short loops captured within the
	extended neighborhoods, while the effects of longer loops are
	treated approximately.

	\noindent\textbf{Remark 2:}
	Theorem~\ref{thm:exact_inference} guarantees that the message
	updates~\eqref{Eq:q_i_j_marginal_mess1_int} converge to a fixed point and yield
	the exact marginals when $r\ge\mathcal{L}_{\max}-2$.
	However, when $r<\mathcal{L}_{\max}-2$, the exact separation properties of
	Lemma~\ref{lemma:exact_separation} no longer hold, and the
	message updates~\eqref{Eq:q_i_j_marginal_mess1_int}
	reduce to an approximate iterative method with no convergence
	guarantee, whose empirical convergence behavior is reported in
	Section~\ref{sec:application}. Algorithm~\ref{alg2} implements
	these updates.

	\noindent\textbf{Remark 3:}
	The condition $r \geq \mathcal{L}_{\max}-2$ is sufficient but not 
	necessary for exact inference: because Algorithm~\ref{alg1} recursively 
	absorbs the preliminary $r$-neighborhoods of included factor nodes, 
	$N_i^{(r)}$ may already capture all loop-induced correlations 
	at some $r < \mathcal{L}_{\max}-2$. Since $\mathcal{L}_{\max}$ is 
	NP-hard to compute on general graphs, the exact threshold cannot be 
	determined in advance. 
	In practice, $r$ serves as a tuning parameter: as $r$ increases, the extended neighborhoods capture more loop-induced correlations and inference accuracy generally improves, at the cost of higher computational overhead (cf.\ Remark~1). The value of $r$ should therefore be chosen according to the available computational resources.

	Note that ENMP with $r=0$ does not in general reduce to BP, since the
	extended $0$-neighborhood still incorporates the variables connected
	through the neighboring factor nodes. Consequently, the resulting
	extended and residual $0$-neighborhoods differ from the local
	neighborhoods used by BP, and the ENMP updates
	in~\eqref{Eq:q_i_marginal_mess1_int} and
	\eqref{Eq:q_i_j_marginal_mess1_int} are generally different from the
	BP updates.

	The computation of \eqref{Eq:q_i_marginal_mess1_int}
	and~\eqref{Eq:q_i_j_marginal_mess1_int} depends on the variable type.
	When $x_i$ is discrete, the integrals reduce to sums over its sample
	space. For continuous non-Gaussian variables, analytical evaluation is
	rarely possible, and the required integrals are approximated numerically
	through quadrature, discretization, or sampling-based
	techniques~\cite{evans1995methods}. In the Gaussian case, however,
	closed-form expressions exist, and by working in the precision form,
	marginalization becomes a purely algebraic operation, typically carried
	out via block matrix manipulations such as the Schur
	complement~\cite{Schur}.

	
	\subsection{Gaussian message passing rules}\label{Sec:GaussianMe}
	When the factor graph is Gaussian, the integrals in \eqref{Eq:q_i_marginal_mess1_int} and \eqref{Eq:q_i_j_marginal_mess1_int} admit closed-form solutions.
	Each local factor $f_a(\bm{x}_a) \propto \exp\!\left(-\frac{1}{2}\bm{x}_a^T\bm{\Lambda}_a\bm{x}_a + \bm{\eta}_a^T\bm{x}_a\right)$ in~\eqref{Eq:P(x)} takes the Gaussian canonical form, where $\bm{\Lambda}_a$ and $\bm{\eta}_a$ are the local precision matrix and potential vector.
	Any incoming message $q_{j\leftarrow k}(x_k)$ from the
	previous iteration can also be represented in Gaussian canonical
	form, i.e.,
	$q_{j\leftarrow k}(x_k)
	\propto
	\exp\!\left(
	-\frac{1}{2}\Lambda_{j\leftarrow k}x_k^2
	+\eta_{j\leftarrow k}x_k
	\right)$.
	This representation includes the noninformative constant message,
	for which
	$\Lambda_{j\leftarrow k}=0$ and
	$\eta_{j\leftarrow k}=0$.
	We assume that all eliminated precision blocks are
	positive definite whenever marginalization is required.

	Consequently, computing the message $q_{i\leftarrow j}(x_j)$
	requires two steps: aggregating the incoming information into a
	joint Gaussian potential over $\bm{x}_{N_{j\setminus i}}$, and then
	marginalizing out the remaining variables.
	Let $\bm{x}_u=\bm{x}_{N_{j\setminus i}\setminus j}$, and order the
	variables as
	$\bm{x}_{N_{j\setminus i}}
	=[x_j,\bm{x}_u^{\mathsf T}]^{\mathsf T}$.
	The aggregated precision matrix $\bm{\Lambda}_{N_{j\setminus i}}$
	and potential vector $\bm{\eta}_{N_{j\setminus i}}$ of this joint
	Gaussian potential are given by
	\begin{equation}\label{Eq:Gaussian_Aggregation}
		\begin{split}
			\bm{\Lambda}_{N_{j\setminus i}}
			&\!\!=\!\!
			\begin{bmatrix}
				\Lambda_{jj} & \bm{\Lambda}_{ju}\\
				\bm{\Lambda}_{uj} & \bm{\Lambda}_{uu}
			\end{bmatrix}
			=\!\!\!\!\!\!\!
			\sum_{f_a\in\bm{f}_{N_{j\setminus i}}}
			\!\!\!\!\![\bm{\Lambda}_a]_{\mathrm{ext}}
			+
			\!\!\!\!\!\!\!\!\!\sum_{x_k\in\bm{x}_{N_{j\setminus i}\setminus j}}
			\!\!\!\!\!\!\!\![\Lambda_{j\leftarrow k}]_{\mathrm{ext}},
			\\[1mm]
			\bm{\eta}_{N_{j\setminus i}}
			&=
			\begin{bmatrix}
				\eta_{j}\\
				\bm{\eta}_{u}
			\end{bmatrix}
			=
			\sum_{f_a\in\bm{f}_{N_{j\setminus i}}}
			\!\!\!\!\![\bm{\eta}_a]_{\mathrm{ext}}
			+\!\!\!\!\!\!
			\sum_{x_k\in\bm{x}_{N_{j\setminus i}\setminus j}}
			\!\!\!\!\![\eta_{j\leftarrow k}]_{\mathrm{ext}},
		\end{split}
	\end{equation}
	where $[\cdot]_{\text{ext}}$ denotes the zero-padded expansion of parameters matching the dimension of $\bm{x}_{N_{j\setminus i}}$, 
	$\Lambda_{jj}$ is the scalar block corresponding to $x_j$, $\bm{\Lambda}_{ju}$ and $\bm{\Lambda}_{uj}$ are the corresponding row and column blocks, respectively, $\bm{\Lambda}_{uu}$ is the block corresponding to $\bm{x}_u$, and, since the aggregated precision matrix is symmetric, $\bm{\Lambda}_{uj}=\bm{\Lambda}_{ju}^{\mathsf T}$.
	When $\bm{x}_u$ is nonempty and
	$\bm{\Lambda}_{uu}\succ\bm{0}$, marginalizing $\bm{x}_u$ gives
	\begin{equation}\label{Eq:Gaussian_Params}
		\begin{split}
			\Lambda_{i\leftarrow j}
			&=
			\Lambda_{jj}
			-
			\bm{\Lambda}_{ju}
			\bm{\Lambda}_{uu}^{-1}
			\bm{\Lambda}_{uj},
			\\
			\eta_{i\leftarrow j}
			&=
			\eta_{j}
			-
			\bm{\Lambda}_{ju}
			\bm{\Lambda}_{uu}^{-1}
			\bm{\eta}_{u}.
		\end{split}
	\end{equation}
	If $V(N_{j\setminus i})=\{x_j\}$, then
	$\Phi_{N_{j\setminus i}}(x_j)=1$ and no marginalization is
	required. The resulting noninformative message is represented by
	$\Lambda_{i\leftarrow j}=0$ and
	$\eta_{i\leftarrow j}=0$.
	
	Finally, the marginal belief $q_i(x_i)$ is computed from the
	extended neighborhood $N_i$. Let
	$\bm{x}_v=\bm{x}_{N_i\setminus i}$, and order the variables as
	$\bm{x}_{N_i}
	=
	[x_i,\bm{x}_v^{\mathsf T}]^{\mathsf T}$.
	The aggregated precision matrix and potential vector are
	\begin{equation}\label{Eq:GaussianBe}
		\begin{split}
			\bm{\Lambda}_{N_i}
			&=
			\begin{bmatrix}
				\Lambda_{ii} & \bm{\Lambda}_{iv}\\
				\bm{\Lambda}_{vi} & \bm{\Lambda}_{vv}
			\end{bmatrix}
			= \!\!\!\!
			\sum_{f_a\in\bm{f}_{N_i}}
			\!\!\![\bm{\Lambda}_a]_{\mathrm{ext}}
			+\!\!\!\!\!
			\sum_{x_j\in\bm{x}_{N_i\setminus i}}
			\!\!\!\!\![\Lambda_{i\leftarrow j}]_{\mathrm{ext}},
			\\[1mm]
			\bm{\eta}_{N_i}
			&=
			\begin{bmatrix}
				\eta_i\\
				\bm{\eta}_v
			\end{bmatrix}
			=
			\sum_{f_a\in\bm{f}_{N_i}}
			[\bm{\eta}_a]_{\mathrm{ext}}
			+
			\sum_{x_j\in\bm{x}_{N_i\setminus i}}
			[\eta_{i\leftarrow j}]_{\mathrm{ext}}.
		\end{split}
	\end{equation}
	When $\bm{x}_v$ is nonempty and
	$\bm{\Lambda}_{vv}\succ\bm{0}$, marginalizing $\bm{x}_v$ gives
	the canonical parameters of $q_i(x_i)$:
	\begin{equation}\label{Eq:GaussianBelief_Params}
		\begin{split}
			\Lambda^{\mathrm{bel}}_i
			\!=\!
			\Lambda_{ii}
			-
			\bm{\Lambda}_{iv}
			\bm{\Lambda}_{vv}^{-1}
			\bm{\Lambda}_{vi},\quad\!\!
			\eta^{\mathrm{bel}}_i
			\!=\!
			\eta_i
			-
			\bm{\Lambda}_{iv}
			\bm{\Lambda}_{vv}^{-1}
			\bm{\eta}_v.
		\end{split}
	\end{equation}
	If $\Lambda^{\mathrm{bel}}_i>0$, the estimated marginal mean and variance are then
	\begin{equation}\label{Eq:GaussianBelief}
		\hat{\Sigma}_i
		=
		\left(\Lambda^{\mathrm{bel}}_i\right)^{-1},
		\qquad
		\hat{\mu}_i
		=
		\hat{\Sigma}_i\eta^{\mathrm{bel}}_i.
	\end{equation}
	If $\bm{x}_{N_i\setminus i}$ is empty, no marginalization is
	required, $\Lambda^{\mathrm{bel}}_i=\Lambda_{ii}$, and $\eta^{\mathrm{bel}}_i=\eta_i$.

	Consequently, the general ENMP framework
	(Algorithm~\ref{alg2}) is specialized to Gaussian factor graphs by
	replacing \eqref{Eq:q_i_marginal_mess1_int} with
	\eqref{Eq:GaussianBelief_Params}--\eqref{Eq:GaussianBelief}, and
	\eqref{Eq:q_i_j_marginal_mess1_int} with
	\eqref{Eq:Gaussian_Params}.

	\noindent\textbf{Remark 4:}
	BP yields exact means but inaccurate variances when it
	converges on Gaussian graphical models with loops.
	The Gaussian ENMP updates in~\eqref{Eq:Gaussian_Params} 
	and~\eqref{Eq:GaussianBelief} address this directly: 
	when $r \geq \mathcal{L}_{\max}-2$, 
	Theorem~\ref{thm:exact_inference} guarantees exact 
	$\hat{\mu}_i$ and $\hat{\Sigma}_i$; when $r < \mathcal{L}_{\max}-2$, ENMP tends to 
	improve variance estimates over BP, 
	as demonstrated empirically in 
	Section~\ref{sec:application}.

	\subsection{Computational complexity analysis}\label{sec:complexity}
	
	Having analyzed the cost of constructing the extended neighborhoods via
	Algorithm~\ref{alg1}, we now analyze the computational
	complexity of the message updates. 
	Let $|\bm{x}_{N_{j\setminus i}}|$ denote the number of variable
	nodes contained in the residual $r$-neighborhood
	$N_{j\setminus i}$, where the superscript $r$ is omitted for
	notational brevity. Since the size of this neighborhood depends
	on $r$, the choice of $r$ affects the computational cost of each
	message update.

	For discrete random variables, the integral in~\eqref{Eq:q_i_j_marginal_mess1_int} reduces to a summation over all configurations, scaling exponentially as 
	$\mathcal{O}\left(C^{|\bm{x}_{N_{j \setminus i}}|}\right)$, where $C$ is the size of the sample space. 
	When the residual neighborhoods become too large for exhaustive summation, the message summation may be approximated by Monte Carlo sampling~\cite{Alec2021BPNL}.
	For Gaussian random variables, Gaussian messages admit closed-form solutions via matrix operations, yielding polynomial complexity $\mathcal{O}\left(|\bm{x}_{N_{j \setminus i}}|^3\right)$ dominated by matrix inversion.
	For non-Gaussian continuous random variables, if discretization approximation is adopted, it leads to a similar exponential cost as the case of discrete random variables with $C$ denoting the number of sampling points used per variable.

	In each iteration, ENMP computes $|\bm{x}_{N_i \setminus i}|$ messages 
	for every variable node $x_i$. For discrete or discretized continuous 
	variables, the per-iteration complexity is
	$\mathcal{C}_{\text{iter}}^{(r)} = \mathcal{O}\!\left( \sum_{x_i \in \mathcal{X}} \sum_{x_j \in \bm{x}_{N_{i} \setminus i}} C^{|\bm{x}_{N_{j \setminus i}}|} \right)$.
	For Gaussian distributions, it reduces to
	$\mathcal{C}_{\text{iter}}^{(r)} = \sum_{x_i \in \mathcal{X}} \sum_{x_j \in \bm{x}_{N_{i} \setminus i}} \mathcal{O}\!\left(|\bm{x}_{N_{j \setminus i}}|^3\right)$.
	Let $I^{(r)}$ denote the number of iterations required for convergence 
	at parameter $r$; the cumulative computation time until convergence is 
	then $\mathcal{O}\!\left(I^{(r)}\,\mathcal{C}_{\text{iter}}^{(r)}\right)$. 

	\section{Simulation and Validation}\label{sec:application}
	We evaluate ENMP in two settings that span the continuous and discrete regimes.
	The Gaussian setting addresses DC state estimation in electric power systems over standard IEEE bus test cases (14, 118, 300, and 1354 buses), of which factor graphs are densely loopy; here ENMP is benchmarked against BP, sGaBP-m, and GBP,
	with GBP included only on the IEEE 14-bus system (Section~\ref{sec:GauEX1}). 
	For the discrete setting, we validate ENMP on an Ising factor graph, comparing against BP and GBP (Section~\ref{sec:IsingEx2}).
	Collectively, these settings assess ENMP on both continuous Gaussian and discrete distributions on factor graphs.
	Both settings are conducted with $N_{\mathrm{MC}}=50$ independent Monte Carlo (MC)
	trials per experimental configuration.
	
	\subsection{ENMP for DC state estimation problem}\label{sec:GauEX1}
	
	\subsubsection{Factor graph model}\label{Sec:ProblemFormulation}
	
	We consider the DC state estimation problem in electric power systems~\cite{pulgar2018fast,DistributedBPPS} using standard IEEE bus test cases (14, 118, 300, and 1354 buses).
	The objective is to estimate the posterior distribution of the bus voltage phase angles $\bm{x} = \left [ x_1, x_2,\dots,x_{n_x} \right ]^{\mathrm{T}}$ at all buses; the general DC model framework is given in~\cite{pulgar2018fast}.
	
	The prior distribution of $\bm{x}$ is $p(\bm{x}) = \prod_{i=1}^{n_x} p(x_i)$, and each $x_i$ is assumed to be Gaussian and centered at the initial bus voltage angle estimates provided by the IEEE data, that is, $p(x_i) = \mathcal{N}(x_i ; \mu_i, \Lambda_{x_i}^{-1})$, where $\mu_i$ is the voltage angle from the IEEE case data and $\Lambda_{x_i}$ is the precision matrix.

	Active power injection and power flow measurements are considered.
	The measurement $z_j$ is defined as
	$z_j = h_j(\bm{x}) + \varepsilon_j$,
	where $\varepsilon_j \sim \mathcal{N}(0,\Lambda_{z_j}^{-1})$.
	For an injection at bus $i$ or a power flow on line $(i,k)$,
	$h_j(\cdot)$ is given by~\cite{pulgar2018fast},
	\begin{equation}\label{Eq:meaFunc}
		\!\! h_j(\bm{x})\!\!=\!\!
		\begin{cases}
			\sum_{k \in \mathcal{H}(i)} B_{ik}(x_i-x_k),
			\!\!& \text{power injection} ,\\
			B_{ik}(x_i-x_k),
			\!\!& \text{power flow}.
		\end{cases}
	\end{equation}
	where $B_{ik}$ is the susceptance of the line between buses $i$ and $k$,
	and $\mathcal{H}(i)$ denotes the set of neighboring buses of bus $i$.
	Stacking all measurements gives 
	$\bm{z} = \mathbf{A}\bm{x} + \bm{\varepsilon}$, 
	where 
	$\mathbf{A} = [\bm{a}_1^{\mathrm{T}},\bm{a}_2^{\mathrm{T}},\ldots,\bm{a}_{n_z}^{\mathrm{T}}]^{\mathrm{T}} 
	\in \mathbb{R}^{n_z \times n_x}$, 
	$n_z$ is the number of measurements, 
	$\bm{a}_j = [a_{1,j},a_{2,j},\ldots,a_{n_x,j}]$ 
	is the coefficient vector of $\mathbf{A}$ 
	corresponding to $z_j$, 
	and $\bm{\varepsilon} = [\varepsilon_1,\varepsilon_2,\ldots,\varepsilon_{n_z}]^{\mathrm{T}}$.


		Given the measurement likelihood $p(\bm{z}\mid\bm{x}) = \prod_{j=1}^{n_z} p(z_j \mid \bm{x}) = \prod_{j=1}^{n_z}\mathcal{N}(z_j;\bm{a}_j\bm{x},\Lambda_{z_j}^{-1})$, the posterior is,
		\begin{equation}\label{Eq:JPDF_PS}
			p(\bm{x} \mid \bm{z}) \propto
			{\textstyle \prod_{i=1}^{n_x} p(x_i) \prod_{j=1}^{n_z} p(z_j \mid \bm{x})}.
		\end{equation}
		Setting $f_i = p(x_i)$ and $g_j = p(z_j\mid\bm{x})$, we construct the factor graph
		for $p(\bm{x}\mid\bm{z})$ as illustrated in Fig.~\ref{IEEE14} and Fig.~\ref{IEEE1354}.
		Fig.~\ref{IEEE14} shows the IEEE 14-bus system in full detail, where the unknown
		variable nodes $x_i$ are represented as circles and the factor nodes $f_i$ and
		$g_j$, depicted as rectangles, indicate the associated conditional PDFs.
		Fig.~\ref{IEEE1354} shows the much larger IEEE 1354-bus system, illustrating the
		densely loopy topology with long loops that characterizes the larger benchmarks.
		The IEEE 14, 118, 300, and 1354-bus systems yield factor graphs with 14, 118, 300,
		and 1354 variable nodes and 35, 305, 838, and 3660 factor nodes, respectively.

		\begin{figure}[!t]
			\centering
			\includegraphics[width=0.5\textwidth]{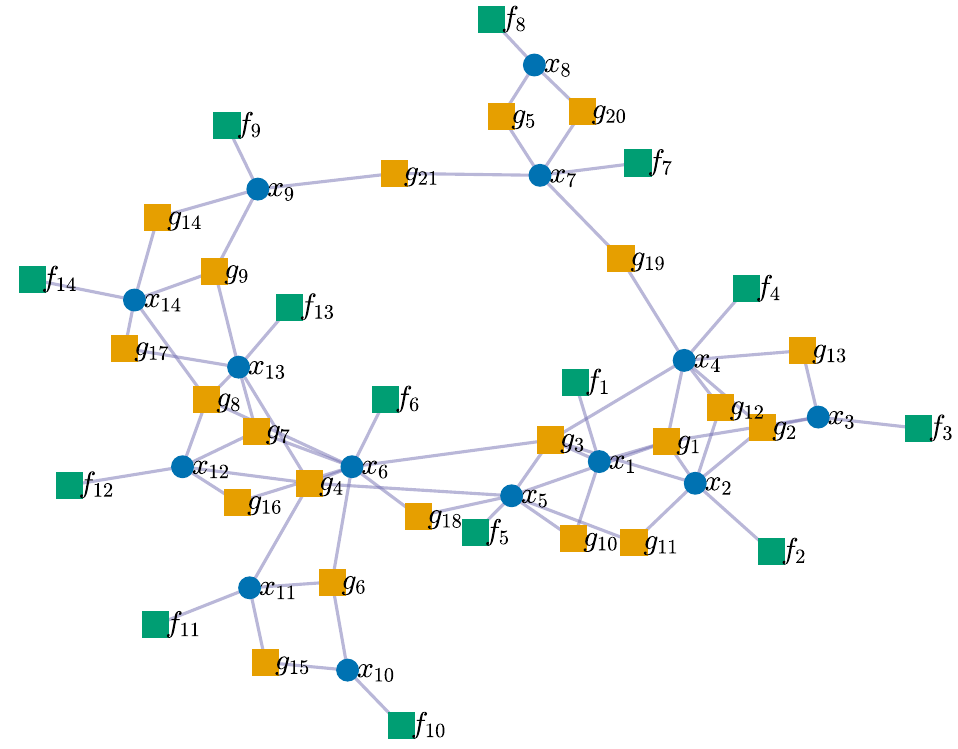}
			\caption{The factor graph of the standard IEEE 14-bus test system. Variable nodes $x_i$ are shown as circles; factor nodes $f_i$ and $g_j$ are shown as rectangles.}
			\label{IEEE14}
		\end{figure}
		
		\begin{figure*}[!t]
			\centering
			\includegraphics[width=0.95\textwidth]{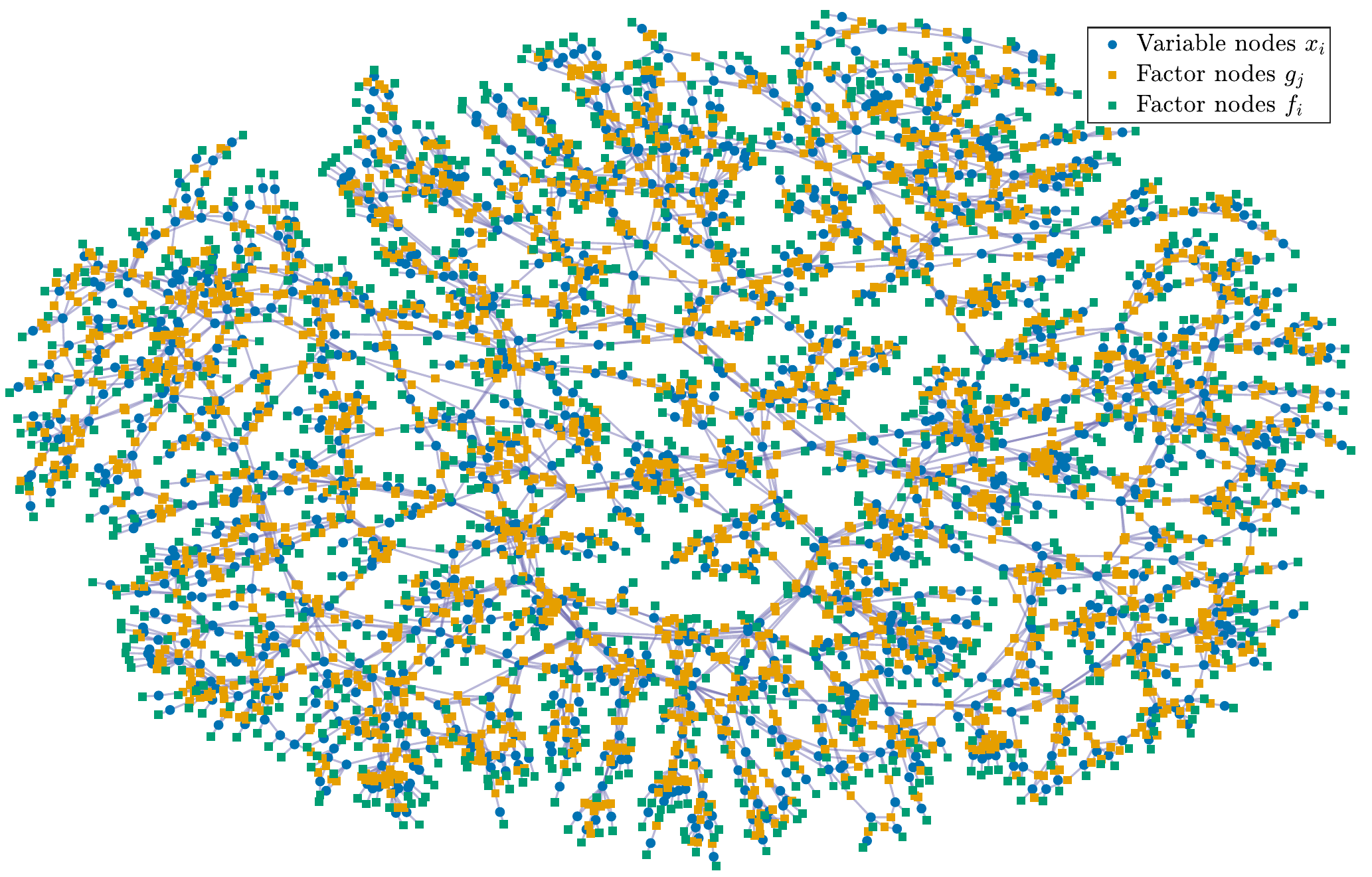}
			\caption{The factor graph of the IEEE $1354$-bus system comprising $1354$ variable nodes and $3660$ factor nodes. The variable nodes $x_i$, factor nodes $f_i$ and $g_j$ are represented as blue circles, green boxes, and orange boxes, respectively.}
			\label{IEEE1354}
		\end{figure*}

		\subsubsection{Experimental setup and baselines}
		
		The prior mean of each bus angle is initialized to the MATPOWER-provided voltage
		angle with an additive Gaussian perturbation~\cite{MATPOWER}, i.e.,
		$\mu_i = x_i^{\text{true}}+\epsilon_i$,
		where $x_i^{\text{true}}$ denotes the ground-truth bus voltage phase angle,
		and $\epsilon_i\sim\mathcal{N}(0,\Lambda_{x_i}^{-1})$; in each MC trial,
		$\{\epsilon_i\}$ and $\{\varepsilon_j\}$ are independently redrawn, yielding a
		distinct factor-graph instance. The reference bus $x_1$ is anchored at zero by a
		tight prior with variance $\Lambda_{x_1}^{-1}=10^{-10}$, while the remaining
		buses use prior variance $\Lambda_{x_i}^{-1}=10^{-2}$ and measurement-noise
		variance $\Lambda_{z_j}^{-1}=10^{-3}$.

		On the IEEE 14-bus system, we run ENMP, BP, sGaBP-m, GBP-A, and GBP-B~(see details later);
		on the larger IEEE 118-, 300-, and 1354-bus systems we run ENMP, BP,
		and sGaBP-m. For ENMP, $r$ is increased in steps of two ($r=0,2,4,\ldots$).
		Messages of ENMP, BP, and GBP are initialized with mean $0$ and
		variance $10^4$; for sGaBP-m, all messages $(\mathbf{Q}_{ji}^{(0)},\bm{v}_{ji}^{(0)})$
		and initial means $\bm{\mu}_i^{(-1)}$ are set to zero (Algorithm~\ref{alg_sgabpm} in
        Section~\ref{AppendixC} of Appendix~\ref{sec:S1}). All methods use a maximum of
		$\beta=300$ iterations and convergence threshold $\delta=10^{-4}$.
		
		For GBP, two manually designed region graphs, GBP-A and GBP-B, with maximum
		region sizes of eight and eleven variable nodes (see Section~\ref{AppendixD}
        of Appendix~\ref{sec:S1}
        for detailed constructions), are evaluated only on the IEEE 14-bus system.
		GBP is not evaluated on the larger systems for two reasons. First, constructing
		an effective region graph requires manual inspection of the graph topology to
		identify variable groups that capture loop-induced correlations, a task that
		grows increasingly prohibitive with the number of nodes and loops. Second, the
		bipartite structure of factor graphs complicates region construction, since each
		factor node couples a different number of variables, resulting in complex overlaps
		among candidate regions. For instance, designing the two configurations (GBP-A
		and GBP-B) for the IEEE 14-bus system (14 variable nodes, 35 factor nodes)
		already required careful manual analysis, while extending this to the IEEE
		118-bus system (118 variable nodes, 305 factor nodes) was found impractical. This
		limitation motivates the automatic neighborhood construction in ENMP.
		
		As sGaBP-m operates on a pairwise Markov graph, the model in~\eqref{Eq:JPDF_PS}
		is first converted into its canonical Gaussian form with global precision matrix
		$\mathbf{S}$ and potential vector $\bm{b}$ (see Appendix~\ref{AppendixC}). The variable
		set $\mathcal{X}$ is partitioned into $p$ disjoint clusters by the METIS balanced
		partitioning algorithm~\cite{Karypis1998METIS} with edge weights derived from
		$\mathbf{S}$, where $p=\lceil\alpha\,n_x\rceil$ with
		$\alpha\in\{0.02,0.05,0.10,0.20,0.50,1.00\}$ and the regularization parameter is
		swept over $\lambda\in\{10^{-2},10^{0},10^{2},10^{4},10^{8},10^{12}\}$. We compare
		ENMP against sGaBP-m only under this baseline configuration.

		Implementations of all methods are provided in Appendix~\ref{sec:S1}.
		All algorithms are implemented in MATLAB, except for the following steps that are
		accelerated by compiled C/MEX code: in ENMP, the core path search in
		extended-neighborhood generation (Algorithm~\ref{alg1}, invoked at Step~1 of
		Algorithm~\ref{alg2}); and in sGaBP-m, the graph partitioning, which invokes the
		C-based METIS library~\cite{Karypis1998METIS} through a MEX interface.

		\subsubsection{Evaluation metrics} \label{sec:GauMetrics}

		All evaluation metrics defined below are first computed for each individual MC
		trial and then averaged over all $N_{\mathrm{MC}}$ trials.
		
		To evaluate estimation accuracy, we use the average AKLD and RMSE for both
		marginal means and variances. Let $p^{(m)}(x_i|\bm{z})$ and $b^{(m)}(x_i)$ denote
		the exact and approximate marginals from the evaluated algorithm in the $m$-th
		trial, respectively, and let $D(p\|q)$ denote the KL divergence between two
		univariate Gaussians $\mathcal{N}(\mu_1,\Lambda_1^{-1})$ and
		$\mathcal{N}(\mu_2,\Lambda_2^{-1})$ given by
		$D(p\|q) = \frac{1}{2}\bigl[
		\ln\frac{\Lambda_2^{-1}}{\Lambda_1^{-1}}
		+\mathrm{Tr}(\Lambda_2\Lambda_1^{-1})
		+(\mu_1-\mu_2)^{\top}\Lambda_2(\mu_1-\mu_2)
		-1\bigr]$.
		These metrics are defined as
		\begin{equation*}\label{Eq:MC_Metrics}
			\begin{aligned}
				\text{AKLD} \!\!&=\!\! {\textstyle\frac{1}{N_{\mathrm{MC}}} \sum_{m=1}^{N_{\mathrm{MC}}} \!\frac{1}{n_x}\!
					\sum_{i=1}^{n_x} \! D\!\left(p^{(m)}(x_i|\bm{z})\parallel b^{(m)}(x_i)\right)},\\
				\text{RMSE}_{\mu} \!\!&=\!\! {\textstyle\frac{1}{N_{\mathrm{MC}}}\sum_{m=1}^{N_{\mathrm{MC}}}
					\sqrt{\frac{1}{n_x}\sum_{i=1}^{n_x}
						\bigl(\hat{\mu}_{i,m}-\mu_{i,m}^{\text{true}}\bigr)^2}},\\
				\text{RMSE}_{\sigma^2} \!\!&=\!\! {\textstyle\frac{1}{N_{\mathrm{MC}}}\sum_{m=1}^{N_{\mathrm{MC}}}
					\sqrt{\frac{1}{n_x}\sum_{i=1}^{n_x}
						\bigl(\hat{\Lambda}_{i,m}^{-1}-\Lambda_{i,m}^{-1,\text{true}}\bigr)^2}},
			\end{aligned}
		\end{equation*}
		where $\hat{\mu}_{i,m}$ and $\hat{\Lambda}_{i,m}^{-1}$ are the estimated
		marginal mean and variance in the $m$-th trial, and $\mu_{i,m}^{\text{true}}$
		and $\Lambda_{i,m}^{-1,\text{true}}$ are the corresponding exact reference values.
		If the AKLD or RMSE of a given trial exceeds $10^{10}$ or $10^5$,
		respectively, it is capped at that value.

		Furthermore, we compare the computational cost by considering the average 
		number of iterations, the average computation time for message updates (Comp), and 
		the pre-construction time (Pre). The Pre is defined as the generation 
		time of the extended $r$-neighborhood for ENMP (Step~1 of Algorithm~\ref{alg2}) 
		and the cluster construction time for sGaBP-m; for GBP, the region graph is 
		constructed manually and thus no pre-construction time is reported. 
		Since the Pre depends only on the graph topology and parameter $r$ (or $p$), 
		it is reported as a single value.

		\subsubsection{Results and analysis}
		
		The results validate the inference accuracy, 
		convergence behavior, and computational cost 
		of the tested methods.
		The experimental results for the IEEE 14-bus, IEEE 118-bus, IEEE 300-bus, and IEEE 1354-bus systems are presented in Fig.~\ref{AKLD}, Fig.~\ref{CMT_MeanVar}, Fig.~\ref{ANI} and Fig.~\ref{CMT}.
		For sGaBP-m, the best-performing configuration of $(\lambda, \alpha)$ 
		is selected for each test system by minimizing the AKLD over the 
		tested grid, and is shown as a horizontal baseline.
		Note that this selection inherently favors sGaBP-m, as it requires
		an exhaustive search over $6 \times 6 = 36$ configurations;
		the majority of these configurations did not satisfy the convergence
		criterion within $\beta=300$ iterations and reach the prescribed caps,
		as detailed in Appendix~\ref{sec:sgabpm_full}.
		GBP-A and GBP-B are included only on the IEEE 14-bus system and are shown as horizontal baselines, since they are independent of $r$.

		Fig.~\ref{AKLD} and Fig.~\ref{CMT_MeanVar} show that the AKLD, 
		$\mathrm{RMSE}_{\mu}$, and $\mathrm{RMSE}_{\sigma^2}$ of ENMP 
		decrease as $r$ increases.
		For the IEEE 14-bus and IEEE 118-bus systems, ENMP achieves exact 
		marginal inference at $r=8$, with this accuracy maintained for all 
		larger tested values of $r$. For the IEEE 300-bus system, exact 
		marginal inference is achieved at $r=10$. In these three systems, 
		exact marginal inference is confirmed by $\mathrm{RMSE}_{\mu}<10^{-10}$, 
		$\mathrm{RMSE}_{\sigma^2}<10^{-15}$, and $\mathrm{AKLD}<10^{-15}$.
		For the IEEE 1354-bus system, exact marginal inference is not achieved 
		within the tested range $r \leq 14$; its AKLD and RMSE continue to 
		decrease as $r$ increases within this range.
		This improvement stems from the fact that a larger $r$ causes the 
		extended neighborhoods to capture more loop-induced correlations, 
		thereby improving inference accuracy.

		\begin{figure}[!h]
			\centering
			\includegraphics[width=0.46\textwidth]{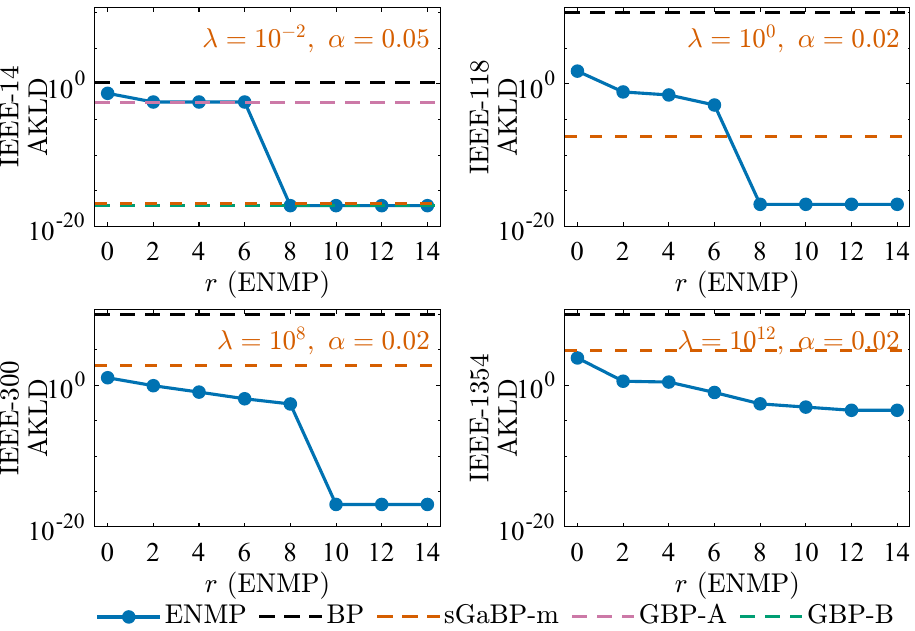}
			\caption{The AKLD of BP, sGaBP-m (best configuration), ENMP, and GBP (GBP-A and GBP-B, on IEEE 14 only) on the IEEE 14, 118, 300, and 1354 test systems. The best-performing sGaBP-m configuration $(\lambda,\alpha)$ is annotated in each subplot. As $r$ increases, the AKLD of ENMP decreases consistently, achieving exact marginal inference on the IEEE-14, IEEE-118, and IEEE-300 systems.}
			\label{AKLD}
		\end{figure}

		\begin{figure}[!t]
			\centering
			\includegraphics[width=0.46\textwidth]{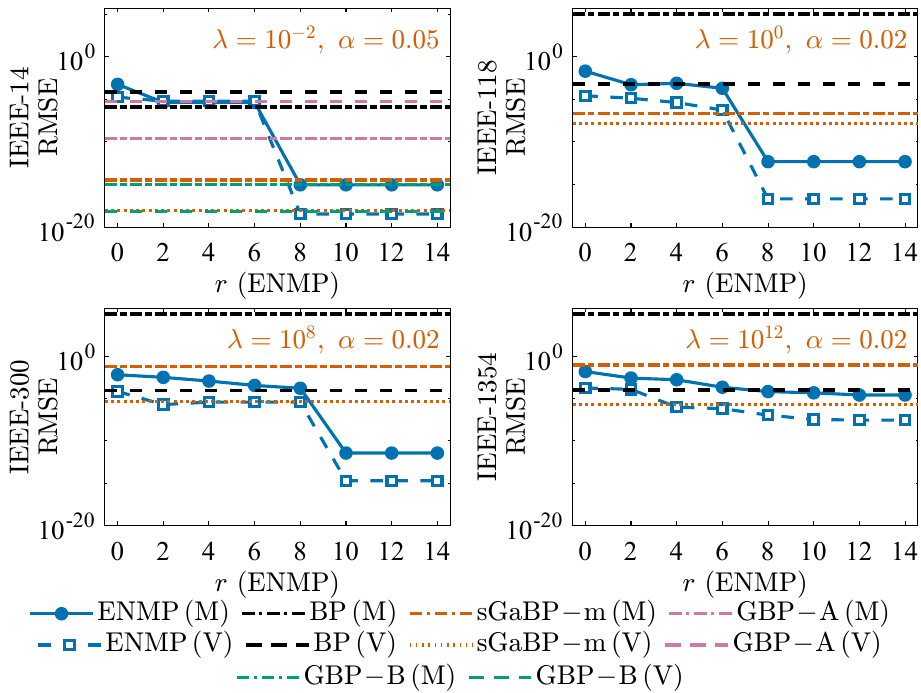}
			\caption{The $\text{RMSE}_{\mu}$ and $\text{RMSE}_{\sigma^2}$ of BP, sGaBP-m (best configuration), ENMP, and GBP (GBP-A and GBP-B, on IEEE 14 only) on the IEEE 14, 118, 300, and 1354 test systems, where (M) and (V) in the legend denote $\text{RMSE}_{\mu}$ and $\text{RMSE}_{\sigma^2}$, respectively. ENMP shows a consistent decrease in both mean and variance errors as $r$ increases.}
			\label{CMT_MeanVar}
		\end{figure}

		The baseline comparisons further distinguish mean 
		and variance estimation performance.
		On the IEEE 14-bus system, GBP-A yields 
		AKLD $< 10^{-3}$, $\mathrm{RMSE}_{\mu} < 10^{-9}$, 
		and $\mathrm{RMSE}_{\sigma^2} < 10^{-5}$, 
		while GBP-B achieves exact marginal inference 
		with AKLD $< 10^{-17}$, 
		$\mathrm{RMSE}_{\mu} < 10^{-14}$, and 
		$\mathrm{RMSE}_{\sigma^2} < 10^{-17}$, 
		comparable to ENMP at $r \geq 8$. 
		This difference demonstrates that the marginal 
		accuracy of GBP depends on the selected 
		region-graph configuration (the largest outer 
		region of GBP-B contains eleven variable nodes, 
		while that of GBP-A contains eight), whereas 
		ENMP achieves comparable accuracy through the 
		single parameter~$r$ without manual region 
		design.
		For sGaBP-m, under the tested configurations, the 
		best configuration on the IEEE 14-bus system achieves 
		AKLD $< 10^{-16}$ and both $\mathrm{RMSE}_{\mu}$ and 
		$\mathrm{RMSE}_{\sigma^2}$ below $10^{-14}$, 
		comparable to ENMP for $r \geq 8$. 
		On the IEEE 118-bus system, the best configuration 
		yields AKLD $< 10^{-6}$, 
		$\mathrm{RMSE}_{\mu} < 10^{-6}$, and 
		$\mathrm{RMSE}_{\sigma^2} < 10^{-7}$, all of which 
		are lower than those of ENMP at $r = 6$ but higher 
		than those of ENMP for $r \geq 8$. 
		On the IEEE 300-bus and IEEE 1354-bus systems, the
		best configurations yield AKLD values above $10^{3}$;
		both the AKLD and $\mathrm{RMSE}_{\mu}$ remain
		higher than those of ENMP across all tested
		values of~$r$, and the $\mathrm{RMSE}_{\sigma^2}$ of sGaBP-m
		likewise exceeds that of ENMP for $r \geq 8$ on the
		IEEE 300-bus system and for $r \geq 4$ on the IEEE 1354-bus
		system.
		Moreover, the performance of sGaBP-m depends on the 
		heuristic regularization parameter~$\lambda$ and the 
		prescribed~$\alpha$; for each test system, many 
		configurations reach the prescribed caps of $10^{10}$ 
		for AKLD and $10^{5}$ for RMSE, as detailed in Appendix~\ref{sec:sgabpm_full}.
		For BP, convergence is not reached within 
		$\beta=300$ iterations on the IEEE 118-bus, IEEE 
		300-bus, and IEEE 1354-bus systems; accordingly, its reported AKLD and 
		RMSE values are capped at $10^{10}$ and $10^{5}$, 
		respectively. On the IEEE 14-bus system, BP 
		converges and provides accurate mean estimates, but 
		its $\text{RMSE}_{\sigma^2}$ remains higher than 
		that of ENMP.

		Fig.~\ref{ANI} shows the convergence behavior 
		of the tested methods.
		For ENMP, the average number of iterations 
		decreases as $r$ increases, and ENMP reaches 
		the convergence criterion within $50$ iterations 
		for $r \geq 4$ on all four systems.
		For GBP, on the IEEE 14-bus system, GBP-A and 
		GBP-B both satisfy the convergence criterion 
		within $50$ iterations, with GBP-A requiring 
		more iterations than ENMP for $r>0$.
		For sGaBP-m, under its best-performing 
		configuration, the convergence criterion is 
		not satisfied within $\beta=300$ iterations 
		on the IEEE 118-, 300-, and 1354-bus systems.
		For BP, the convergence criterion is 
		not satisfied within $\beta=300$ iterations 
		on the IEEE 118-, 300-, and 1354-bus systems.

		\begin{figure}[!t]
			\centering
			\includegraphics[width=0.46\textwidth]{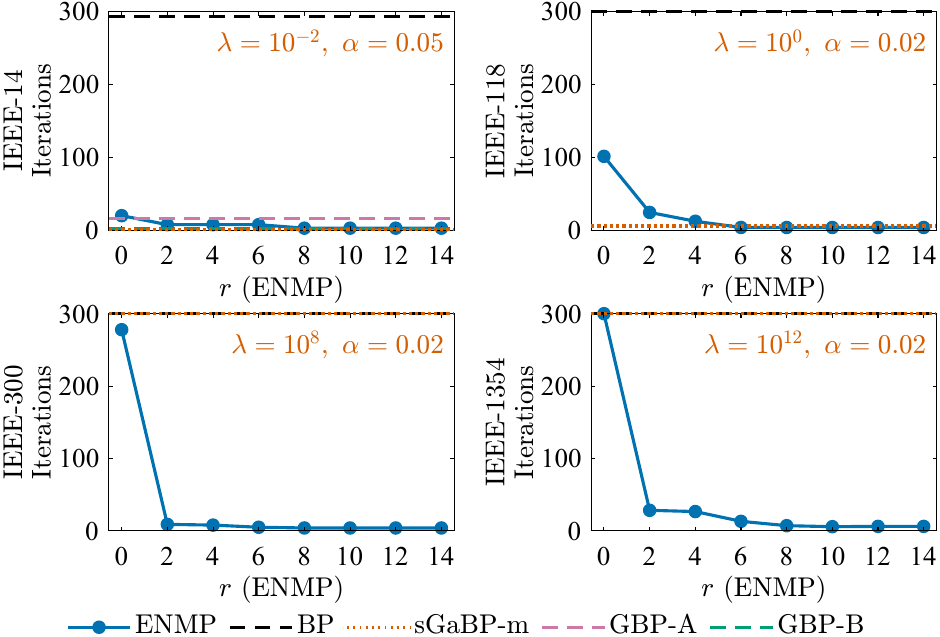}
			\caption{The average number of iterations required for convergence of BP, sGaBP-m (best configuration), ENMP, and GBP (GBP-A and GBP-B, on IEEE 14 only) on the IEEE 14, 118, 300, and 1354 test systems. ENMP converges in fewer iterations as $r$ increases, typically within $50$ iterations for $r \geq 4$.}
			\label{ANI}
		\end{figure}

		Fig.~\ref{CMT} reports the Pre and Comp of the tested methods.
		For ENMP, the Pre generally increases with~$r$ (roughly exponentially), since
		larger~$r$ enlarges the extended neighborhoods; however, the prunings 
		in the extended-neighborhood construction keep the Pre tractable over the tested range
		of~$r$ (e.g., about $45$~s on IEEE-1354 at $r=14$).
		The Comp does not necessarily increase monotonically with~$r$, since although
		a larger~$r$ increases the complexity of each individual message update, it may
		simultaneously reduce the number of iterations required for convergence, as
		shown in Fig.~\ref{ANI}.
		For GBP, the Comp is comparable to that of ENMP under the tested GBP-A and
		GBP-B configurations; however, this comparison does not account for the manual
		design effort required to construct the region graphs, whereas ENMP
		automatically constructs its extended neighborhoods via the single
		parameter~$r$.
		For sGaBP-m, the Comp is typically lower than that of ENMP because its
		per-iteration message updates are simpler; its Pre is also generally lower,
		although the optimized construction has markedly narrowed this gap, and at
		small~$r$ on the larger systems the two become comparable.
		
		\begin{figure}[!t]
			\centering
			\includegraphics[width=0.47\textwidth]{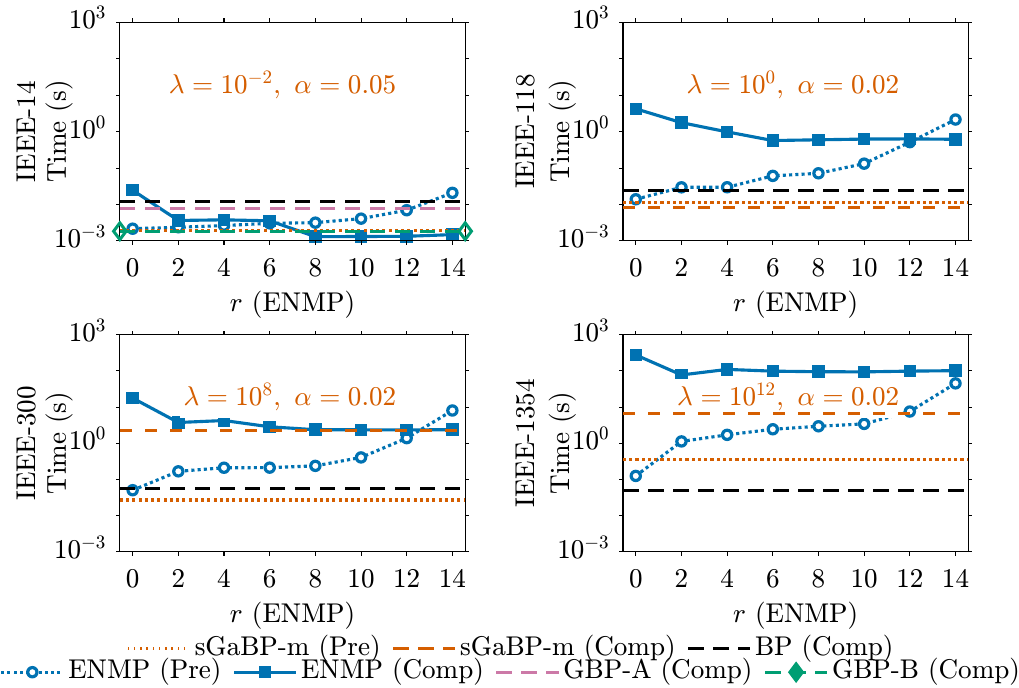}
			\caption{The Pre and Comp of BP, sGaBP-m (best configuration), ENMP, and GBP (GBP-A and GBP-B, on IEEE 14 only) on the IEEE 14, 118, 300, and 1354 test systems. 
				The Pre values are reported as single values since they depend only on the graph topology.}
			\label{CMT}
		\end{figure}

		The maximum values of $|\bm{x}_{N_{i}}|$ and $|\bm{x}_{N_{j\setminus i}}|$ for
		each IEEE system and each tested value of~$r$ are reported in
		Fig.~\ref{fig:size} in Appendix~\ref{sec:MaximumSN}. For all systems, both sizes
		grow with~$r$ and then level off as~$r$ increases, reflecting how many nodes each
		extended neighborhood absorbs. Because the IEEE power-grid graphs are densely
		loopy and contain long loops, the extended neighborhoods expand rapidly with~$r$,
		and the larger systems (IEEE 300- and 1354-bus) reach substantially larger
		neighborhood sizes than the smaller ones (IEEE 14- and 118-bus). Since the
		pre-construction cost grows roughly exponentially in~$r$, attaining exact inference
		on the IEEE 1354-bus system would require an~$r$ beyond the tested range, at which
		this cost becomes prohibitive; this is why exact marginal inference is not attained
		on this system within $r\le 14$.
		
		In summary, these experiments demonstrate that ENMP provides a tunable
		accuracy--complexity trade-off through $r$, achieving exact marginal inference on
		the IEEE 14-bus, IEEE 118-bus, and IEEE 300-bus systems, while continuously
		improving marginal accuracy on the IEEE 1354-bus system over the tested range.
		
		\subsection{ENMP for the Ising inference problem}\label{sec:IsingEx2}
		
		\subsubsection{Factor graph model}\label{Sec:IsingProblemFormulation}
		To assess the applicability of ENMP to discrete factor graphs, we consider an
		Ising model~\cite{Alec2021BPNL} with $n_x$ binary variables
		$x_i\in\{-1,+1\}$ and a set $\mathcal{E}_{\mathrm{I}}$ of pairwise interactions.
		The corresponding factor graph $\mathbb{G}=(\mathcal{X},\mathcal{F},\mathcal{E})$
		has factor set $\mathcal{F}$ consisting of pairwise factors
		$f_{ij}(x_i,x_j)=\exp(J_{ij}x_ix_j)$ for each interacting pair
		$(i,j)\in\mathcal{E}_{\mathrm{I}}$ and factors $f_i(x_i)=\exp(h_ix_i)$, yielding
		the joint PMF,
		\begin{equation}\label{Eq:Ising}
			P(\bm{x}) = \frac{1}{Z}
			\exp\!\Bigl(
			\sum_{(i,j)\in\mathcal{E}_{\mathrm I}} J_{ij}x_i x_j
			+ \sum_{i=1}^{n_x} h_i x_i
			\Bigr).
		\end{equation}
		where $J_{ij}$ is the coupling strength, $h_i$ is the external field, and $Z$ is
		the partition function. The objective is to infer the marginal PMF $P(x_i)$ for
		each variable node $x_i$, $i=1,2, \ldots,n_x$.

		In our experiments, the Ising factor graph is instantiated on the bus topologies
		of the Gaussian model in Section~\ref{Sec:ProblemFormulation}: the physical 
		bus-adjacency topology is retained, the continuous variables are replaced by binary variables
		$x_i\in\{-1,+1\}$, the coupling strength is set to $J_{ij}=1.0$ for every
		interacting pair, and the external field is drawn as
		$h_i\sim\mathcal{U}(-0.5,0.5)$.

		\subsubsection{Experimental setup and baselines}
		
		In each MC trial, $\{h_i\}_{i=1}^{n_x}$ are independently redrawn from
		$\mathcal{U}(-0.5,0.5)$, producing a distinct Ising factor-graph instance.
		
		For discrete variables, exact marginalization requires summing over all states in
		the relevant neighborhood or region, with cost exponential in the neighborhood or
		region size---a limitation shared by ENMP and region-based methods such as GBP.
		Since the IEEE bus topologies are very densely loopy with long loops, the ENMP
		extended neighborhoods grow rapidly with~$r$; we therefore select the 14- and
		118-bus systems, which retain nontrivial loop structures while keeping the
		extended neighborhoods tractable over the tested range of~$r$. 
		For larger or denser graphs, evaluating the message summations
		in~\eqref{Eq:q_i_j_marginal_mess1_int} becomes computationally prohibitive due to
		the exponential cost; approximating them via Monte Carlo sampling~\cite{Alec2021BPNL},
		as noted in Section~\ref{sec:complexity}, is left to future work.

		ENMP, run for $r\in\{0,2,4,\ldots,10\}$, and BP are compared on both
		systems. For GBP, the same GBP-A and GBP-B region-graph configurations as in the
		Gaussian experiment are applied to the IEEE 14-bus Ising factor graph; GBP is
		excluded from the 118-bus system, as constructing comparable configurations
		becomes impractical at this scale. Messages are initialized uniformly over
		$\{-1,+1\}$, and the maximum number of iterations and convergence threshold are
		again $\beta=300$ and $\delta=10^{-4}$. Exact marginals computed by the junction
		tree algorithm~\cite{Lauritzen1988Local} serve as the reference for accuracy
		evaluation in each MC trial. Algorithm implementations are detailed in
		Appendix~\ref{sec:S2}, and all
		algorithms in this experiment are implemented in MATLAB.

		\subsubsection{Evaluation metrics}
		
		Accuracy is evaluated by the average total variation (ATV),
		\begin{equation*}\label{Eq:MC_ATV}
			\!\!\mathrm{ATV} \!\!=\!\!
			\frac{1}{N_{\mathrm{MC}}}\sum_{m=1}^{N_{\mathrm{MC}}}\!
			\frac{1}{2n_x}\sum_{i=1}^{n_x}  \!\! \sum_{x_i}\!
			\left|P^{(m)}(x_i)\!-\!q^{(m)}(x_i)\right|,
		\end{equation*}
		where $P^{(m)}(x_i)$ and $q^{(m)}(x_i)$ denote the exact marginal from the
		junction tree algorithm and the approximate marginal from the evaluated algorithm in the
		$m$-th trial, respectively. Computational cost is assessed using the average number of iterations and the average message-update time (Comp), as defined in
		Section~\ref{sec:GauMetrics}.
		
		\subsubsection{Results and analysis}
		The results for the Ising factor graphs constructed on the topologies of the IEEE 14-bus and IEEE 118-bus systems are presented in Fig.~\ref{ATV}. 
		
		\begin{figure}[!h]
			\centering
			\includegraphics[width=0.47\textwidth]{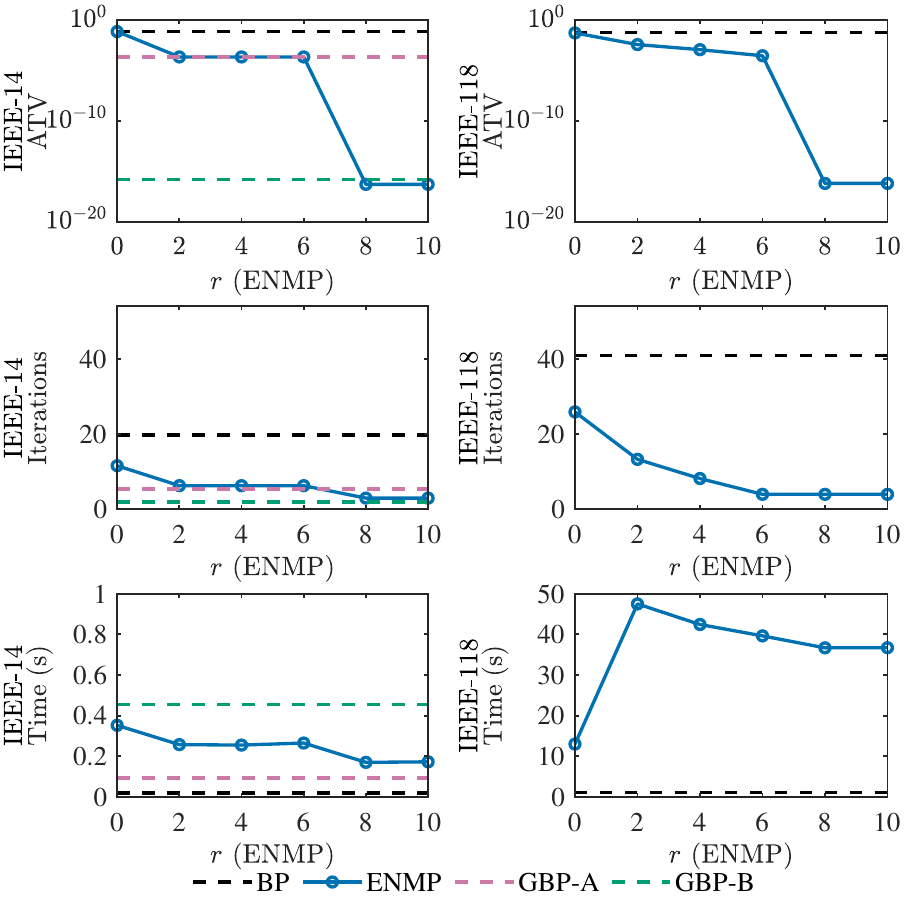}
			\caption{The ATV, the average number of iterations 
				required for convergence, and the Comp
				of BP, ENMP, and GBP on the Ising factor 
				graphs constructed on the topologies of the IEEE 
				14-bus and IEEE 118-bus systems. GBP-A and GBP-B 
				are included only on the IEEE 14-bus system.}
			\label{ATV}
		\end{figure}

		The ATV results are presented at the top of Fig.~\ref{ATV}.
		For ENMP, the ATV decreases substantially as 
		$r$ increases on both test systems. 
		On the IEEE 14-bus system, the ATV decreases 
		from approximately $10^{-1}$ at $r=0$ and 
		reaches approximately $10^{-16}$ at $r \geq 8$, 
		achieving exact marginal inference. 
		On the IEEE 118-bus system, ENMP exhibits a 
		similar trend, with the ATV reaching 
		approximately $10^{-16}$ at $r \geq 8$.
		For GBP, on the IEEE 14-bus system, GBP-A 
		yields ATV approximately $10^{-4}$, comparable 
		to ENMP at $r \in \{2,4,6\}$, while GBP-B 
		achieves exact marginal inference with ATV 
		approximately $10^{-16}$, comparable to ENMP 
		at $r \geq 8$. 
		This difference demonstrates that the marginal 
		accuracy of GBP depends on the selected 
		region-graph configuration.
		For BP, the ATV remains above that 
		of ENMP at all tested values of~$r$ on both 
		test systems.

		The average number of iterations is shown in 
		the middle panels of Fig.~\ref{ATV}. 
		For ENMP, the number of iterations decreases as 
		$r$ increases, from approximately $10$ to $3$ 
		on the IEEE 14-bus system and from 
		approximately $25$ to $4$ on the IEEE 118-bus 
		system. Across all tested values of $r$, ENMP 
		requires fewer iterations than BP. 
		On the IEEE 14-bus system, GBP-A and GBP-B 
		also require fewer iterations than BP 
		under the tested region-graph configurations.
		
		The Comp is shown in the lower panels of Fig.~\ref{ATV}. The Pre is not further
		discussed here, since its dependence on $r$ has already been examined in the
		Gaussian experiments. For ENMP, on the IEEE 14-bus system, the Comp shows a
		generally decreasing trend as $r$ increases. On the IEEE 118-bus system, it first
		increases to a peak at $r=2$ and then decreases as $r$ increases, since the
		reduction in the number of iterations required for convergence dominates the
		increase in per-message cost at larger~$r$. For GBP, on the IEEE 14-bus system,
		the Comp varies with the choice of region-graph configuration: GBP-A yields a
		higher Comp than ENMP, whereas GBP-B yields a lower Comp. For BP, the
		Comp is the lowest among all methods on both systems.
		
		Overall, these experiments validate the applicability of ENMP to discrete factor graphs and demonstrate that adjusting $r$ provides a tunable
		accuracy--complexity trade-off and enables exact marginal inference without manual
		region-graph construction.

		\section{Conclusion}
		\label{Sec:discussion}
		In this work, we proposed the ENMP method for marginal inference in
		loopy factor graphs. 
		ENMP automatically constructs extended neighborhoods from
		the bipartite factor-graph topology to capture loop-induced
		correlations, without requiring the manual design or
		application-specific heuristics typically used for region-graph
		construction in region-based methods such as GBP.
		A single parameter~$r$ trades off
		accuracy against complexity: in the exact regime, the message
		recursion converges to its unique fixed point in finitely many
		iterations and yields the exact marginals, 
		which for Gaussian factor graphs gives exact variances and resolves the
		long-standing inaccuracy of variance estimates in loopy Gaussian BP;
		in the approximate regime, the experiments show that
		the approximate marginals generally become more accurate as~$r$ increases.
		Experiments on power-grid benchmarks confirm these properties for
		Gaussian factor graphs, with ENMP attaining exact marginal inference
		on several systems and matching or outperforming the tested
		baselines, including a manually designed GBP region graph on the
		IEEE 14-bus system. Experiments on discrete Ising
		factor graphs confirm that ENMP extends to discrete variables and
		achieves exact marginals as~$r$ increases.

\appendix
\numberwithin{equation}{section}
\numberwithin{figure}{section}
\numberwithin{table}{section}
\numberwithin{algorithm}{section}
\numberwithin{lemma}{section}
\numberwithin{theorem}{section}
\numberwithin{corollary}{section}

\makeatletter
\let\ENMPorigseccntformat\@seccntformat
\renewcommand{\@seccntformat}[1]{%
  \ifstrequal{#1}{section}
    {Appendix~\thesection\quad}
    {\ENMPorigseccntformat{#1}}%
}
\makeatother

\section{Proof of Lemma~\ref{lemma:exact_separation}}
	\label{appendix:exact_separation}
	
	Before proving Lemma~\ref{lemma:exact_separation}, we state an
	auxiliary lemma used in the subsequent proofs. Recall that the
	superscript $r$ is omitted throughout this appendix. 
	For a subgraph $\mathbb H$, a walk $w$, a path $p$, and a
	loop $C$, $V(\mathbb H)$, $V(w)$, $V(p)$, and $V(C)$ denote
	their node sets, respectively.
	
	\vspace{0.8\baselineskip}
	
	\begin{lemma}\label{lemma:cycle_absorption}
		Consider the factor graph
		$\mathbb G=(\mathcal X,\mathcal F,\mathcal E)$.
		For any variable node $x_i\in\mathcal X$, let $N_i$ be its
		extended $r$-neighborhood and $V(N_i)$ its node set.
		If a loop $C$ in $\mathbb G$ satisfies
		$|C|\le r+2$ with $r\ge2$ and
		$V(C)\cap V(N_i)\cap\mathcal F\ne\emptyset$, then
		$V(C)\subseteq V(N_i)$.
	\end{lemma}
	
	\vspace{-0.8\baselineskip}
	
	\begin{pf}
		Choose $f_b\in V(C)\cap V(N_i)\cap\mathcal F$. Since
		$f_b\in V(N_i)\cap\mathcal F$,
		Definition~\ref{definition3} gives
		$V(M_{f_b})\subseteq V(N_i)$.
		Thus, it suffices to show that
		$V(C)\subseteq V(M_{f_b})$.
		
		Since $\mathbb G$ is bipartite, the two neighbors of $f_b$ on $C$
		are two distinct variable nodes, denoted by $x_k$ and $x_q$
		(Fig.~\ref{fig:proof_1}). Removing $f_b$ and its two incident
		edges from $C$ leaves a path $p$ from $x_k$ to $x_q$ with
		$\ell(p)=|C|-2\le r$. Since
		$x_k,x_q\in\mathcal N_{f_b}$, $x_k\ne x_q$, and $p$ does not
		pass through $f_b$, Definition~\ref{definition1} gives
		$p\in\mathcal P_{f_b}$ and hence
		$V(p)\subseteq V(M_{f_b})$.
		
		Finally, since $f_b\in V(M_{f_b})$, we have
		$V(C)=V(p)\cup\{f_b\}\subseteq V(M_{f_b})
		\subseteq V(N_i)$. This proves the lemma.
		\hfill$\square$
	\end{pf}
	
	\vspace{-0.8\baselineskip}
	
	\begin{figure}[!htp]
		\centering
		\includegraphics[width=0.47\textwidth]{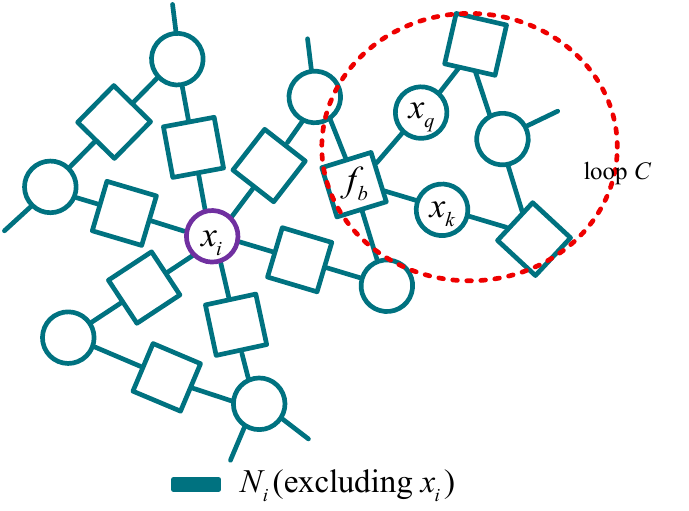}
		\caption{Extended neighborhood $N_i$ and loop $C$, with $f_b\in V(C)\cap V(N_i)\cap\mathcal F$ and its neighbors $x_k,x_q$ on $C$.}
		\label{fig:proof_1}
	\end{figure}

	We next prove Lemma~\ref{lemma:exact_separation} by establishing
	(P1), (P2), (P3) and (P4) one by one.
	Since $\mathbb G$ is a bipartite factor graph, every loop has even length which is at least four. Hence, whenever $\mathbb G$ contains a loop, the length of the longest loop  in $\mathbb{G}$ satisfies $\mathcal L_{\max}\ge 4$. Under the hypothesis $r\ge \mathcal L_{\max}-2$ of Lemma~\ref{lemma:exact_separation}, it follows that $r\ge 2$, so Lemma~\ref{lemma:cycle_absorption} can be applied.

	\subsection{Proof of (P1)}
	
	We prove that
	$V(N_{j\setminus i})\cap V(N_i)=\{x_j\}$.
	By Definition~\ref{definition4},
	$\Delta_{j\setminus i}
	=\{x_j\}\cup
	(V(N_j)\setminus V(N_i))$, and hence
	$\Delta_{j\setminus i}\cap V(N_i)=\{x_j\}$.
	Moreover,
	\[
	V(N_{j\setminus i})
	=
	\Delta_{j\setminus i}
	\cup
	\bigcup_{f_a\in\Delta_{j\setminus i}\cap\mathcal F}
	\mathcal N_{f_a}.
	\]
	It therefore remains to show that
	$\mathcal N_{f_a}\cap V(N_i)\subseteq\{x_j\}$ for every
	$f_a\in\Delta_{j\setminus i}\cap\mathcal F$.
	
	Suppose, to the contrary, that there exists a factor node
	$f_a\in\Delta_{j\setminus i}\cap\mathcal F$ having a neighbor
	$x_k\in V(N_i)$ with $x_k\ne x_j$. Since
	$f_a\in\Delta_{j\setminus i}\cap\mathcal F$, we have
	$f_a\in V(N_j)\setminus V(N_i)$. Since $N_i$ is
	connected and $x_k,x_j\in V(N_i)$, there is a path in $N_i$
	from $x_k$ to $x_j$. Let $f_b\in V(N_i)$ be the factor node
	adjacent to $x_k$ on this path (Fig.~\ref{fig:proof_2}).
	
	\begin{figure}[!htp]
		\centering
		\includegraphics[width=0.47\textwidth]{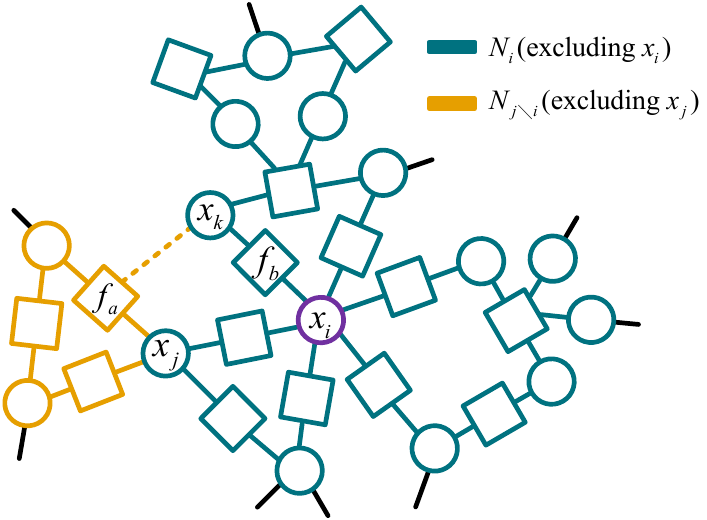}
		\caption{The hypothetical edge $x_k\!-\!f_a$, with $f_b$ on
			the path in $N_i$ from $x_k$ to $x_j$.}
		\label{fig:proof_2}
	\end{figure}
	
	
	Remove $x_k$ and its incident edges from $\mathbb G$. In the
	remaining graph, let $\mathbb G_a$ be the subgraph consisting of
	all nodes connected to $f_a$ by a path, together with all edges
	between these nodes. Define $\mathbb G_b$ analogously with
	respect to $f_b$. Then either
	$\mathbb G_a=\mathbb G_b$ or
	$V(\mathbb G_a)\cap V(\mathbb G_b)=\emptyset$.

	\emph{Case 1: $\mathbb G_a=\mathbb G_b$.}
	There is a path from $f_a$ to $f_b$ in
	$\mathbb G\setminus\{x_k\}$. Together with the edges
	$(f_a,x_k)$ and $(x_k,f_b)$, it forms a loop $C$ satisfying
	$|C|\le\mathcal L_{\max}\le r+2$. Since
	$f_b\in V(N_i)$, Lemma~\ref{lemma:cycle_absorption} gives
	$V(C)\subseteq V(N_i)$ and thus
	$f_a\in V(N_i)$, a contradiction.
	
	\emph{Case 2: $\mathbb G_a\ne\mathbb G_b$.}
	After $x_k$ and its incident edge $(x_k,f_b)$ being removed from
	the path in $N_i$ from $x_k$ to $x_j$ introduced above,
	the remaining nodes and edges form a path from
	$f_b$ to $x_j$. Hence, $x_j\in V(\mathbb G_b)$.
	Every neighbor of $x_j$ remains connected to $x_j$ after
	$x_k$ is removed and therefore belongs to
	$V(\mathbb G_b)$.
	Hence, $\mathcal N_{x_j}\subseteq V(\mathbb G_b)$.
	Since $x_k$ separates $V(\mathbb G_b)$ from all nodes outside
	$V(\mathbb G_b)\cup\{x_k\}$, every path used to construct $M_j$
	is contained in $V(\mathbb G_b)\cup\{x_k\}$. Therefore,
	$V(M_j)\subseteq V(\mathbb G_b)\cup\{x_k\}$.
	
	The same argument gives
	$V(M_{f_c})\subseteq V(\mathbb G_b)\cup\{x_k\}$ for every
	factor node $f_c\in V(\mathbb G_b)$. Starting from $V(M_j)$,
	each node-set expansion in Definition~\ref{definition3}
	therefore includes only nodes in
	$V(\mathbb G_b)\cup\{x_k\}$. By induction,
	$V(N_j)\subseteq V(\mathbb G_b)\cup\{x_k\}$.
	
	Since $\mathbb G_a$ and $\mathbb G_b$ have disjoint node sets,
	$f_a\in V(\mathbb G_a)$ implies
	$f_a\notin V(\mathbb G_b)\cup\{x_k\}$. Together with
	$V(N_j)\subseteq V(\mathbb G_b)\cup\{x_k\}$, this gives
	$f_a\notin V(N_j)$, contradicting
	$f_a\in V(N_j)$.
	
	Both cases give contradictions. Thus,
	$\mathcal N_{f_a}\cap V(N_i)\subseteq\{x_j\}$ for every
	$f_a\in\Delta_{j\setminus i}\cap\mathcal F$, and therefore
	$V(N_{j\setminus i})\cap V(N_i)=\{x_j\}$.	
	\hfill$\square$

	\subsection{Proof of (P2)}
	\noindent\textit{Part 1:
		$\bm D_{j\setminus i}$ induces a connected subgraph of $\mathbb G$.}
	
	It suffices to show that every residual neighborhood appearing in
	$\bm D_{j\setminus i}$ is connected and that successive residual
	neighborhoods have a nonempty intersection. We first prove that
	$N_{j\setminus i}$ is connected.
	
	By Definitions~\ref{definition3} and~\ref{definition4},
	$V(N_{j\setminus i})\subseteq V(N_j)$. Indeed,
	$\Delta_{j\setminus i}\subseteq V(N_j)$, and the neighbors of every
	factor node in $\Delta_{j\setminus i}$ also belong to $V(N_j)$.
	
	Let $u\in V(N_{j\setminus i})\setminus\{x_j\}$. Since $N_j$
	is connected, there is a path
	\[
	p=(v_0,v_1,\ldots,v_t),\qquad v_0=u,\quad v_t=x_j,
	\]
	in $N_j$. By (P1), $u\notin V(N_i)$. Let $v_s$ be the first
	node on $p$ that belongs to $V(N_i)$.
	
	If $v_s$ were a factor node, Definition~\ref{definition3} would
	give $v_{s-1}\in\mathcal N_{v_s}\subseteq V(N_i)$,
	contradicting the choice of $v_s$. Hence $v_s$ is a variable node.
	Since $\mathbb G$ is bipartite, $v_{s-1}$ is a factor node.
	Moreover,
	\[
	v_{s-1}\in
	V(N_j)\setminus V(N_i)
	\subseteq\Delta_{j\setminus i}.
	\]
	Definition~\ref{definition4} then gives
	$v_s\in V(N_{j\setminus i})$. By (P1), $v_s=x_j$, and hence
	$s=t$.
	The resulting path is illustrated in
	Fig.~\ref{fig:proof_2_1}.
	
	\begin{figure}[!htp]
		\centering
		\includegraphics[width=0.47\textwidth]{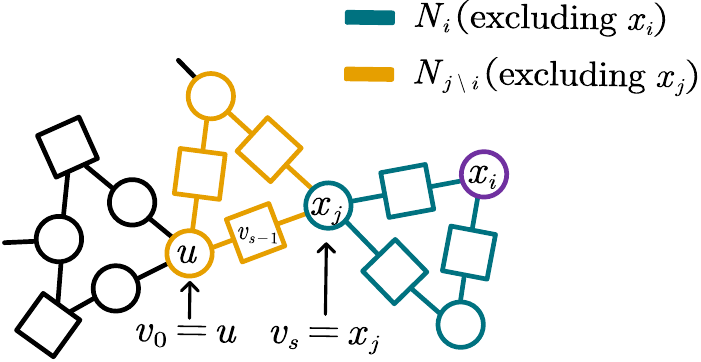}
		\caption{Illustration of the path $p$ from $v_0=u$ to
			$v_s=x_j$. The factor node $v_{s-1}$ is the last node on $p$
			outside $N_i$.}
		\label{fig:proof_2_1}
	\end{figure}

	It follows that $v_0,\ldots,v_{t-1}$ belong to
	$V(N_j)\setminus V(N_i)$, while
	$v_t=x_j$. Therefore,
	\[
	V(p)\subseteq\Delta_{j\setminus i}
	\subseteq V(N_{j\setminus i}).
	\]
	Thus, every node in $N_{j\setminus i}$ is connected to $x_j$
	within $N_{j\setminus i}$, so $N_{j\setminus i}$ is connected.
	
	The same argument applies to every residual neighborhood appearing
	in $\bm D_{j\setminus i}$. Moreover, if
	$x_k\in\bm x_{N_{j\setminus i}\setminus j}$, then
	\[
	x_k\in
	V(N_{j\setminus i})
	\cap
	V(N_{k\setminus j}).
	\]
	Thus, successive residual neighborhoods are connected and have a
	nonempty intersection. Since the union of two connected subgraphs
	with a nonempty intersection is connected, repeated application
	of this fact shows that $\bm D_{j\setminus i}$ induces a connected
	subgraph of $\mathbb G$.
	
	\medskip
	\noindent\textit{Part 2:
		$\bm D_{j\setminus i}\cap V(N_i)=\{x_j\}$.}
	
	Let $x_i\in\mathcal X$ and
	$x_j\in\bm x_{N_i\setminus i}$. By (P1),
	$V(N_{j\setminus i})\cap V(N_i)=\{x_j\}$.
	
	Suppose, to the contrary, that a residual neighborhood in
	$\bm D_{j\setminus i}$ intersects
	$V(N_i)\setminus\{x_j\}$. 
	By~\eqref{Eq:Dji_def}, choose $m\ge2$ and variable nodes
	$x_{s_0},x_{s_1},\ldots,x_{s_m}$, with
	$x_{s_0}=x_i$ and $x_{s_1}=x_j$, such that
	\[
	x_{s_{\ell+1}}\in
	\bm x_{N_{s_\ell\setminus s_{\ell-1}}\setminus s_\ell},
	\qquad \ell=1,2,\ldots,m-1,
	\]
	and $N_{s_m\setminus s_{m-1}}$ is the first corresponding
	residual neighborhood that intersects
	$V(N_i)\setminus\{x_j\}$. Thus,
	\[
	\begin{aligned}
		V(N_{s_m\setminus s_{m-1}})
		\cap(V(N_i)\setminus\{x_j\})&\ne\emptyset,\\
		V(N_{s_\ell\setminus s_{\ell-1}})
		\cap(V(N_i)\setminus\{x_j\})&=\emptyset,
		\quad \ell<m.
	\end{aligned}
	\]
	We now use these residual neighborhoods to construct a loop $C$
	satisfying $|C|\le\mathcal L_{\max}\le r+2$ that contains a
	factor node in $N_i$ and a node outside $N_i$. The factor node
	in $N_i$ allows Lemma~\ref{lemma:cycle_absorption} to be applied,
	which gives $V(C)\subseteq V(N_i)$. This contradicts the
	existence of a node of $C$ outside $V(N_i)$.

	We first construct a path $p$ from $x_j$ to a node
	$w\in V(N_i)\setminus\{x_j\}$ such that
	$V(p)\cap V(N_i)=\{x_j,w\}$.
	By Definition~\ref{definition4}, successive residual
	neighborhoods meet at $x_{s_2},x_{s_3}\ldots,x_{s_m}$. Since every
	residual neighborhood is connected by Part~1, concatenating
	paths within these residual neighborhoods gives a walk from
	$x_j$ to $x_{s_m}$ through
	$x_{s_2},x_{s_3},\ldots,x_{s_m}$.
	Since $N_{s_m\backslash s_{m-1}}$ is connected and intersects
	$V(N_i)\setminus\{x_j\}$, there is a path in
	$N_{s_m\backslash s_{m-1}}$ from $x_{s_m}$ to a node of
	$V(N_i)\setminus\{x_j\}$. By the choice of $m$,
	$x_{s_m}\notin V(N_i)\setminus\{x_j\}$. Let $w$ be the first
	node on this path, starting from $x_{s_m}$, that belongs to
	$V(N_i)\setminus\{x_j\}$, and let $v$ be the node immediately
	preceding $w$ on this path.
	By the choice of $w$, $v\notin V(N_i)$. If $w$ were a factor
	node, then $w\in V(N_i)\cap\mathcal F$. By
	Definition~\ref{definition3}, every neighbor of $w$ would belong
	to $V(N_i)$. Since $v\in\mathcal N_w$, this would give
	$v\in V(N_i)$, a contradiction. Hence, $w$ is a variable node.
	Combining these two parts gives a walk
	$\widetilde p$ from $x_j$ to $w$. By the choice of $m$ and $w$,
	\[
	V(\widetilde p)\cap V(N_i)=\{x_j,w\}.
	\]
	Let $p$ be a walk of minimum length from $x_j$ to $w$ whose
	nodes all belong to $V(\widetilde p)$.
	By its minimality, $p$ has no repeated nodes and hence is a path.
	Therefore,
	\[
	V(p)\cap V(N_i)=\{x_j,w\}.
	\]
	\begin{figure}[!htp]
		\centering
		\includegraphics[width=0.47\textwidth]{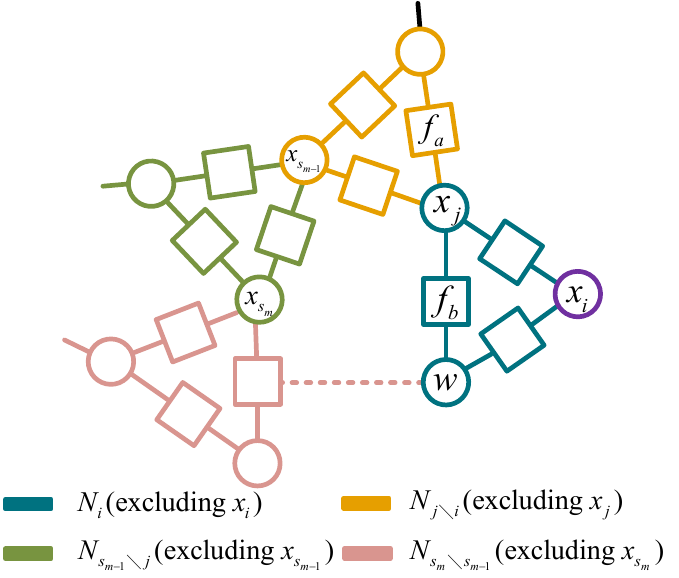}
		\caption{The path $p$ from $x_j$ to the hypothetical common
			node $w$ through successively generated residual
			neighborhoods.}
		\label{fig:proof_3}
	\end{figure}
	We next construct a path $q$ from $x_j$ to $w$ within $N_i$.
	Such a path exists because $N_i$ is connected and contains both
	$x_j$ and $w$. Since
	$V(p)\cap V(N_i)=\{x_j,w\}$ and
	$V(q)\subseteq V(N_i)$, the paths $p$ and $q$ intersect only
	at $x_j$ and $w$. Hence, their union forms a loop $C$ satisfying
	$|C|\le\mathcal L_{\max}\le r+2$.
	
	Since $w\ne x_j$, the path $q$ contains at least one edge.
	Because $\mathbb G$ is bipartite and $x_j$ is a variable node,
	$q$ contains a factor node
	$f_b\in V(N_i)\cap\mathcal F$. Since $V(q)\subseteq V(C)$, we have
	$f_b\in V(C)\cap V(N_i)\cap\mathcal F$. Therefore,
	Lemma~\ref{lemma:cycle_absorption} gives
	$V(C)\subseteq V(N_i)$.
	
	However, $x_j$ and $w$ are distinct variable nodes. Since
	$\mathbb G$ is bipartite, the path $p$ contains at least one
	factor node and therefore at least one node other than
	$x_j$ and $w$.
	The relation
	$V(p)\cap V(N_i)=\{x_j,w\}$ shows that every such node lies
	outside $V(N_i)$. Since $V(p)\subseteq V(C)$, this contradicts
	$V(C)\subseteq V(N_i)$.

	Hence, the assumed residual neighborhood
	$N_{s_m\setminus s_{m-1}}$ does not exist. Together with (P1)
	and~\eqref{Eq:Dji_def}, this gives
	$\bm D_{j\setminus i}\cap V(N_i)=\{x_j\}$.
	
	\medskip
	\noindent\textit{Part 3:
		$\bm D_{j_1\setminus i}\cap\bm D_{j_2\setminus i}
		=\emptyset$.}
	
	Let $x_{j_1},x_{j_2}\in\bm x_{N_i\setminus i}$ with
	$x_{j_1}\ne x_{j_2}$. By Part~2,
	$\bm D_{j_1\setminus i}\cap V(N_i)=\{x_{j_1}\}$, and
	$\bm D_{j_2\setminus i}\cap V(N_i)=\{x_{j_2}\}$.
	
	Suppose, to the contrary, that there is a node
	$v_z\in\bm D_{j_1\setminus i}\cap\bm D_{j_2\setminus i}$. If
	$v_z\in V(N_i)$, Part~2 gives both $v_z=x_{j_1}$ and
	$v_z=x_{j_2}$, contradicting $x_{j_1}\ne x_{j_2}$. Hence
	$v_z\notin V(N_i)$ (Fig.~\ref{fig:proof_4}).
	
	\begin{figure}[!htp]
		\centering
		\includegraphics[width=0.47\textwidth]{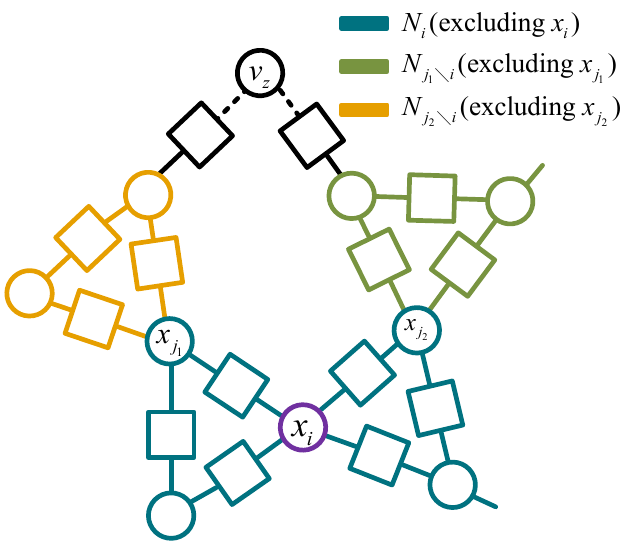}
		\caption{The hypothetical common node $v_z$ of
			$\bm D_{j_1\setminus i}$ and $\bm D_{j_2\setminus i}$.}
		\label{fig:proof_4}
	\end{figure}

	By Part~1, each of
	$\bm D_{j_1\setminus i}$ and $\bm D_{j_2\setminus i}$ induces a
	connected subgraph of $\mathbb G$. Under the assumption that they
	share $v_z$, there is a path from $x_{j_1}$ to $v_z$ within
	$\bm D_{j_1\setminus i}$ and a path from $v_z$ to $x_{j_2}$
	within $\bm D_{j_2\setminus i}$. Hence, their union contains a path
	$\tau$ from $x_{j_1}$ to $x_{j_2}$.

	Since $N_i$ is connected and contains $x_{j_1}$ and $x_{j_2}$,
	there is also a path $\rho$ in $N_i$ between them. The paths
	$\tau$ and $\rho$ meet only at $x_{j_1}$ and $x_{j_2}$, so their
	union forms a loop $C$ satisfying
	$|C|\le\mathcal L_{\max}\le r+2$.
	
	The path $\rho$ connects two distinct variable nodes and therefore
	contains a factor node
	$f_b\in V(N_i)\cap\mathcal F$. By
	Lemma~\ref{lemma:cycle_absorption},
	$V(C)\subseteq V(N_i)$. Since
	$V(\tau)\subseteq V(C)$, this contradicts
	$V(\tau)\cap V(N_i)=\{x_{j_1},x_{j_2}\}$.
	
	Therefore,
	$\bm D_{j_1\setminus i}\cap\bm D_{j_2\setminus i}=\emptyset$.
	\hfill$\square$
	
	\subsection{Proof of (P3)}
	
	Let $x_{j_1},x_{j_2},\ldots,x_{j_n}$ satisfy the conditions in
	(P3), with $j_0=i$. We prove that
	$x_i,x_{j_1},\ldots,x_{j_n}$ are pairwise distinct.
	Suppose, to the contrary, that
	$x_{j_s}=x_{j_t}$ for some $0\le t<s\le n$.

	We first show that $s\ge t+2$. Since
	$x_{j_1}\in\bm x_{N_i\setminus i}$, we have
	$x_{j_1}\ne x_i$. Similarly,
	$x_{j_{\ell+1}}\in
	\bm x_{N_{j_\ell\setminus j_{\ell-1}}\setminus j_\ell}$
	implies $x_{j_{\ell+1}}\ne x_{j_\ell}$ for
	$\ell=1,2,\ldots,n-1$. Thus, successive variables are distinct.
	If $s=t+1$, then $x_{j_s}\ne x_{j_t}$, contradicting the
	assumption. Therefore, $s\ge t+2$.
	
	We next show that
	$x_{j_s}=x_{j_t}\in\bm D_{j_{t+1}\setminus j_t}\cap
	V(N_{j_t})$. For $\ell=t+1,\ldots,s$,
	$N_{j_\ell\setminus j_{\ell-1}}$ is generated recursively from
	$N_{j_{t+1}\setminus j_t}$. Hence,
	$V(N_{j_\ell\setminus j_{\ell-1}})
	\subseteq\bm D_{j_{t+1}\setminus j_t}$ by~\eqref{Eq:Dji_def}.
	Taking $\ell=s$ and using
	$x_{j_s}\in V(N_{j_s\setminus j_{s-1}})$ gives
	$x_{j_s}\in\bm D_{j_{t+1}\setminus j_t}$. Moreover,
	$x_{j_s}=x_{j_t}$ and $x_{j_t}\in V(N_{j_t})$ give
	$x_{j_s}\in V(N_{j_t})$. Therefore,
	$x_{j_s}\in\bm D_{j_{t+1}\setminus j_t}\cap
	V(N_{j_t})$.
	
	We now derive a contradiction from (P2). Applying (P2) with
	$i=j_t$ and $j=j_{t+1}$ yields
	$\bm D_{j_{t+1}\setminus j_t}\cap V(N_{j_t})
	=\{x_{j_{t+1}}\}$. Since
	$x_{j_s}\in\bm D_{j_{t+1}\setminus j_t}\cap
	V(N_{j_t})$, it follows that
	$x_{j_s}=x_{j_{t+1}}$. Together with
	$x_{j_s}=x_{j_t}$, this implies
	$x_{j_t}=x_{j_{t+1}}$, contradicting
	$x_{j_{t+1}}\ne x_{j_t}$.
	
	Hence, no indices $0\le t<s\le n$ satisfy
	$x_{j_s}=x_{j_t}$. Thus,
	$x_i,x_{j_1},\ldots,x_{j_n}$ are pairwise distinct.
	Consequently, every sequence of successively generated residual
	neighborhoods contains no repeated variable, and the recursive
	expansion in~\eqref{Eq:recursive_layers} is acyclic.

	If $\mathcal T_i^{(\ell)}\ne\emptyset$, then
	\eqref{Eq:recursive_layers} gives $\ell-1$ successive
	generations from $N_i$ to a neighborhood in
	$\mathcal T_i^{(\ell)}$. These generations determine the
	variables $x_i,x_{j_1},\ldots,x_{j_{\ell-1}}$, where
	$N_i\in\mathcal T_i^{(1)}$ and
	$N_{j_k\setminus j_{k-1}}\in\mathcal T_i^{(k+1)}$ for
	$k=1,2,\ldots,\ell-1$, with $j_0=i$. By (P3), these $\ell$
	variables are pairwise distinct, so $\ell\le|\mathcal X|$.
	Therefore, if $\mathcal T_i^{(K)}$ is the last nonempty layer,
	then $K\le|\mathcal X|$ and
	$\mathcal T_i^{(K+1)}=\emptyset$.
	\hfill$\square$
	
	\subsection{Proof of (P4)}
	
	Let $x_i\in\mathcal X$, and consider the recursive expansion
	starting from $N_i$ defined in~\eqref{Eq:recursive_layers}. By
	(P3), this expansion terminates after $K$ layers.
	
	By~\eqref{Eq:fai_1}, proving~\eqref{Eq:P_S_1} reduces to showing
	that every factor node in $\mathcal F$ belongs to exactly one
	factor-node set $\bm f_N$, where
	$N\in\bigcup_{\ell=1}^{K}\mathcal T_i^{(\ell)}$. Equivalently,
	we prove
	\[
	\mathcal F
	=
	\bigcup_{\ell=1}^{K}
	\bigcup_{N\in\mathcal T_i^{(\ell)}}\bm f_N
	\]
	and
	$\bm f_N\cap\bm f_{N'}=\emptyset$ for all distinct
	$N,N'\in\bigcup_{\ell=1}^{K}\mathcal T_i^{(\ell)}$.
	
	\medskip
	\noindent\textit{Part 1:
		Every $f_l\in\mathcal F$ belongs to at least one $\bm f_N$.}
	
	For the given variable node $x_i$, let
	\begin{equation}\label{Eq:Ui_hierarchy}
		\bm U_i
		=
		\bigcup_{\ell=1}^{K}
		\bigcup_{N\in\mathcal T_i^{(\ell)}}
		V(N)
	\end{equation}
	denote the union of the node sets of all neighborhoods in the
	recursive expansion starting from $N_i$. Since
	$\bm f_N=V(N)\cap\mathcal F$, it is sufficient to prove that
	$
	\bm U_i=\mathcal X\cup\mathcal F$,
	that is, $\bm U_i$ contains every node of $\mathbb G$.
	
	First, we establish
	\begin{equation}\label{Eq:Na_in_Ui}
		x_a\in\bm U_i\cap\mathcal X
		\quad\Longrightarrow\quad
		V(N_a)\subseteq\bm U_i.
	\end{equation}
	The result is immediate for $x_a=x_i$, since
	$N_i\in\mathcal T_i^{(1)}$. 
	Now consider
	$x_a\in\bm U_i\cap\mathcal X$ with $x_a\ne x_i$.
	By~\eqref{Eq:Ui_hierarchy} and~\eqref{Eq:recursive_layers},
	there exist $m\ge1$ and variables
	$x_{a_0},x_{a_1},\ldots,x_{a_m}$, with $a_0=i$ and $a_m=a$,
	such that
	\[
	N_{a_s\setminus a_{s-1}}
	\in\mathcal T_i^{(s+1)},
	\qquad s=1,2,\ldots,m,
	\]
	and these residual neighborhoods are generated successively.
	For each $s$, \eqref{Eq:Ui_hierarchy} and
	Definition~\ref{definition4} give
	\[
	V(N_{a_s\setminus a_{s-1}})\subseteq\bm U_i,
	\quad
	V(N_{a_s})
	\subseteq
	V(N_{a_{s-1}})
	\cup
	V(N_{a_s\setminus a_{s-1}}).
	\]
	Starting from $V(N_{a_0})=V(N_i)\subseteq\bm U_i$,
	these relations successively give
	$V(N_{a_s})\subseteq\bm U_i$ for
	$s=1,2,\ldots,m$. In particular,
	$V(N_a)\subseteq\bm U_i$, proving
	\eqref{Eq:Na_in_Ui}.
	
	Next, we show that every neighbor of a node in $\bm U_i$ also
	belongs to $\bm U_i$. 
	
	If $u=x_a\in\mathcal X$, then
	\eqref{Eq:Na_in_Ui} gives
	$
	\mathcal N_u \subseteq V(N_a) \subseteq \bm U_i$.
	
	If $u=f_b\in\mathcal F$, then
	\eqref{Eq:Ui_hierarchy} gives
	$f_b\in V(N_i)$ or
	$f_b\in V(N_{a\setminus c})$ for some residual neighborhood
	$N_{a\setminus c}$.
	Since
	$V(N_{a\setminus c})\subseteq V(N_a)$,
	we have $f_b\in V(N_a)$ for some
	$x_a\in\bm U_i\cap\mathcal X$.
	Definition~\ref{definition3} therefore gives
	$\mathcal N_u\subseteq V(N_a)\subseteq\bm U_i$.

	Therefore,
	\begin{equation}\label{Eq:Ui_neighbor}
		u\in\bm U_i
		\quad\Longrightarrow\quad
		\mathcal N_u\subseteq\bm U_i.
	\end{equation}
	
	Finally, let $v\in\mathcal X\cup\mathcal F$. Since $\mathbb G$
	is connected, there is a path
	$p=(v_0,v_1,\ldots,v_t)$ with $v_0=x_i$ and $v_t=v$.
	We have $v_0=x_i\in\bm U_i$. Moreover,
	$v_{s+1}\in\mathcal N_{v_s}$, so
	\eqref{Eq:Ui_neighbor} successively gives
	$v_s\in\bm U_i$ for $s=0,1,\ldots,t$. 
	In particular,
	$v=v_t\in\bm U_i$. Since $v$ was arbitrary,
	$\mathcal X\cup\mathcal F\subseteq\bm U_i$. The reverse inclusion
	follows from~\eqref{Eq:Ui_hierarchy}; hence,
	$\bm U_i=\mathcal X\cup\mathcal F$.
	
	Therefore, every $f_l\in\mathcal F$ belongs to
	$V(N)$ for some
	$N\in\bigcup_{\ell=1}^{K}\mathcal T_i^{(\ell)}$.
	Since $f_l$ is a factor node, $f_l\in\bm f_N$. This proves
	Part~1.

	\textit{Part 2: Each factor node belongs to at most one $\bm f_N$.}
	
	Let
	\[
	N,N'
	\in
	\bigcup_{\ell=1}^{K}\mathcal T_i^{(\ell)},
	\qquad N\ne N'.
	\]
	We prove that $\bm f_N\cap\bm f_{N'}=\emptyset$.
	By the recursive construction and (P3), there is a finite
	sequence of successively generated neighborhoods from $N_i$ to
	$N$, and another from $N_i$ to $N'$. 
	Let $N^*$ be the last neighborhood common to the two sequences.
	Then either $N^*=N_i$ or
	$N^*=N_{c\setminus d}$ for some
	$x_c,x_d\in\mathcal X$. When $N^*=N_i$, set $x_c=x_i$.
	We distinguish whether one sequence ends at $N^*$ or both
	continue beyond $N^*$.

	\emph{Case 1: $N^*$ equals $N$ or $N'$.}
	Without loss of generality, let $N^*=N$. Let
	$N_{s\setminus c}$ be the first residual neighborhood generated
	after $N^*$ in the sequence ending at $N'$, as illustrated in
	Fig.~\ref{fig:proof_5}. 
	Then, \eqref{Eq:Dji_def} gives $V(N')\subseteq\bm D_{s\setminus c}$.
	Hence, $\bm f_{N'}\subseteq\bm D_{s\setminus c}$.
	By (P2),
	$
	\bm D_{s\setminus c}\cap V(N_c)=\{x_s\}$, we have
	$\bm f_{N'}\cap V(N_c)=\emptyset$.
	Moreover, $N=N^*$ is either $N_c$ or
	$N_{c\setminus d}$. Hence,
	$V(N)\subseteq V(N_c)$, and therefore
	$\bm f_N\subseteq V(N_c)$. It follows that
	$\bm f_N\cap\bm f_{N'}=\emptyset$.

	\begin{figure}[!htp]
		\centering
		\includegraphics[width=0.46\textwidth]{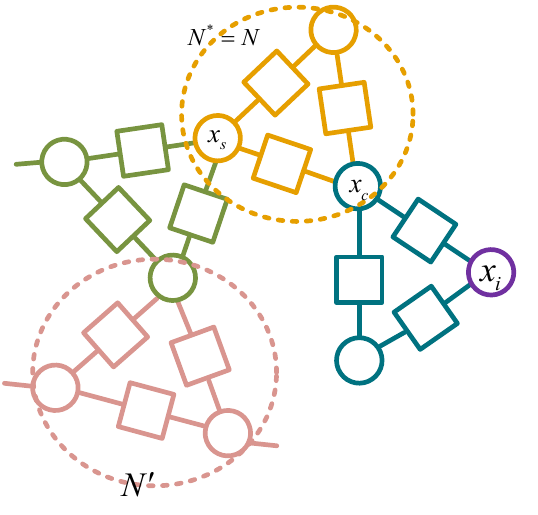}
		\caption{Case~1, where $N^*=N$ and the generation sequence to
			$N'$ continues through $x_s$.}
		\label{fig:proof_5}
	\end{figure}
	
	\emph{Case 2: $N^*$ equals neither $N$ nor $N'$.}
	The two sequences continue from $N^*$ through variables $x_p$
	and $x_q$, respectively, as illustrated in
	Fig.~\ref{fig:proof_6}. Since $N^*$ is their last common
	neighborhood, $x_p\ne x_q$. Moreover,
	\[
	\bm f_N\subseteq\bm D_{p\setminus c},
	\qquad
	\bm f_{N'}\subseteq\bm D_{q\setminus c}.
	\]
	By (P2),
	$\bm D_{p\setminus c}\cap\bm D_{q\setminus c}=\emptyset$.
	Therefore, $\bm f_N\cap\bm f_{N'}=\emptyset$.
	
	\begin{figure}[!htp]
		\centering
		\includegraphics[width=0.46\textwidth]{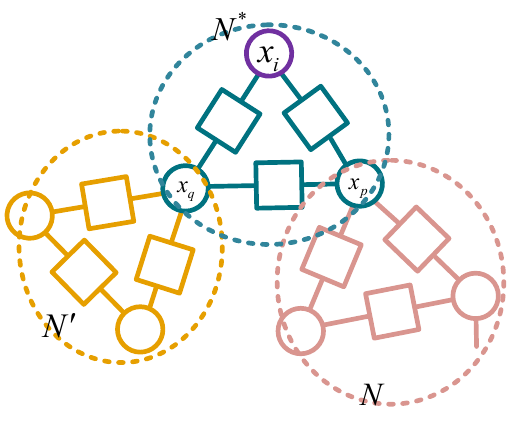}
		\caption{Case~2, where the two generation sequences continue
			from $N^*$ through the distinct variables $x_p$ and $x_q$.}
		\label{fig:proof_6}
	\end{figure}
	
	Since $N$ and $N'$ were arbitrary, both cases show that
	$\bm f_N\cap\bm f_{N'}=\emptyset$ for any two distinct
	neighborhoods in the recursive expansion. Hence, no factor node
	belongs to more than one $\bm f_N$. This proves Part~2.
	
	Parts~1 and~2 show that every factor node
	$f_l\in\mathcal F$ belongs to exactly one $\bm f_N$.
	Using~\eqref{Eq:fai_1}, we therefore obtain
	\[
	\begin{aligned}
		\prod_{\ell=1}^{K}
		\prod_{N\in\mathcal T_i^{(\ell)}}
		\Phi_N(\bm x_N)
		&=
		\prod_{\ell=1}^{K}
		\prod_{N\in\mathcal T_i^{(\ell)}}
		\prod_{f_l\in\bm f_N}
		f_l(\bm x_l)\\
		&=
		\prod_{f_l\in\mathcal F}
		f_l(\bm x_l)
		=
		\psi(\bm x).
	\end{aligned}
	\]
	Thus, the hierarchical product
	in~\eqref{Eq:P_S_1} equals $\psi(\bm x)$, proving (P4).
	This also completes the proof of
	Lemma~\ref{lemma:exact_separation}.
	\hfill$\square$
	
	\section{Implementation Details of Algorithms for DC State Estimation}
	\label{sec:S1}
	
	\subsection{ENMP for DC State Estimation}
	\label{AppendixA}
	
	To apply ENMP (Algorithm~\ref{alg2}) to the DC state estimation
	problem, we specialize the general Gaussian update rules in
	Section~\ref{Sec:GaussianMe}. For any variable node $x_i$, the
	local factor set satisfies
	$\bm{f}_{N_i}\subseteq
	\{f_i\}_{i=1}^{n_x}\cup\{g_j\}_{j=1}^{n_z}$.
	
	For the message $q_{i\leftarrow j}(x_j)$, the aggregated
	canonical parameters in \eqref{Eq:Gaussian_Aggregation} become
	\begin{equation}\label{Eq:DC_Aggregation_Total}
		\begin{split}
			\bm{\Lambda}_{N_{j\setminus i}}
			&=
			\bm{\Lambda}_{N_{j\setminus i}}^{\mathrm{loc}}
			+
			\sum_{x_k\in\bm{x}_{N_{j\setminus i}\setminus j}}
			[\Lambda_{j\leftarrow k}]_{\mathrm{ext}},
			\\
			\bm{\eta}_{N_{j\setminus i}}
			&=
			\bm{\eta}_{N_{j\setminus i}}^{\mathrm{loc}}
			+
			\sum_{x_k\in\bm{x}_{N_{j\setminus i}\setminus j}}
			[\eta_{j\leftarrow k}]_{\mathrm{ext}},
		\end{split}
	\end{equation}
	where
	\begin{equation*}
		\begin{split}
			\bm{\Lambda}_{N_{j\setminus i}}^{\mathrm{loc}}
			&=
			\sum_{g_m\in\bm{f}_{N_{j\setminus i}}}
			[\bm{a}_m^{\mathsf T}\Lambda_{z_m}\bm{a}_m]_{\mathrm{ext}}
			+
			\sum_{f_k\in\bm{f}_{N_{j\setminus i}}}
			[\Lambda_{x_k}]_{\mathrm{ext}},
			\\
			\bm{\eta}_{N_{j\setminus i}}^{\mathrm{loc}}
			&=
			\sum_{g_m\in\bm{f}_{N_{j\setminus i}}}
			[\bm{a}_m^{\mathsf T}\Lambda_{z_m}z_m]_{\mathrm{ext}}
			+
			\sum_{f_k\in\bm{f}_{N_{j\setminus i}}}
			[\Lambda_{x_k}\mu_k]_{\mathrm{ext}}.
		\end{split}
	\end{equation*}
	The message parameters
	$\Lambda_{i\leftarrow j}$ and
	$\eta_{i\leftarrow j}$ are then obtained from
	\eqref{Eq:Gaussian_Params}.
	
	Similarly, the aggregated canonical parameters for the belief
	$q_i(x_i)$ are
	\begin{equation}\label{Eq:DC_Final_Belief_Params}
		\begin{split}
			\bm{\Lambda}_{N_i}
			&=
			\bm{\Lambda}_{N_i}^{\mathrm{loc}}
			+
			\sum_{x_j\in\bm{x}_{N_i\setminus i}}
			[\Lambda_{i\leftarrow j}]_{\mathrm{ext}},
			\\
			\bm{\eta}_{N_i}
			&=
			\bm{\eta}_{N_i}^{\mathrm{loc}}
			+
			\sum_{x_j\in\bm{x}_{N_i\setminus i}}
			[\eta_{i\leftarrow j}]_{\mathrm{ext}},
		\end{split}
	\end{equation}
	where
	\begin{equation*}
		\begin{split}
			\bm{\Lambda}_{N_i}^{\mathrm{loc}}
			&=
			\sum_{g_m\in\bm{f}_{N_i}}
			[\bm{a}_m^{\mathsf T}\Lambda_{z_m}\bm{a}_m]_{\mathrm{ext}}
			+
			\sum_{f_k\in\bm{f}_{N_i}}
			[\Lambda_{x_k}]_{\mathrm{ext}},
			\\
			\bm{\eta}_{N_i}^{\mathrm{loc}}
			&=
			\sum_{g_m\in\bm{f}_{N_i}}
			[\bm{a}_m^{\mathsf T}\Lambda_{z_m}z_m]_{\mathrm{ext}}
			+
			\sum_{f_k\in\bm{f}_{N_i}}
			[\Lambda_{x_k}\mu_k]_{\mathrm{ext}}.
		\end{split}
	\end{equation*}
	The belief
	canonical parameters are computed using
	\eqref{Eq:GaussianBelief_Params}, and the estimated mean
	$\hat{\mu}_{i,\mathrm{ENMP}}$ and variance
	$\hat{\Sigma}_{i,\mathrm{ENMP}}$ are obtained from
	\eqref{Eq:GaussianBelief}.
	
	In Algorithm~\ref{alg2}, the message update in Step~9 uses
	\eqref{Eq:DC_Aggregation_Total} and
	\eqref{Eq:Gaussian_Params}, while the belief computation in Step~15
	uses \eqref{Eq:DC_Final_Belief_Params},
	\eqref{Eq:GaussianBelief_Params}, and
	\eqref{Eq:GaussianBelief}.

	\subsection{GBP for DC state estimation}\label{AppendixD}
	
	To apply GBP~\cite{Yedidia2005Constructing} to the DC state estimation
	problem, we construct a region graph over $\mathbb{G}=(\mathcal{X},\mathcal{F},\mathcal{E})$
	and specialize the message update rules to Gaussian distributions.
	
	\subsubsection{Region graph construction}
	A region graph $\mathcal{G}_R=(\mathcal{R},\mathcal{E}_R)$ organizes
	$\mathbb{G}$ into regions, where each region $r\in\mathcal{R}$ contains
	a subset of variable nodes and all factor nodes whose scope is fully
	contained within that variable set.
	Following~\cite{Yedidia2005Constructing}, let $\mathcal{P}_r$,
	$\mathcal{C}_r$, and $\mathcal{D}_r$ denote the sets of parents, children,
	and descendants of region $r$, respectively, and let
	$\mathcal{D}_r^{+}=\{r\}\cup\mathcal{D}_r$ denote region $r$ together with
	all its descendants.

	Two region-graph configurations are evaluated on the IEEE 14-bus system.
	GBP-A uses three large regions:
	$\mathcal{R}_1^{\mathrm{A}}=\{x_1,x_2,x_3,x_4,x_5,x_6\}$,
	$\mathcal{R}_2^{\mathrm{A}}=\{x_4,x_7,x_8,x_9\}$, and
	$\mathcal{R}_3^{\mathrm{A}}=\{x_5,x_6,x_9,x_{10},x_{11},x_{12},x_{13},x_{14}\}$,
	together with three intersection regions
	$s_{12}^{\mathrm{A}}=\{x_4\}$,
	$s_{13}^{\mathrm{A}}=\{x_5,x_6\}$, and
	$s_{23}^{\mathrm{A}}=\{x_9\}$.
	GBP-B uses two large regions:
	$\mathcal{R}_1^{\mathrm{B}}=\{x_1,x_2,x_3,x_4,x_5,x_6\}$ and
	$\mathcal{R}_2^{\mathrm{B}}=\{x_4,x_5,x_6,x_7,x_8,x_9,x_{10},x_{11},x_{12},x_{13},x_{14}\}$,
	together with one intersection region
	$s_{12}^{\mathrm{B}}=\{x_4,x_5,x_6\}$.
	Each region also contains all factor nodes whose neighboring variable
	nodes are contained in the region.
	In both configurations, each intersection region $s_{ij}$ has the two
	large regions sharing it as its parents in $\mathcal{G}_R$.

	\subsubsection{Gaussian message update}
	Since all factors in the DC state estimation factor graph are Gaussian,
	the message $m_{rs}(\bm{x}_s)$ from region $r$ to its child $s$
	is Gaussian, parameterized by precision matrix $\bm{\Lambda}_{rs}$ and
	potential vector $\bm{\eta}_{rs}$.
	
	The local precision matrix and potential vector accumulated over the
	factors in $r$ but not in $s$ (denoted $\bm{f}_{r\setminus s}$) are:
	\begin{subequations}\label{Eq:GBP_loc}
		\begin{align}
			\bm{\Lambda}_{r\setminus s}^{\mathrm{loc}} &\!=\!
			\sum_{g_m\in\bm{f}_{r\setminus s}}
			\!\bigl[\bm{a}_m^{\mathrm{T}}\Lambda_{z_m}\bm{a}_m\bigr]_{\mathrm{ext}}
			\!+\!\sum_{f_k\in\bm{f}_{r\setminus s}}
			\!\bigl[\bm{\Lambda}_{x_k}\bigr]_{\mathrm{ext}},\\
			\bm{\eta}_{r\setminus s}^{\mathrm{loc}} &\!=\! 
			\sum_{g_m\in\bm{f}_{r\setminus s}}\!\!\!\bigl[\bm{a}_m^{\mathrm{T}}\Lambda_{z_m}z_m\bigr]_{\mathrm{ext}}
			\!\!+\!\!\sum_{f_k\in\bm{f}_{r\setminus s}}
			\!\!\bigl[\bm{\Lambda}_{x_k}\bm{\mu}_k\bigr]_{\mathrm{ext}},
		\end{align}
	\end{subequations}
	where $[\cdot]_{\mathrm{ext}}$ denotes zero-padded expansion to the
	dimension of $\bm{x}_r$.
	Incorporating incoming messages from regions outside $\mathcal{D}_r^{+}$
	to nodes in $\mathcal{D}_r^{+}\setminus\mathcal{D}_s^{+}$,
	the aggregated precision and potential over $\bm{x}_r$ are:
	\begin{subequations}\label{Eq:GBP_agg}
		\begin{align}
			\bm{\Lambda}_r^{\mathrm{agg}} &= \bm{\Lambda}_{r\setminus s}^{\mathrm{loc}}
			+\!\!\!\sum_{\substack{u\in\mathcal{R}\setminus\mathcal{D}_r^{+} \\
					v\in\mathcal{D}_r^{+}\setminus\mathcal{D}_s^{+},\;u\in\mathcal{P}_v}}
			\!\!\!\bigl[\bm{\Lambda}_{uv}\bigr]_{\mathrm{ext}},\\
			\bm{\eta}_r^{\mathrm{agg}} &= \bm{\eta}_{r\setminus s}^{\mathrm{loc}}
			+\!\!\!\sum_{\substack{u\in\mathcal{R}\setminus\mathcal{D}_r^{+} \\
					v\in\mathcal{D}_r^{+}\setminus\mathcal{D}_s^{+},\;u\in\mathcal{P}_v}}
			\!\!\!\bigl[\bm{\eta}_{uv}\bigr]_{\mathrm{ext}}.
		\end{align}
	\end{subequations}
	The message parameters are obtained by marginalizing over
	$\bm{x}_{r\setminus s}$ via the Schur complement.
	Let $\bm{\Sigma}_r^{\mathrm{agg}}=(\bm{\Lambda}_r^{\mathrm{agg}})^{-1}$.
	Then,
	\begin{equation}\label{Eq:GBP_msg}
		\bm{\Lambda}_{rs}=
		\bigl([\bm{\Sigma}_r^{\mathrm{agg}}]_{ss}\bigr)^{-1},\quad
		\bm{\eta}_{rs}=
		\bm{\Lambda}_{rs}\bigl[\bm{\Sigma}_r^{\mathrm{agg}}\,\bm{\eta}_r^{\mathrm{agg}}\bigr]_s,
	\end{equation}
	where $[\cdot]_{ss}$ and $[\cdot]_s$ denote the block sub-matrix and
	sub-vector corresponding to $\bm{x}_s$, respectively.
	
	\subsubsection{Belief and marginal}
	Upon convergence, the belief of region $r$ is a Gaussian with parameters,
	\begin{subequations}\label{Eq:GBP_bel}
		\begin{align}
			\bm{\Lambda}_r^{\mathrm{bel}} &= \bm{\Lambda}_r^{\mathrm{loc}}
			+\sum_{\substack{v\in\mathcal{D}_r^{+} \\
					u\in\mathcal{P}_v\setminus\mathcal{D}_r^{+}}}
			\bigl[\bm{\Lambda}_{uv}\bigr]_{\mathrm{ext}},\\
			\bm{\eta}_r^{\mathrm{bel}} &= \bm{\eta}_r^{\mathrm{loc}}
			+\sum_{\substack{v\in\mathcal{D}_r^{+} \\
					u\in\mathcal{P}_v\setminus\mathcal{D}_r^{+}}}
			\bigl[\bm{\eta}_{uv}\bigr]_{\mathrm{ext}},
		\end{align}
	\end{subequations}
	where $\bm{\Lambda}_r^{\mathrm{loc}}$ and $\bm{\eta}_r^{\mathrm{loc}}$ follow
	\eqref{Eq:GBP_loc} with the summation extended to all factors in $r$
	(i.e., $\bm{f}_r$ instead of $\bm{f}_{r\setminus s}$).
	The marginal mean $\hat{\mu}_{i,\text{GBP}}$ and variance $\hat{\Sigma}_{i,\text{GBP}}$ of
	variable $x_i$ in the smallest region $r\ni x_i$ are given by
	\begin{equation}\label{Eq:GBP_marginal}
		\bm{\Sigma}_r^{\mathrm{bel}} \! = \!(\bm{\Lambda}_r^{\mathrm{bel}})^{-1},
		\hat{\mu}_{i,\text{GBP}}  \!=  \!\bigl[\bm{\Sigma}_r^{\mathrm{bel}}\,\bm{\eta}_r^{\mathrm{bel}}\bigr]_i,
		\hat{\Sigma}_{i,\text{GBP}} \!= \! \bigl[\bm{\Sigma}_r^{\mathrm{bel}}\bigr]_{ii}.
	\end{equation}
	
	The complete GBP algorithm for DC state estimation is given in
	Algorithm~\ref{alg_GBP}.
	
	\begin{algorithm}
		\caption{GBP for DC state estimation.}\label{alg_GBP}
		\begin{algorithmic}[1]
			\REQUIRE Factor graph $\mathbb{G}$, region graph $\mathcal{G}_R$,
			priors $\{\mu_{x_i},\Lambda_{x_i}^{-1}\}_{i=1}^{n_x}$,
			measurements $\{z_j,\Lambda_{z_j}^{-1}\}_{j=1}^{n_z}$,
			maximum iterations $\beta$, convergence threshold $\delta$;
			\ENSURE Beliefs $b(x_i)$, $i=1,2,\ldots,n_x$;
			\STATE Let $t=1$, $e=\infty$; initialize $\bm{\Lambda}_{rs}^{(0)}=10^{-4}\mathbf{I}$,
			$\bm{\eta}_{rs}^{(0)}=\mathbf{0}$ for all $(r,s)$ with $s\in\mathcal{C}_r$;
			\WHILE{$e>\delta$ \textbf{and} $t\leq\beta$}
			\FOR{each $(r,s)$ with $s\in\mathcal{C}_r$}
			\STATE Compute $(\bm{\Lambda}_{rs}^{(t)},\bm{\eta}_{rs}^{(t)})$
			via \eqref{Eq:GBP_loc}--\eqref{Eq:GBP_msg};
			\ENDFOR
			\STATE $e\leftarrow\max_{r,s}\,\left\| m_{rs}^{(t)} - m_{rs}^{(t-1)} \right\|_{\infty}$;
			\STATE $t\leftarrow t+1$;
			\ENDWHILE
			\FOR{each $r\in\mathcal{R}$}
			\STATE Compute $(\bm{\Lambda}_r^{\mathrm{bel}},\bm{\eta}_r^{\mathrm{bel}})$
			via \eqref{Eq:GBP_bel};
			\ENDFOR
			\FOR{$i=1,2,\ldots,n_x$}
			\STATE Select the smallest region $r\in\mathcal{R}$ with $x_i\in r$; compute
			$(\hat{\mu}_{i,\text{GBP}},\hat{\Sigma}_{i,\text{GBP}})$ via \eqref{Eq:GBP_marginal};
			\STATE $b(x_i)\leftarrow\mathcal{N}(x_i;\hat{\mu}_{i,\text{GBP}},\hat{\Sigma}_{i,\text{GBP}})$;
			\ENDFOR
			\RETURN $b(x_i)$, $i=1,2,\ldots,n_x$.
		\end{algorithmic}
	\end{algorithm}

	\subsection{sGaBP-m for DC state estimation}\label{AppendixC}
	The sGaBP-m algorithm serves as a robust baseline by introducing the regularization parameter $\lambda$ to enforce convergence~\cite{Kamper2020}.
	However, sGaBP-m operates on a pairwise Markov graph rather than a factor graph. Thus, to apply sGaBP-m, we first convert the factor graph model in \eqref{Eq:JPDF_PS} into its canonical Gaussian form defined by a global precision matrix $\mathbf{S} \in \mathbb{R}^{n_x \times n_x}$ and a potential vector $\bm{b} \in \mathbb{R}^{n_x}$, which are derived as,
	\begin{equation}\label{Eq:CanonicalForm}
		\begin{split}
			\mathbf{S} &= \bm{\Lambda}_{\text{prior}} + \mathbf{A}^T \bm{\Lambda}_z \mathbf{A}, \\
			\bm{b} &= \bm{\Lambda}_{\text{prior}} \bm{\mu}_{\text{prior}} + \mathbf{A}^T \bm{\Lambda}_z \bm{z},
		\end{split}
	\end{equation}
	where $\bm{\Lambda}_{\text{prior}} = \text{diag}(\Lambda_{x_1}, \dots, \Lambda_{x_{n_x}})$, $\bm{\mu}_{\text{prior}} = [\mu_{1}, \dots, \allowbreak \mu_{n_x}]^T$, and $\bm{\Lambda}_z = \text{diag}(\Lambda_{z_1}, \dots, \Lambda_{z_{n_z}})$.
	Once $\mathbf{S}$ and $\bm{b}$ are obtained, a weighted 
	undirected graph is constructed over the variable nodes, 
	with the edge weight between nodes $i$ and $j$ being defined as
	\begin{equation}\label{Eq:MetisWeight}
		w_{ij} = \frac{|S_{ij}|}{\sqrt{S_{ii}\,S_{jj}}}, 
		\quad i \neq j,
	\end{equation}
	where $S_{ij}$ denotes the $(i,j)$-th element of 
	$\mathbf{S}$. An edge is included only when 
	$S_{ij} \neq 0$, i.e., when variables $x_i$ and 
	$x_j$ are coupled in the pairwise Markov graph.
	The variable set $\mathcal{X}$ is then partitioned 
	into $p = \lceil \alpha\, n_x \rceil$ disjoint 
	clusters using the METIS balanced graph partitioning 
	algorithm~\cite{Karypis1998METIS}, which produces 
	balanced partitions with minimal edge cut with 
	respect to the weights in~\eqref{Eq:MetisWeight}.

	Let $\mathcal{C}_c \subset \mathcal{X}$ denote the subset of variable nodes
	belonging to the $c$-th cluster, so that the clusters partition the
	variable set, $\bigcup_{c=1}^p \mathcal{C}_c = \mathcal{X}$.
	To facilitate the algorithmic description, we adopt a block notation
	indexed by clusters. 
	Let $\mathbf{S}_{cc'} = \mathbf{S}[\mathcal{C}_c,
	\mathcal{C}_{c'}]$ denote the block of $\mathbf{S}$ with rows indexed by
	$\mathcal{C}_c$ and columns by $\mathcal{C}_{c'}$ (the diagonal block
	$\mathbf{S}_{cc}$ being the within-cluster precision on $\mathcal{C}_c$),
	in contrast to the scalar entry $S_{ij}$ in~\eqref{Eq:MetisWeight}.
	Similarly, let $\bm{b}_c$ and $\bm{\mu}_c$ denote the sub-vectors of
	$\bm{b}$ and of the cluster mean on $\mathcal{C}_c$, respectively.
	Let $d_c = |\mathcal{C}_c|$ be the dimension of $\mathcal{C}_c$, and
	$\mathcal{A}_c$ the set of indices of clusters neighboring $\mathcal{C}_c$.
	Based on these definitions, the sGaBP-m algorithm is summarized as Algorithm~\ref{alg_sgabpm}.

	\begin{algorithm}
		\caption{The sGaBP-m algorithm for inferring $x_k$.}\label{alg_sgabpm}
		\begin{algorithmic}[1]
			\REQUIRE Precision matrix $\mathbf{S}$, potential vector $\bm{b}$, clusters $\{\mathcal{C}_c\}_{c=1}^p$, parameters $\lambda, \beta, \delta$;
			\ENSURE Marginal beliefs $b(x_k)$ for all variable nodes $k=1, 2, \ldots, n_x$;
			\STATE Let $t = 0$ and error $e = \infty$;
			\STATE Initialize $\bm{\mu}_c^{(-1)}=\mathbf{0}$, $\mathbf{Q}_{c'c}^{(0)}=\mathbf{0}$, and $\bm{v}_{c'c}^{(0)}=\mathbf{0}$ for all clusters $c=1,2,\ldots,p$ and all $c'\in\mathcal{A}_c$;
			\WHILE{$e > \delta$ \textbf{and} $t < \beta$}
			\FOR{$c=1$ to $p$ }
			\STATE $\mathbf{P}_{c}^{(t)} = \lambda \mathbf{I}_{d_{c}} + \mathbf{S}_{cc} + \sum_{c' \in \mathcal{A}_c} \mathbf{Q}_{c'c}^{(t)}$;
			\STATE $\bm{h}_{c}^{(t)} = \bm{b}_{c} + \sum_{c' \in \mathcal{A}_c} \bm{v}_{c'c}^{(t)}$;
			\STATE $\bm{\mu}_{c}^{(t)} = [\mathbf{P}_{c}^{(t)}]^{-1} [\lambda\bm{\mu}_{c}^{(t-1)} + \bm{h}_{c}^{(t)}]$;
			\ENDFOR
			\STATE $e \leftarrow \max_{c}\left\|\mathbf{S}_{cc}\bm{\mu}_{c}^{(t)}+\sum_{c'\in\mathcal{A}_c}\mathbf{S}_{cc'}\bm{\mu}_{c'}^{(t)}-\bm{b}_{c}\right\|_{\infty}$;
			\STATE \textbf{If} $e \le \delta$ \textbf{break};
			\FOR{$c=1$ to $p$ \textbf{and} each neighbor $c' \in \mathcal{A}_c$}
			\STATE $\mathbf{Q}_{cc'}^{(t+1)} = -\mathbf{S}_{c'c} [\mathbf{P}_{c}^{(t)} - \mathbf{Q}_{c'c}^{(t)}]^{-1} \mathbf{S}_{cc'}$;
			\STATE $\bm{v}_{cc'}^{(t+1)} = -\mathbf{S}_{c'c} [\mathbf{P}_{c}^{(t)} - \mathbf{Q}_{c'c}^{(t)}]^{-1} [ \lambda\bm{\mu}_{c}^{(t-1)} + \bm{h}_{c}^{(t)} - \bm{v}_{c'c}^{(t)} ]$;
			\ENDFOR
			\STATE $t \leftarrow t+1$;
			\ENDWHILE
			\STATE Compute $\hat{\mathbf{P}}_c^{(t)} = \mathbf{S}_{cc} + \sum_{c' \in \mathcal{A}_c} \mathbf{Q}_{c'c}^{(t)}$ for all $c$;
			\STATE \textbf{Return} $b(x_k) = \mathcal{N}\!\left( x_k \mid [\bm{\mu}_{c}^{(t)}]_{k}, [(\hat{\mathbf{P}}_{c}^{(t)})^{-1}]_{kk} \right)$, where $x_k \in \mathcal{C}_c$.
		\end{algorithmic}
	\end{algorithm}

	\subsection{BP for DC state estimation}\label{AppendixB}
	For comparative analysis, we implement the BP algorithm for the DC state estimation problem.  
	We adopt the BP message update rules in (7), (12), and (14) of~\cite{pulgar2018fast}.
	
	The message from a factor node $f_i$ to a variable node $x_i$ can be expressed as $m_{f_{i} \rightarrow x_i} = p(x_i) = \mathcal{N}(x_i \mid \mu_{x_i}, \Lambda_{x_i}^{-1}).$
	The message from a factor node $g_{j}$ to a variable node $x_i$ can be expressed as $m_{g_{j} \rightarrow x_i}\propto \mathcal{N}(x_i \mid \mu_{g_{j}}^{x_i}, {\Lambda_{g_{j}}^{x_i}}^{-1})$, where the mean $\mu_{g_{j}}^{x_i}$ and precision matrix $\Lambda_{g_{j}}^{x_i}$ can be obtained as,
	\begin{equation}\label{Eq:mess_g_a_TBP}
		\begin{split}
			\mu_{g_{j}}^{x_i} = & {\textstyle \frac{1}{a_{i,j}} \left (z_j - \sum_{x_u\in \mathcal{N}_{g_j}\setminus{x_i} } a_{u,j} \mu_{x_u}^{g_j} \right)}, \\
			{\Lambda_{g_{j}}^{x_i}}^{-1} = & {\textstyle\frac{1}{a_{i,j}^2} \left (\Lambda_{z_j}^{-1} + \sum_{x_u\in \mathcal{N}_{g_j}\setminus{x_i} } a_{u,j}^2 {\Lambda_{x_u}^{g_j}}^{-1} \right)},
		\end{split}
	\end{equation}
	where $\mathcal{N}_{g_j}$ denotes the set of variable nodes directly connected to the factor node $g_j$.
	The message from a variable node $x_i$ to a factor node $g_{j}$ is $m_{x_i \rightarrow g_{j}}\propto \mathcal{N}(x_i \mid \mu_{x_i}^{g_{j}}, {\Lambda_{x_i}^{g_{j}}}^{-1})$, where,
	\begin{equation}\label{Eq:mess_a_g_TBP}
		\begin{split}
			\mu_{x_i}^{g_{j}} = & {\Lambda_{x_i}^{g_{j}}}^{-1} {\textstyle \left (\sum_{g_u\in \mathcal{N}_{x_i} \setminus{g_j}}  {\Lambda_{g_{u}}^{x_i}}\mu_{g_{u}}^{x_i} + {\Lambda_{x_i}}\mu_{x_i} \right)},\\
			{\Lambda_{x_i}^{g_{j}}} = & {\textstyle \sum_{g_u\in \mathcal{N}_{x_i} \setminus{g_j}}  {\Lambda_{g_{u}}^{x_i}} + {\Lambda_{x_i}} },
		\end{split}
	\end{equation}
	where $\mathcal{N}_{x_i}$ denotes the set of factor nodes connected to variable node $x_i$.
	
	Upon convergence, the belief $b(x_i)$ is given by,
	\begin{equation}\label{Eq:belief_TBP}
		\begin{split}
			b(x_i) =& m_{f_i \rightarrow x_i} {\textstyle \prod_{\substack{g_u \in \mathcal{N}_{x_i}}} m_{g_{u} \rightarrow x_i} } \\
			=& \mathcal{N}(x_i \!\mid\! \mu_{x_i}, \Lambda_{x_i}^{-1}) {\textstyle \prod_{\substack{g_u \in \mathcal{N}_{x_i}}} \!\!\!\mathcal{N}(x_i \!\mid\! \mu_{g_{u}}^{x_i}, (\Lambda_{g_{u}}^{x_i})^{-1})} \\
			\propto& \mathcal{N}(x_i \!\mid\! \hat{\mu}_{i,\text{BP}}, {\hat{\Lambda}_{i,\text{BP}}}^{-1}).
		\end{split}
	\end{equation}
	Let $m_{g_{j} \rightarrow x_i}^t$ and $m_{x_i\rightarrow g_{j}}^t$ denote the messages passed from factor node $g_j$ to variable node $x_i$ and from variable node $x_i$ to factor node $g_j$ at the $t$-th iteration, respectively.
	The BP algorithm is summarized as Algorithm~\ref{alg_bp}.
	\begin{algorithm}
		\caption{The BP algorithm for inferring $x_i$.}\label{alg_bp}
		\begin{algorithmic}[1]
			\REQUIRE The factor graph $\mathbb{G}$, priors
			$\{\mu_{x_i},\Lambda_{x_i}^{-1}\}_{i=1}^{n_x}$,
			measurements
			$\{z_j,\Lambda_{z_j}^{-1}\}_{j=1}^{n_z}$,
			parameters $\beta,\delta$;
			\ENSURE Beliefs $b(x_i)$, $i=1, 2,\ldots, n_x$;
			\STATE  Let $t = 1$ and error $e = \infty$;
			\STATE For all connected nodes $x_i$ ($i=1, 2, \ldots, n_x$) and $g_j$ ($j=1, 2,\ldots ,n_z$), initialize messages $m_{x_i \rightarrow g_j}^0$ and $m_{g_j \rightarrow x_i}^0$ as Gaussian distributions;
			\WHILE{$e > \delta$ \textbf{and} $t \le \beta$}
			\FORALL{factor node $g_j$ ($j=1, 2,\ldots ,n_z$) \textbf{and} each neighbor $x_i \in \mathcal{N}_{g_j}$}
			\STATE Compute $m_{g_{j} \rightarrow x_i}^t$ via \eqref{Eq:mess_g_a_TBP} using incoming messages $\{ m_{x_k \rightarrow g_j}^{t-1} \}_{x_k \in \mathcal{N}_{g_j} \setminus \{x_i\}}$ from iteration $t-1$;
			\ENDFOR
			\FORALL{variable node $x_i$ ($i=1, 2,\ldots, n_x$) \textbf{and} each neighbor $g_j \in \mathcal{N}_{x_i}$}
			\STATE Compute $m_{x_i\rightarrow g_{j}}^t$ via \eqref{Eq:mess_a_g_TBP} using incoming messages $\{ m_{g_k \rightarrow x_i}^{t-1} \}_{g_k \in \mathcal{N}_{x_i} \setminus \{g_j\}}$ from iteration $t-1$;
			\ENDFOR
			\STATE $e \leftarrow \max_{i,j} \left\| m_{g_j \rightarrow x_i}^{(t)} - m_{g_j \rightarrow x_i}^{(t-1)} \right\|_{\infty}$;
			\STATE $t \leftarrow t + 1$;
			\ENDWHILE
			\STATE Calculate beliefs via \eqref{Eq:belief_TBP} for all $i \in \{1, 2,\ldots, n_x\}$; 
			\STATE \textbf{Return} $b(x_i)$.
		\end{algorithmic}
	\end{algorithm}

	\section{Implementation Details of Algorithms for the Ising Model}\label{sec:S2}

	The following three subsections adapt Algorithms~\ref{alg2},
	\ref{alg_GBP}, and \ref{alg_bp} to the Ising factor graph.
	Three aspects change uniformly across all the three algorithms.
	
	First, the inputs are changed as follows. 
	The Gaussian priors
	$\{\mu_{x_i},\Lambda_{x_i}^{-1}\}$ and measurements
	$\{z_j,\Lambda_{z_j}^{-1}\}$ are replaced by
	$\{h_i\}_{i=1}^{n_x}$ and $\{J_{ij}\}_{f_{ij}\in\mathcal{F}}$,
	with $x_i\in\{-1,+1\}$.
	The factor set consists of factors $f_i(x_i)=\exp(h_ix_i)$
	and pairwise factors $f_{ij}(x_i,x_j)=\exp(J_{ij}x_ix_j)$.
	
	Second, all messages are discrete distributions over $\{-1,+1\}$
	rather than Gaussian distributions parameterized by precision
	and potential; accordingly, the Gaussian message formulas and
	the Schur-complement marginalization are replaced by finite
	summations over $\{-1,+1\}$.
	
	The specific message formulas for each algorithm are given
	in Sections~\ref{AppendixE}--\ref{AppendixF} of Appendix~\ref{sec:S2}.

	\subsection{ENMP for Ising model}\label{AppendixE}
	
	The local potential function over the residual $r$-neighborhood
	$N_{j\setminus i}$ is
	\begin{equation}\label{Eq:Ising_Phi_res}
		\Phi_{N_{j\setminus i}}(\bm{x}_{N_{j\setminus i}})
		= \!\!\!\!\!\!\prod_{f_{kl}\in\bm{f}_{N_{j\setminus i}}}\!  \!\!\!\!\!\!\exp(J_{kl}x_kx_l)
		\cdot \!\!\!\!\!\!\!\prod_{f_k\in\bm{f}_{N_{j\setminus i}}} \!\!\!\!\!\!\!\exp(h_kx_k).
	\end{equation}
	The message in~\eqref{Eq:q_i_j_marginal_mess1_int} reduces to
	\begin{equation}\label{Eq:Ising_ENMP_msg}
		q_{i\leftarrow j}(x_j) = \!\!\!
		\sum_{\bm{x}_{N_{j\setminus i}\setminus j}}
		\!\!\!\Phi_{N_{j\setminus i}}(\bm{x}_{N_{j\setminus i}})
		\!\!\!\prod_{x_k\in\bm{x}_{N_{j\setminus i}\setminus j}}
		\!\!\!\!\!\!q_{j\leftarrow k}(x_k),
	\end{equation}
	and the belief in~\eqref{Eq:q_i_marginal_mess1_int} reduces to
	\begin{equation}\label{Eq:Ising_ENMP_bel}
		q_i(x_i) = \frac{1}{Z_i}
		\sum_{\bm{x}_{N_i\setminus i}}
		\Phi_{N_i}(\bm{x}_{N_i})
		\prod_{x_j\in\bm{x}_{N_i\setminus i}}
		q_{i\leftarrow j}(x_j),
	\end{equation}
	where $\Phi_{N_i}$ follows~\eqref{Eq:Ising_Phi_res} with
	the summation extended to $\bm{f}_{N_i}$.
	
	In Algorithm~\ref{alg2}, Step~9 uses~\eqref{Eq:Ising_ENMP_msg} and
	Step~15 uses~\eqref{Eq:Ising_ENMP_bel}; all other steps are unchanged.

	\subsection{GBP for Ising model}\label{AppendixG}
	Since the Ising factor graph is constructed on the
	IEEE 14-bus topology, the region
	graph configurations GBP-A and GBP-B defined in
	Appendix~\ref{AppendixD} are directly reused. Only the message
	representation changes, with Gaussian parameters being replaced by discrete
	distributions over $\{-1,+1\}$ and Schur-complement marginalization being replaced by finite summation.
	The local potential over factors in region $r$ but not in
	child $s$ is
	\begin{equation}\label{Eq:Ising_GBP_Phi}
		\Phi_{r\setminus s}(\bm{x}_r)
		= \!\!\!\prod_{f_{kl}\in\bm{f}_{r\setminus s}}\!\!\!\!\exp(J_{kl}x_kx_l)
		\cdot\!\!\!\prod_{f_k\in\bm{f}_{r\setminus s}}\!\!\!\!\exp(h_kx_k).
	\end{equation}
	For the Ising model, the corresponding GBP message is computed
    by discrete marginalization as
	\begin{equation}\label{Eq:Ising_GBP_msg}
		m_{rs}(\bm{x}_s)
		\propto
		\sum_{\bm{x}_{r\setminus s}}
		\Phi_{r\setminus s}(\bm{x}_r)
		\!\!\!\prod_{\substack{u\in\mathcal{R}\setminus\mathcal{D}_r^{+}\\
				v\in\mathcal{D}_r^{+}\setminus\mathcal{D}_s^{+},
				\;u\in\mathcal{P}_v}}
		\!\!\!m_{uv}(\bm{x}_v),
	\end{equation}
	where the sum runs over all $2^{|\bm{x}_{r\setminus s}|}$
	configurations of $\bm{x}_{r\setminus s}$.
	Upon convergence, the belief of region $r$ is
	\begin{equation}\label{Eq:Ising_GBP_bel}
		b_r(\bm{x}_r)
		\propto \Phi_r(\bm{x}_r)
		\!\!\!\prod_{\substack{v\in\mathcal{D}_r^{+}\\
				u\in\mathcal{P}_v\setminus\mathcal{D}_r^{+}}}
		\!\!\!m_{uv}(\bm{x}_v),
	\end{equation}
	where $\Phi_r(\bm{x}_r)=\prod_{f_{kl}\in\bm{f}_r}\exp(J_{kl}x_kx_l)
	\cdot\prod_{f_k\in\bm{f}_r}\exp(h_kx_k)$.
	The marginal of variable $x_i$ is obtained from the smallest region
	$r\ni x_i$ by
	\begin{equation}\label{Eq:Ising_GBP_marginal}
		b(x_i)\propto\sum_{\bm{x}_{r\setminus i}}b_r(\bm{x}_r).
	\end{equation}
	In Algorithm~\ref{alg_GBP}, Step~4 uses~\eqref{Eq:Ising_GBP_msg},
	Step~10 uses~\eqref{Eq:Ising_GBP_bel}, and Step~13
	uses~\eqref{Eq:Ising_GBP_marginal}; all other steps are unchanged.

	\subsection{BP for Ising model}\label{AppendixF}
	
	The message from factor node $f_i$ to variable node $x_i$ is
	\begin{equation}\label{Eq:Ising_BP_fi}
		m_{f_i\to x_i}(x_i) = \exp(h_ix_i).
	\end{equation}
	The message from pairwise factor node $f_{ij}$ to variable node
	$x_i$ is
	\begin{equation}\label{Eq:Ising_BP_fij}
		m_{f_{ij}\to x_i}(x_i)\propto \!\!\!\!\!\!
		\sum_{x_j\in\{-1,+1\}} \!\!\!\!\!\! \exp(J_{ij}x_ix_j)\,m_{x_j\to f_{ij}}(x_j).
	\end{equation}
	The message from variable node $x_i$ to pairwise factor node
	$f_{ij}$ is
	\begin{equation}\label{Eq:Ising_BP_xi}
		m_{x_i\to f_{ij}}(x_i) \propto
		\exp(h_ix_i)
		\!\!\!\!\!\!\!\prod_{f_{ik}\in\mathcal{N}_{x_i}\setminus\{f_{ij}\}} \!\!\!\!\!\!
		m_{f_{ik}\to x_i}(x_i),
	\end{equation}
	where $\mathcal{N}_{x_i}$ denotes the set of pairwise factor nodes
	connected to $x_i$.
	Upon convergence, the belief is
	\begin{equation}\label{Eq:Ising_belief_BP}
		b(x_i)\propto
		\exp(h_ix_i)
		\prod_{f_{ij}\in\mathcal{N}_{x_i}}m_{f_{ij}\to x_i}(x_i).
	\end{equation}
	In Algorithm~\ref{alg_bp}, Step~5 uses
	\eqref{Eq:Ising_BP_fij}, Step~8 uses~\eqref{Eq:Ising_BP_xi}, 
	and Step~13 uses~\eqref{Eq:Ising_belief_BP};
	all other steps are unchanged.

	\section{Full sGaBP-m Results over the Tested $(\lambda, \alpha)$ Configurations}
	\label{sec:sgabpm_full}

	This section reports the complete sGaBP-m results 
	over all tested configurations of $\lambda$ 
	and $\alpha$. In the main text, for each IEEE 
	test system, only the best-performing configuration 
	is shown as the horizontal baseline, selected by 
	minimizing the AKLD over the tested configurations. 
	The full results are provided here to illustrate 
	the sensitivity of sGaBP-m to these two parameters.

	The tested values of~$\alpha$ are shown on the 
	horizontal axes of 
	Figs.~\ref{fig:sgabpm_full_akld} 
	and~\ref{fig:sgabpm_full_rmse}, with different 
	curves corresponding to different values 
	of~$\lambda$. Under the tested configurations, 
	the AKLD, $\mathrm{RMSE}_{\mu}$, and 
	$\mathrm{RMSE}_{\sigma^2}$ of sGaBP-m vary 
	substantially across the $(\lambda,\alpha)$ 
	configurations. 
	Across all four test systems, lower AKLD values 
	are generally obtained only at the smallest 
	tested~$\alpha$, while larger values of~$\alpha$ 
	combined with many choices of~$\lambda$ lead to 
	AKLD values that reach the prescribed cap 
	of~$10^{10}$. 
	Moreover, when a configuration reaches the prescribed cap, 
	it does so across all $N_{\mathrm{MC}} = 50$ trials, 
	indicating systematic rather than occasional failure.
	For a given~$\alpha$, a larger~$\lambda$ can 
	alleviate non-convergence; however, when 
	sGaBP-m already converges, a larger~$\lambda$ 
	increases the AKLD. 
	The non-convergence observed under many 
	configurations is consistent with the weak 
	diagonal dominance of the precision matrix 
	in the tested IEEE systems, which requires 
	a sufficiently large~$\lambda$ for sGaBP-m 
	to converge. 
	The elevated AKLD values are primarily 
	attributable to inaccurate mean estimates, 
	as $\mathrm{RMSE}_{\mu}$ is generally larger 
	than $\mathrm{RMSE}_{\sigma^2}$ under the 
	tested configurations; moreover, a configuration 
	with low $\mathrm{RMSE}_{\mu}$ does not 
	necessarily yield low $\mathrm{RMSE}_{\sigma^2}$, 
	and many configurations reach the prescribed 
	RMSE cap of $10^{5}$, particularly on the 
	larger test systems.
	\begin{figure}[t]
		\centering
		\includegraphics[width=0.47\textwidth]{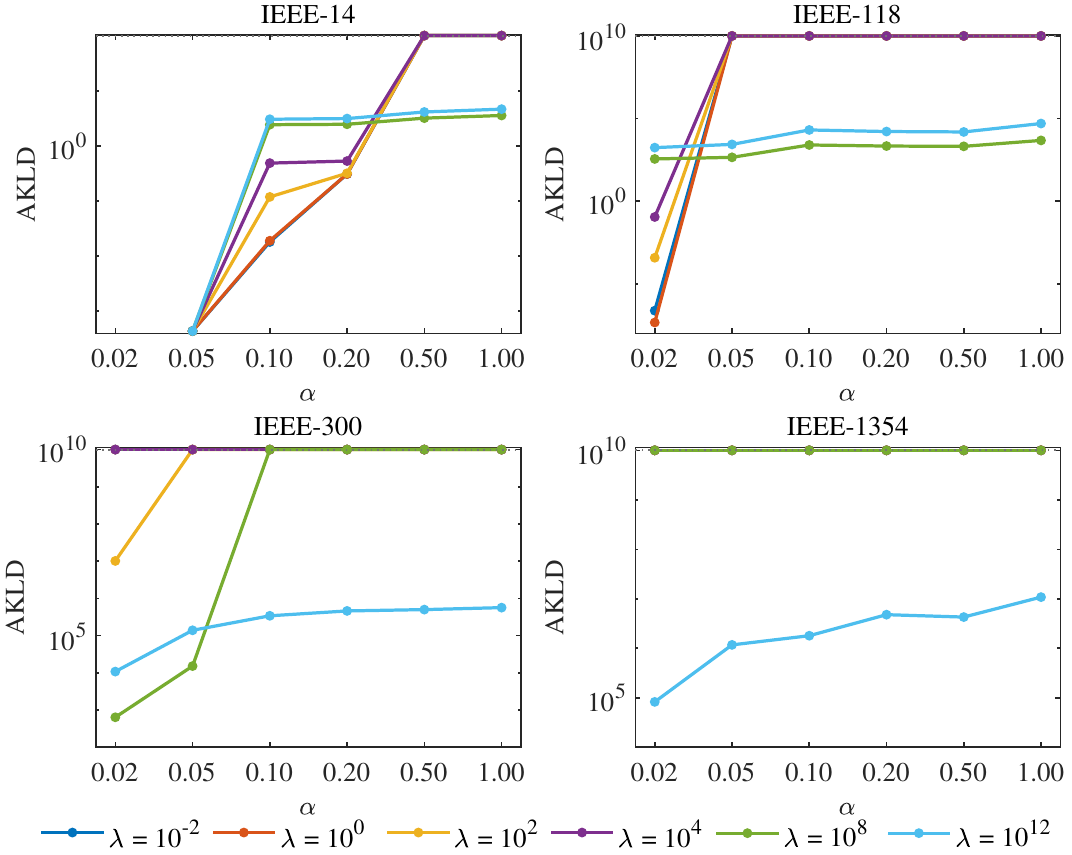}
		\caption{Full AKLD results of sGaBP-m over 
			the tested $(\lambda,\alpha)$ configurations 
			on the IEEE 14-, 118-, 300-, and 1354-bus 
			systems. The horizontal axis denotes~$\alpha$, 
			and different curves correspond to different 
			values of~$\lambda$. The dashed horizontal 
			line indicates the prescribed AKLD cap 
			of $10^{10}$.}
		\label{fig:sgabpm_full_akld}
	\end{figure}
	\begin{figure}[t]
		\centering
		\includegraphics[width=0.47\textwidth]{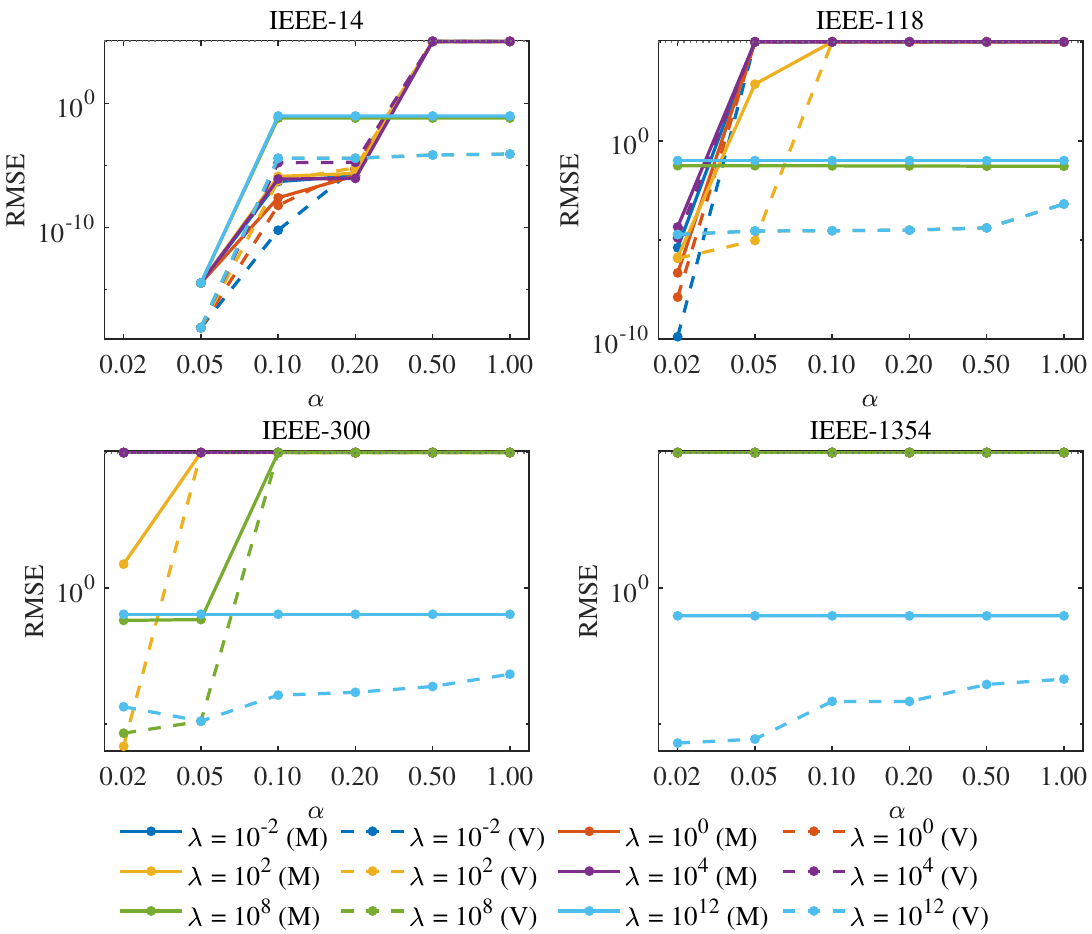}
		\caption{Full $\mathrm{RMSE}_{\mu}$ and 
			$\mathrm{RMSE}_{\sigma^2}$ results of 
			sGaBP-m over the tested $(\lambda,\alpha)$ 
			configurations on the IEEE 14-, 118-, 300-, 
			and 1354-bus systems. The labels ``M'' and 
			``V'' denote $\mathrm{RMSE}_{\mu}$ and 
			$\mathrm{RMSE}_{\sigma^2}$, respectively. 
			The dashed horizontal line indicates the 
			prescribed RMSE cap of $10^{5}$.}
		\label{fig:sgabpm_full_rmse}
	\end{figure}
	
	Overall, under the tested configurations, the 
	accuracy of sGaBP-m depends strongly on the 
	selected $(\lambda,\alpha)$ configuration. 
	On the IEEE 14-bus system, sGaBP-m can achieve 
	accurate estimates under carefully chosen 
	configurations, but its performance becomes 
	increasingly sensitive to $(\lambda,\alpha)$ 
	on larger systems, with many configurations 
	reaching the prescribed AKLD or RMSE caps. 
	The sGaBP-m baseline reported in the main text 
	therefore represents the best-performing 
	configuration over the tested configurations 
	rather than a consistently achievable level 
	of accuracy.

	\section{Maximum Sizes of Extended and Residual Neighborhoods}
	\label{sec:MaximumSN}
	
	Fig.~\ref{fig:size} shows the maximum sizes of the extended neighborhood
	$|\bm{x}_{N_{i}}|$ and the residual neighborhood $|\bm{x}_{N_{j\setminus i}}|$ as a
	function of~$r$ for each tested IEEE system. For all systems, both sizes increase
	with~$r$ and eventually saturate, since each neighborhood is bounded by the finite
	factor graph and stops growing once it has absorbed all reachable loop-induced
	nodes. The extended neighborhood sizes are consistently larger than the residual
	neighborhood sizes across all systems and values of~$r$. The larger IEEE systems
	(IEEE 300-bus and IEEE 1354-bus) exhibit substantially larger neighborhood sizes
	than the smaller systems (IEEE 14-bus and IEEE 118-bus).
	
	\begin{figure}[t]
		\centering\includegraphics[width=0.46\textwidth]{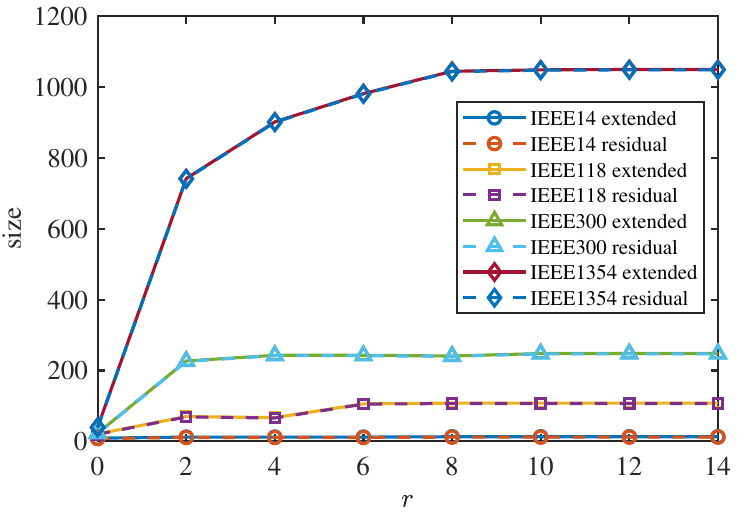}
		\caption{The maximum sizes of $|\bm{x}_{N_{i}}|$ (extended neighborhood) 
			and $|\bm{x}_{N_{j\setminus i}}|$ (residual neighborhood) as a function 
			of~$r$ for the IEEE 14-bus, IEEE 118-bus, IEEE 300-bus, 
			and IEEE 1354-bus systems.}
		\label{fig:size}
	\end{figure}

 	\bibliographystyle{unsrt}        
	\bibliography{autosam_TSP3}           

\end{document}